\documentclass[
reprint,
superscriptaddress,
amsmath,
amssymb,
aps,
pra
]{revtex4-2}
\usepackage{xurl}
\usepackage{graphicx}
\usepackage{tabularx}
\usepackage{etoolbox}

\newcolumntype{Y}{>{\centering\arraybackslash}X}

\begin{document}

\title{Certified randomness using a trapped-ion quantum processor}
\affiliation{Global Technology Applied Research, JPMorganChase, New York, NY 10017, USA}
\affiliation{Quantinuum, Broomfield, CO 80021, USA}
\author{Minzhao~Liu}
\thanks{Equal contribution}
\affiliation{Global Technology Applied Research, JPMorganChase, New York, NY 10017, USA}
\affiliation{Computational Science Division, Argonne National Laboratory, Lemont, IL 60439, USA}
\affiliation{Department of Physics, The University of Chicago, Chicago, IL 60637, USA}
\author{Ruslan~Shaydulin}
\thanks{Equal contribution}
\thanks{Corresponding author,\\email: \url{ruslan.shaydulin@jpmorgan.com}}
\affiliation{Global Technology Applied Research, JPMorganChase, New York, NY 10017, USA}
\author{Pradeep~Niroula}
\thanks{Equal contribution}
\affiliation{Global Technology Applied Research, JPMorganChase, New York, NY 10017, USA}
\author{Matthew~DeCross} 
\affiliation{Quantinuum, Broomfield, CO 80021, USA}
\author{Shih-Han~Hung}
\affiliation{Department of Computer Science, The University of Texas at Austin, Austin, TX 78712, USA}
\affiliation{Department of Electrical Engineering, National Taiwan University, Taipei City, 10617, ROC}
\author{Wen Yu Kon} 
\affiliation{Global Technology Applied Research, JPMorganChase, New York, NY 10017, USA}
\author{Enrique~Cervero-Mart\'{i}n} 
\affiliation{Global Technology Applied Research, JPMorganChase, New York, NY 10017, USA}
\author{Kaushik~Chakraborty} 
\affiliation{Global Technology Applied Research, JPMorganChase, New York, NY 10017, USA}
\author{Omar~Amer} 
\affiliation{Global Technology Applied Research, JPMorganChase, New York, NY 10017, USA}
\author{Scott~Aaronson} 
\affiliation{Department of Computer Science, The University of Texas at Austin, Austin, TX 78712, USA}
\author{Atithi~Acharya}
\affiliation{Global Technology Applied Research, JPMorganChase, New York, NY 10017, USA}
\author{Yuri~Alexeev}
\thanks{Current address: NVIDIA Corporation, Santa Clara, CA, USA}
\affiliation{Computational Science Division, Argonne National Laboratory, Lemont, IL 60439, USA}
\author{K.~Jordan~Berg}
\affiliation{Quantinuum, Broomfield, CO 80021, USA}
\author{Shouvanik~Chakrabarti}
\affiliation{Global Technology Applied Research, JPMorganChase, New York, NY 10017, USA}
\author{Florian J. Curchod}
\affiliation{Quantinuum, Terrington House, 13–15 Hills Road, Cambridge CB2 1NL, United Kingdom}
\author{Joan~M.~Dreiling}
\affiliation{Quantinuum, Broomfield, CO 80021, USA}
\author{Neal~Erickson}
\affiliation{Quantinuum, Broomfield, CO 80021, USA}
\author{Cameron~Foltz}
\affiliation{Quantinuum, Broomfield, CO 80021, USA}
\author{Michael~Foss-Feig}
\affiliation{Quantinuum, Broomfield, CO 80021, USA}
\author{David~Hayes}
\affiliation{Quantinuum, Broomfield, CO 80021, USA}
\author{Travis~S.~Humble}
\affiliation{Quantum Science Center, Oak Ridge National Laboratory, Oak Ridge, TN 37831, USA}
\author{Niraj~Kumar}
\affiliation{Global Technology Applied Research, JPMorganChase, New York, NY 10017, USA}
\author{Jeffrey~Larson}
\affiliation{Mathematics and Computer Science Division, Argonne National Laboratory, Lemont, IL 60439, USA}
\author{Danylo~Lykov}
\thanks{Current address: NVIDIA Corporation, Santa Clara, CA, USA}
\affiliation{Global Technology Applied Research, JPMorganChase, New York, NY 10017, USA}
\affiliation{Computational Science Division, Argonne National Laboratory, Lemont, IL 60439, USA}
\author{Michael~Mills} 
\affiliation{Quantinuum, Broomfield, CO 80021, USA}
\author{Steven~A.~Moses} 
\affiliation{Quantinuum, Broomfield, CO 80021, USA}
\author{Brian~Neyenhuis} 
\affiliation{Quantinuum, Broomfield, CO 80021, USA}
\author{Shaltiel~Eloul} 
\affiliation{Global Technology Applied Research, JPMorganChase, New York, NY 10017, USA}
\author{Peter~Siegfried} 
\affiliation{Quantinuum, Broomfield, CO 80021, USA}
\author{James~Walker} 
\affiliation{Quantinuum, Broomfield, CO 80021, USA}
\author{Charles~Lim}
\thanks{Corresponding author, email: \url{charles.lim@jpmorgan.com}}
\affiliation{Global Technology Applied Research, JPMorganChase, New York, NY 10017, USA}
\author{Marco~Pistoia}
\thanks{Principal investigator}
\thanks{Corresponding author,\\email: \url{marco.pistoia@jpmorgan.com}}
\affiliation{Global Technology Applied Research, JPMorganChase, New York, NY 10017, USA}

\date{\today}

\begin{abstract} 
    While quantum computers have the potential to perform a wide range of practically important tasks beyond the capabilities of classical computers \cite{Alexeev2021, Herman2023}, realizing this potential remains a challenge. One such task is to use an untrusted remote device to generate random bits that can be certified to contain a certain amount of entropy \cite{aaronson2023certified}. Certified randomness has many applications \cite{applicationsCertRand} but is fundamentally impossible to achieve solely by classical computation. In this work, we demonstrate the generation of certifiably random bits using the 56-qubit Quantinuum H2-1 trapped-ion quantum computer accessed over the internet. Our protocol leverages the classical hardness of recent random circuit sampling demonstrations \cite{qntm_rcs,Arute2019}: a client generates quantum ``challenge'' circuits using a small randomness seed, sends them to an untrusted quantum server to execute, and verifies the server's results. We analyze the security of our protocol against a restricted class of realistic near-term adversaries.
    Using classical verification with measured combined sustained performance of $1.1\times10^{18}$ floating-point operations per second across multiple supercomputers, we certify $71,313$ bits of entropy under this restricted adversary and additional assumptions. Our results demonstrate a step towards the practical applicability of today's quantum computers.
\end{abstract}

\maketitle
\renewcommand{\theequation}{\arabic{equation}}

In recent years, numerous theoretical results have shown evidence that quantum computers have the potential to tackle a wide range of problems out of reach of classical techniques. Prime examples include factoring large integers~\cite{shor1994algorithms}, implicitly solving exponentially sized systems of linear equations \cite{harrow2009quantum}, optimizing intractable problems~\cite{Shaydulin2024}, learning certain functions~\cite{liu2021rigorous}, and simulating large quantum many-body systems~\cite{berry2007Efficient}. However, accounting for considerations such as quantum error correction overheads and gate speeds, the resource requirements of known quantum algorithms for these problems put them far outside the reach of near-term quantum devices, including many of hypothesized fault-tolerant architectures. Consequently, it remains unclear if the devices available in the near-term will benefit a practical application \cite{hoefler2023disentangling}.

Starting with one of the first ``quantum supremacy'' demonstrations~\cite{Arute2019}, multiple groups have  used random circuit sampling (RCS) as an example of a task that can be executed faster and with a lower energy cost on today's quantum computers compared to what is achievable classically~\cite{Wu2021,Zhu2022,morvan2023phase,qntm_rcs}. Yet, despite rapid experimental progress, a beyond-classical demonstration of a \emph{practically useful} task performed by gate-based quantum computers has thus far remained elusive.

\begin{figure*}[!ht]
    \centering
    \includegraphics[width=\textwidth]{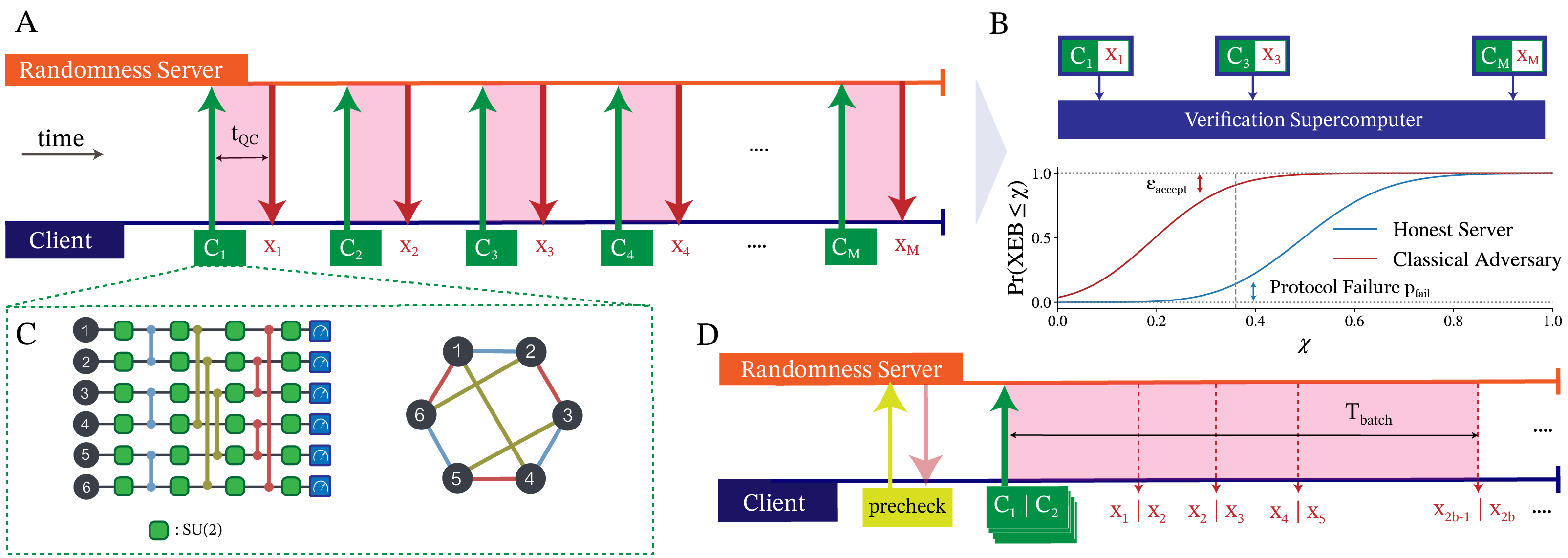}
    \caption{
    \textbf{Overview of the Protocol.} \textbf{a}, The idealized protocol. A client submits $M$ random circuits $\{C_i\}_{i\in[M]}$ serially to a randomness server and expects bitstrings $\{x_i\}_{i \in [M]}$ back, each within a time $t_{\rm QC}$. \textbf{b}, A subset of circuit-bitstring pairs is used to compute the {\rm XEB}{} score. The {\rm XEB}{} score has distributions (bottom plot for qualitative illustration only) corresponding to either an honest server or an adversarial server performing a low-fidelity classical simulation. For any {\rm XEB}{} target indicated by the dashed line, an honest server may fail to achieve a score above this threshold with probability $p_{\rm fail}$. \textbf{c}, Illustration of the challenge circuits, consisting of layers of $U_{\rm ZZ}$ gates sandwiched between layers of random SU(2) gates on all qubits. The arrangement of two-qubit gates is obtained via edge coloring (right) on a random $n$-node graph. \textbf{d}, Client-server interaction as implemented in our protocol. Following a device-readiness check (``precheck"), the client submits a batch of $2b$ circuits and expects all the samples corresponding to the batch to be returned within a cutoff duration $T_{b, \rm{cutoff}}$. Note that only one batc with execution time of $T_{\rm batch}$ is illustrated in the figure. The client continues the protocol until $M$ total circuits have been successfully executed. }
    \label{fig:main-overview-panel}
\end{figure*}

Random number generation is a natural task for such a demonstration since randomness is intrinsic to quantum mechanics, and it is important in a broad variety of applications, ranging from information security to ensuring fairness of processes such as jury selection~\cite{Acn2016,HerreroCollantes2017,Mannalatha2023,applicationsCertRand}. A central challenge for any client receiving randomness from a third-party provider, such as a hardware security module, is to verify that the bits received are truly random and freshly generated. While certified randomness is not necessary for every use of random numbers, the freshness requirement is especially important in applications such as lotteries and e-games, where multiple parties (which may or may not trust each other) need to ensure that a publicly distributed random number was generated on demand. We refer interested reader to a companion Perspective outlining applications that benefits in specific ways from certified randomness~\cite{applicationsCertRand}. Additionally, certified randomness can be used to verify the position of a dishonest party \cite{amer2024certified}.

Protocols exist for certifying random numbers based on the violation of Bell inequalities  \cite{acin2016certified, pironio2010random, Liu2018, Foreman2023practicalrandomness, bierhorst2018experimentally}. However, these protocols typically require the underlying Bell test to be loophole-free, which can be hard for the client to enforce when the quantum devices are controlled by a third-party provider. This approach thus necessitates that the client trust a third-party quantum device provider to perform the Bell test faithfully.

Alternatively, Ref.~\cite{aaronson2023certified} proposed a certified randomness protocol that combines RCS with ``verification" on classical supercomputers~\cite{aaronson2023certified, bassirian2021certified}. This type of protocol allows a classical client to verify randomness using only remote access to an untrusted quantum server. A classical client pseudorandomly generates $n$-qubit challenge circuits and sends them to a quantum server, which is asked to return length-$n$ bitstrings sampled from the output distribution of these circuits within a short amount of time (Fig. 1a, c). The circuits are chosen such that no realistic adversarial server can classically simulate them within the short response time. A small subset of circuits are then used to compute the cross-entropy benchmarking ({\rm XEB}) score~\cite{Boixo2018} (Fig. 1b), which reflects how well the samples returned by the server match the ideal output distributions of the submitted circuits. Crucially, extensive complexity-theoretic evidence suggests that {\rm XEB}{} is hard to ``spoof'' classically~\cite{aaronson2017complexity, Aaronson2020}. Therefore, a high {\rm XEB}{} score, combined with a short response time, allows the client to certify that the server must have used a quantum computer to generate its responses, thereby guaranteeing a certain amount of entropy with high probability. Our analysis quantifies the minimum amount of entropy that an untrusted server, possibly acting as an adversary, must provide to achieve a given {\rm XEB}{} score in a short amount of time.

The protocol proposed in Ref.~\cite{aaronson2023certified} provides a complexity-theoretic guarantee of $\Omega(n)$ bits of entropy for a server returning \textit{many samples from the same circuit}. This protocol is best suited for quantum computing architectures with overheads that make it preferable to sample a circuit many times after loading it once. In practice, the classical simulation cost of sampling a circuit many times is comparable to the cost of sampling only once~\cite{Liu2024}. Furthermore, the trapped-ion based quantum computer used in this work is configured to feature minimal overhead per circuit, such that executing many single-shot circuits does not introduce a substantial time penalty per circuit compared with sampling one circuit many times. Together, these two observations motivate strengthening the security of the protocol by requesting the server to return only one sample per circuit. To this end, in Supplemental Material (SM) Sec. I we extend the complexity-theoretic analysis to this modified setting of one sample per circuit, guaranteeing $\Omega(n)$ bits of entropy.

In this work we report an experimental demonstration of an RCS-based certified randomness protocol. Our main contributions are as follows. First, inspired by Ref.~\cite{aaronson2023certified}, we propose a modified RCS-based certified randomness protocol that is tailored to near-term quantum servers. Second, we prove the security of our implementation against a class of realistic finite-sized adversaries. Third, we use a high-fidelity quantum computer and exascale classical computation to experimentally realize this proposed protocol, pushing the boundaries of both quantum and classical computing capabilities. By combining the high-fidelity Quantinuum H2-1 quantum processor with exascale verification, we demonstrate a useful beyond-classical application of gate-based digital quantum computers.

In our proposed protocol, illustrated in Fig. 1d and detailed in the Methods section, the client pseudorandomly generates a sufficiently large number of $n$-qubit quantum circuits and then sends them in batches of $2b$ circuits, where $b$ is an integer. After a batch is submitted, the client waits for $2b$ length-$n$ bitstrings to be returned within $T_{b, \rm{cutoff}}$ seconds. The batch cutoff time prevents the protocol from stalling and is fixed in advance based on preliminary experiments to a value intended to maximize the amount of certifiable entropy while ensuring that the average response time per circuit remains low enough to preclude classical simulation as a viable strategy for the server to generate responses. If a batch times out or if a failure status is reported, all of the outstanding jobs in the batch are canceled, and \textit{all} bitstrings received from the batch are discarded. Consequently, results from a failed batch are not included in calculating the {\rm XEB}{} score or entropy extraction. Batches are continually submitted until $M$ valid samples are collected. The cumulative response time for successful batches gives the total time $T_{\rm tot}$ and the average time per sample $t_{\rm QC}=T_{\rm tot}/M$. Subsequently, the client calculates the {\rm XEB}{} score on a subset of size $m$ randomly sampled from the $M$ circuit-sample pairs:
\begin{equation}
    {\rm XEB}_{\rm test}= \frac{2^n}{m}\sum_{i \in \mathcal{V}} p_{C_i}(x_i) - 1
    \label{eq:xeb-definition},
\end{equation}
where $\mathcal{V}$ is the set of indices for the random subset of size $m$ and $p_{C}(x) = |\langle x | C|0\rangle|^2$ is the probability of measuring bitstring $x$ from an ideal quantum computer executing circuit $C$. If the bitstrings $x_i$ are perfectly drawn from the output distributions of sufficiently deep random circuits $C_i$, the {\rm XEB}{} score is expected to concentrate around 1. On the other hand, if the $x_i$ are drawn from distributions uncorrelated with the distributions induced by $C_i$, the {\rm XEB}{} score is expected to concentrate around $0$. The client decides to accept the received samples as random bits based on two criteria. First, the average time per sample must be lower than a threshold $t_{\rm{threshold}}$, which is chosen to preclude high-fidelity classical simulation. This time can be lower than $T_{b, \rm{cutoff}}$ since it is advantageous from the perspective of extractable entropy to accept some samples with response time slightly larger than $t_{\rm{threshold}}$ as long as the average response time remains low. Second, the {\rm XEB}{} score on $\mathcal{V}$ must be greater than a threshold $\chi \in [0, 1]$. All of $t_{\rm{threshold}}$, $\chi$, and $T_{b, \rm{cutoff}}$ are determined in advance of protocol execution, based on (for example) preliminary hardware experiments, with the goal of certifying a certain fixed amount of entropy at the end of the protocol with high probability. Together, the protocol succeeds if
\begin{equation}
   t_{\rm QC}=T_{\rm tot}/M \leq t_{\rm{threshold}}, \quad {\text{ and }} \:\: {\rm XEB}_{\rm test}\geq \chi
    \label{eq:protocol-success-criteria},
\end{equation}
and otherwise aborts.

\begin{figure*}[!ht]
    \centering
    \includegraphics[width=\textwidth]{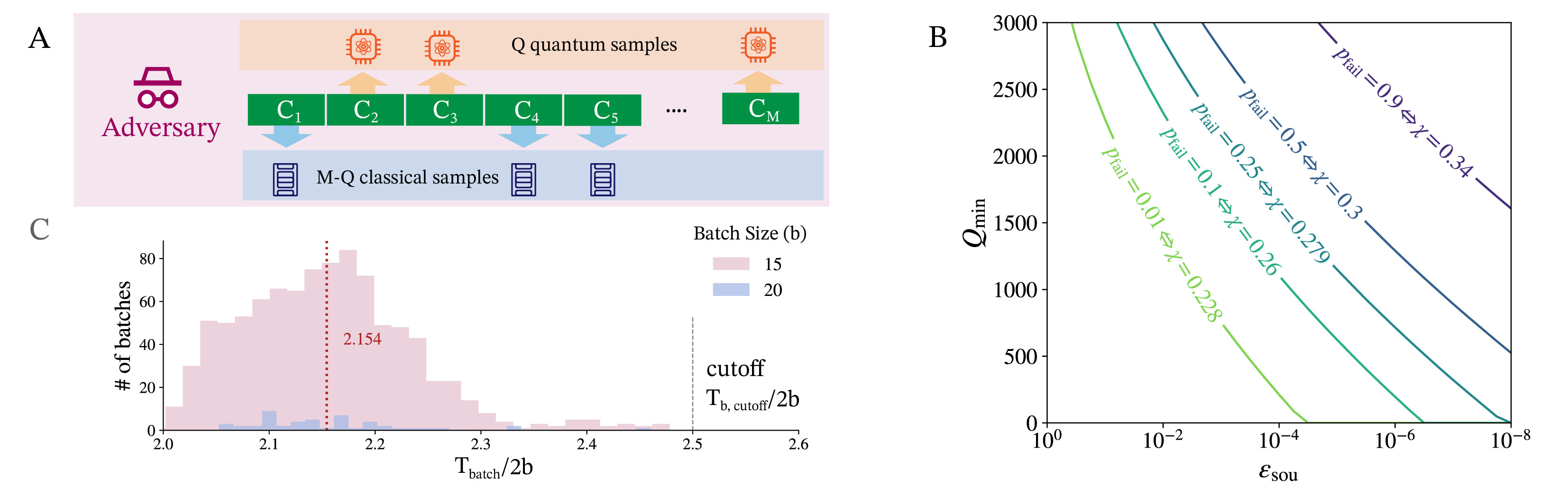}
    \caption{\textbf{Adversary model and protocol security.} \textbf{a}, In the adversarial model considered in this work, $Q$ samples are obtained using a perfect-fidelity quantum computer and $M-Q$ using classical simulation. \textbf{b}, Probability of an honest server with fidelity $\phi=0.3$ failing to certify $Q_{\rm min}$ quantum samples (and corresponding threshold $\chi$) with soundness $\varepsilon_{\rm sou}$ against an adversary four times more powerful than Frontier over repeated experiments, with the protocol parameters set to those from Table I. \textbf{c}, Distribution of batch times per successful sample, from a total of 984 successful batches, in our experiment. Vertical dashed line indicates the average time per sample.
    }
    \label{fig:panel-2}
\end{figure*}

The security of our protocol relies on the central assumption that, for the family of pseudorandom circuits we consider, there exists no practical classical algorithm that can spoof the {\rm XEB}{} test used in the protocol. We analyze the protocol security by modeling a restricted but realistic adversarial server that we believe to be the most relevant: for each circuit received, the adversary either samples an output honestly from a quantum computer or performs classical simulation (see Fig. 2a). Since only the former contains entropy, the adversary tries to achieve the threshold {\rm XEB}{} score with the fewest quantum samples, in order to pass the {\rm XEB}{} test while returning as little entropy as possible. For our protocol we assume an adversary with a perfect-fidelity quantum computer, which allows the adversary to spoof the maximum number of bitstrings classically. We further assume that the adversary's classical computational power is bounded by a fixed number of floating-point operations per second (FLOPS) $\mathcal{A}$, which may be measured relative to the most powerful supercomputer in the world (at the time of experiment, the Frontier supercomputer \cite{top500}), and that the adversary possesses the same optimized methods to simulate the circuits as the client has. Note that an adversary possessing more powerful classical methods for simulating circuits than expected can equivalently be modeled as an adversary with identical classical methods and larger computational power. We note that since the adversaries we analyze are only allowed a restricted set of strategies, the subsequent mathematical results hold only in this limited setting, conditioned on some additional assumptions further detailed in SM Sec. III C. To the best of our knowledge, the restricted set of classical and quantum adversary strategies considered here correspond to the current state of the art. We leave the incorporation of a broader class of adversaries to future analysis.

Crucially, the client needs to ensure that the circuits are difficult to simulate within the time $t_{\rm{threshold}}$. Otherwise, the server can use its classical supercomputer to deterministically simulate the circuits with high fidelity and generate samples that readily pass the tests in equation 2. For the family and size of circuits we consider, tensor network contraction is the most performant known method for both finite-fidelity and exact simulation \cite{qntm_rcs} as well as sampling. If a circuit has a verification (exact simulation) cost of $\mathcal{B}$ FLOPs, the adversary can simulate each circuit to a target fidelity of $\mathcal{A}\cdot t_{\rm{threshold}}/\mathcal{B}$ using partial contraction of tensor networks, for which the simulation cost and simulation fidelity are related linearly~\cite{markov2018quantum}. 
The protocol is successful only if the parameters are chosen such that the fidelity $\phi$ of an honest server satisfies
\begin{equation}
\phi \gg \mathcal{A}\cdot t_{\rm{threshold}}/\mathcal{B}.
\end{equation}
This condition requires that there exist a gap between the fidelity of an honest server and that achievable by an adversary performing mostly classical simulations. If this condition is satisfied, the {\rm XEB}{} score of an honest server will have a probability distribution with a higher average value than the probability distribution of the adversary's {\rm XEB}{} (qualitatively illustrated in Fig. 1b), allowing the client to distinguish between the two.

After certification (i.e., if the tests in equation 2 pass) the client uses a randomness extractor to process the $M$ samples. An ideal protocol for certified randomness either aborts, resulting in an ``abort state,'' or succeeds, resulting in a uniformly distributed bitstring that is uncorrelated with any side information.
Viewing the protocol as a channel acting on some initial state composed of both server and the client, an end-to-end protocol is said to be $\varepsilon_{\rm sou}$-sound if for any initial state, the end result is $\varepsilon_{\rm sou}$-close (in terms of trace distance) to the ideal output: a mixture of the ``abort state'' and the maximally mixed state (see SM Sec. III A for the rigorous definition of \textit{soundness}). 

The entropy that the client can extract out of the received samples upon successful execution of the protocol depends on how stringent its thresholds on the response time ($t_{\rm{threshold}}$) and the ${\rm XEB}$ score ($\chi$) are. It is in the client's interest to set these thresholds as stringently as possible, to force the hypothetical adversary to draw more samples from the quantum computer, while still allowing that an honest server can succeed with high probability. Since the thresholds are known to both parties, the adversary's strategy is to minimize the use of the quantum computer while ensuring that the protocol does not abort. Based on the protocol thresholds, the client can determine the number of quantum samples $Q_{\rm min}$ such that the protocol aborts with a large probability $1-\varepsilon_{\rm accept}$ if the adversary returns fewer than $Q_{\rm min}$ samples from the quantum computer (see SM Sec. IV F for details). Such a lower bound on $Q_{\rm min}$ can be used to derive the minimum smooth min-entropy of the received samples. Note that the smooth min-entropy of an information source characterizes the number of random bits that can be extracted from the source. In particular, we devise an $\varepsilon_{\rm sou}$-sound protocol that provides a lower bound on the smooth min-entropy $H^{\varepsilon_{\rm s}}_{\rm min}$ (defined in SM Sec. III D) with smoothness parameter $\varepsilon_{\rm s}=\varepsilon_{\rm sou}/4$ and with $\varepsilon_{\rm accept}=\varepsilon_{\rm sou}$. The results in the paper are reported in terms of the soundness parameter $\varepsilon_{\rm sou}$ and the smooth min-entropy $H^{\varepsilon_{\rm s}}_{\rm min}$.

A smaller $\varepsilon_{\rm sou}$ makes a stronger security guarantee by making it more difficult for an adversary to pass the {\rm XEB}{} test with a small $Q_{\rm min}$. This may be achieved by choosing a higher threshold $\chi$. However, a higher threshold also makes it more likely for an honest server to fail the {\rm XEB}{} test, meaning that the honest server cannot be certified to have produced the target amount of extractable entropy. Note that this does not necessarily mean that the samples provided by the honest server do not contain entropy, only that they fail to satisfy the criteria of equation 2 and consequently the protocol aborts. In practice, it is desirable to ensure that an honest server fails only with a low failure probability $p_{\rm fail}$. To that end, we may compute a threshold $\chi(p_{\rm fail})$ corresponding to any acceptable $p_{\rm fail}$. This threshold, along with $t_{\rm{threshold}}$, then allows us to determine $Q_{\rm min}$ for a target soundness $\varepsilon_{\rm sou}$ (see SM Sec. III D). Fig. 2b illustrates the achievable $Q_{\rm min}$ at different $p_{\rm fail}$ and $\varepsilon_{\rm sou}$, showing the trade-off between the three quantities at fixed experimental configuration and adversary's classical computational power ($\phi,t_{\rm QC},M,m,\mathcal{B}$, and $\mathcal{A}$).

\begin{table*}[ht]
\begin{ruledtabular}
\begin{tabular}{clc}
Label & Meaning                  & Value \\[.1em]  \hline \\[-.6em]
$n$ & Number of qubits & $56$  \\[.1em]
$\mathcal{B}$ & Cost of simulating challenge circuits & $90 \times 10^{18}$ FLOPs  \\[.1em]
$\mathcal{A}$ & Sustained peak performance of the Frontier supercomputer & $0.897 \times 10^{18}$ FLOPS  \\[.1em]
\mbox{}  & Time to simulate challenge circuits on the Frontier supercomputer & $100.3$ s \\[.1em]
$\chi$ & Threshold for XEB test & 0.3 \\ [.1em]
$t_{\rm{threshold}}$ & Threshold for average time per sample & 2.2 s \\
[.1em]
$T_{b, \rm{cutoff}}$ & Cutoff time for the server to respond to a batch of $2b$ circuits  & $2.5 \times 2b$ s \\
[1em] 
 $M$ & Number of successful samples & $30,010$  \\[.1em]
$t_{\rm QC}$ & Average response time per successful quantum sample & $2.154$ s  \\[.1em]
$m$ & Number of samples used to measure {\rm XEB}{}  &  $1,522$  \\[.1em]
${\rm XEB}_{\rm test}$ & Measured {\rm XEB}{}  & $0.32$  \\[.1em]
\end{tabular}
\end{ruledtabular}
\caption{Summary of experimental parameters.} \label{tab:summary-of-parameters}
\end{table*}

We demonstrate our protocol using the Quantinuum H2-1 trapped-ion quantum processor accessed remotely over the internet. The experimental parameters are summarized in Table~\ref{tab:summary-of-parameters}. The challenge circuits (illustrated in Fig. 1c, see SM Sec. IV C for the considerations involved in choosing the circuits) have a fixed arrangement of 10 layers of entangling $U_{\rm ZZ}$ gates, each sandwiched between layers of pseudorandomly generated $SU(2)$ gates on all qubits. The arrangement of two-qubit gates is obtained via edge coloring on a random $n$-node graph. Preliminary mirror-benchmarking experiments, along with gate-counting arguments based on the measured fidelities of component operations, allow us to estimate the fidelity of an honest server \cite{qntm_rcs}. At the time of the experiment, the H2-1 quantum processor was expected to attain a fidelity of $\phi \gtrsim 0.3$ or better on depth-10 circuits (multiple improvements were made to the H2-1 device after collection of this experiment's data that slightly increased the fidelity estimate in Ref.~\cite{qntm_rcs}). Likewise, the same preliminary experiments also let us anticipate average time per sample to be approximately $2.1 \text{ s}$, with a long-tailed timing distribution out to just below $2.5 \text{ s}$, as also seen in the full experiment in Fig. 2c. Reasonable ($p_{\rm fail} = 50\%$) protocol success rates can therefore be achieved with thresholds $t_{\rm{threshold}} = 2.2$ s and $\chi = 0.3$. For illustrative purposes, we describe the experiment based on these choices (in practice, one might want to lower $p_{\rm fail}$ by setting $\chi$ somewhat below the expected value). The batch cutoff time is set to be $T_{b, \rm{cutoff}} = (2b) \cdot 2.5$ seconds, anticipating that the relatively small expected fraction of batches taking average time per sample between $t_{\rm{threshold}} = 2.2 \text{ s}$ and $2.5 \text{ s}$ would contribute additional entropy to the received samples while being unlikely to increase the average time per sample from the expected $2.1 \text{ s}$ past the threshold of $2.2 \text{ s}$.

\begin{table}[h]
    \centering
\begin{tabularx}{\columnwidth}{ |m{0.8 in}| *{5}{Y|} }
\cline{2-6}
   \multicolumn{1}{c|}{} 
 & \multicolumn{5}{m{2 in}|}{\centering $\mathcal{A}$ (multiples of Frontier)}\\
\hline
 \centering $\varepsilon_{\rm sou}$ & 1 & 2 & 4 & 6 & 8 \\
         \hline
     \centering $10^{-2}$ & 0.19 & 0.16 & 0.11 & 0.06 & 0.01\\\hline
     \centering $10^{-4}$ & 0.15 & 0.12 & 0.07 & 0.02 & 0.00\\\hline
     \centering $10^{-6}$ & 0.12 & 0.09 & \textbf{0.04} & 0.00 & 0.00\\\hline
     \centering $10^{-8}$ & 0.10 & 0.07 & 0.02 & 0.00 & 0.00\\\hline
     \centering $10^{-10}$ & 0.08 & 0.05 & 0.00 & 0.00 & 0.00\\
\hline
\end{tabularx}
    \caption{Smooth min-entropy rate at varying $\varepsilon_{\rm sou}$ and $\mathcal{A}$. The adversary is assumed to have the same efficiency for classical simulation as client verification. The ratio corresponding to entropy we report in the main text is boldfaced.}
    \label{tab:ratio}
\end{table}

The circuit family considered has a simulation cost of $\mathcal{B} = 90\times 10^{18}$ FLOPs on the Department of Energy's Frontier supercomputer~\cite{frontier}, the most powerful supercomputer in the world at the time of writing~\cite{top500}. Following a detailed estimate of runtime on Frontier, we determine an exact simulation time of $100.3$ seconds per circuit when utilizing the entire supercomputer at a ``numerical efficiency" of $45\%$, where numerical efficiency is the ratio between the actual algorithm runtime and its theoretical expectation. See SM Sec. IV A for details on the circuit simulation cost.

In our experiment we use two batch sizes, $b = 15$ and $b = 20$;  most of the batches have $b=15$. In total, we submitted $1,993$ batches for a total of $60,952$ circuits. From those, we obtain a total of $M=30,010$ valid samples out of  $984$ successful batches. The cumulative device time of the successful samples was $64,652$ seconds, giving an average time of $t_{\rm QC} = 2.154$ seconds per sample, inclusive of all overheads such as communication time. Fig. 2c shows the distribution of $t_{\rm QC}$ per successful sample.

\begin{figure*}[!ht]
    \centering
     \includegraphics[width=\textwidth]{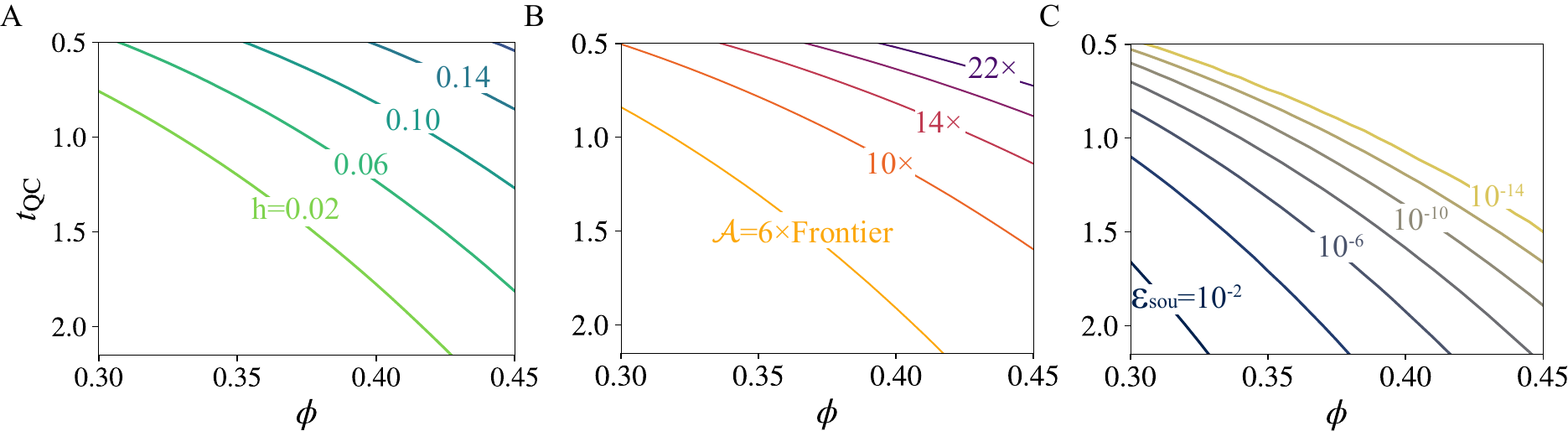}
    \caption{\textbf{Future improvements.} Improvement in metrics as fidelity $\phi$ and time per sample $t_{\rm QC}$ improve. All panels assume the same verification budget as this experiment, classical simulation numerical efficiency of $50\%$ for both verification and spoofing, and target failure probability $p_{\rm fail}=10^{-4}$. \textbf{a}, Smooth min-entropy rate, $h = H^{\varepsilon_{\rm s}}_{\rm min}/(M \cdot n)$, against an adversary four times as powerful as Frontier with $\varepsilon_{\rm sou} = 10^{-6}$ and $\varepsilon_{\rm s}=\varepsilon_{\rm sou}/4$. \textbf{b}, Adversarial power that still allows $h=0.01$ to be guaranteed with $\varepsilon_{\rm sou}=10^{-6}$. \textbf{c}, Soundness parameter $\varepsilon_{\rm sou}$ that still allows $h=0.01$ to be guaranteed with an adversary that is four times as powerful as Frontier.}
    \label{fig:future}
\end{figure*}

In this work the client's classical computational budget is spread across the Frontier \cite{frontier}, Summit \cite{summit}, Perlmutter \cite{perlmutter}, and Polaris \cite{polaris} supercomputers equipped with graphics processing units (GPUs), which are especially suitable for quantum circuit simulations. Of the four supercomputers, Frontier and Summit were utilized at full-machine scale during verification. We measure the sustained peak performance of 897 petaFLOPS and 228 petaFLOPS respectively (corresponding to numerical efficiencies of $45\%$ and $59\%$), achieving a combined performance of 1.1 exaFLOPS (see SM Sec. IV E). We compute the {\rm XEB}{} score for $m=1,522$ circuit-sample pairs, obtaining  ${\rm XEB}_{\rm test}=0.32$. The complete set of experimental parameters is listed in Table I.

The measured fidelity of ${\rm XEB}_{\rm test} = 0.32$ and measured time per sample $t_{\rm QC} = 2.154$ s pass the protocol specified by $\chi=0.3$ and $t_{\rm{threshold}}=2.2$ s. For a choice of soundness parameter $\varepsilon_{\rm sou}$ and a smoothness parameter $\varepsilon_{\rm s}=\varepsilon_{\rm sou}/4$, the protocol thresholds determine the number of quantum samples $Q$ and the smooth min-entropy $H^{\varepsilon_{\rm s}}_{\rm min}$ guaranteed by the success of this protocol against an adversary with classical resources bounded by $\mathcal{A}$. In Table II, we report the smooth min-entropy rate, $H^{\varepsilon_{\rm s}}_{\rm min}/(56\cdot M)$, for a range of $\mathcal{A}$ and $\varepsilon_{\rm sou}$ (see SM Sec. IV F for details of this calculation). This is to show that if we want to increase the security of the protocol either by increasing the assumed adversary's classical computational power or by reducing the soundness parameter, the amount of entropy that we can obtain must reduce.
In particular, we highlight that at $\varepsilon_{\rm sou}=10^{-6}$, we have $Q_{\rm min}=1,297$, corresponding to $H^{\varepsilon_{\rm s}}_{\rm min}=71,313$ against an adversary four times more powerful than Frontier (under the assumptions discussed earlier).

We feed the $56 \times 30,010$ raw bits into a Toeplitz randomness extractor~\cite{foreman2024cryptomite}
and extract $71,273$ bits 
(see SM Sec. IV F for details on extraction and the determination of extractable entropy). We note that the Toeplitz extractor is a ``strong'' seeded extractor for which the output is independent of the seed. For private use of the randomness, in which the extracted bits are not revealed, the extractor seed can be reused. We append the seed used in the extractor to the protocol output and do not count the seed as randomness ``consumed'' by our protocol.
The total input randomness used to seed the pseudorandom generator is thereby only 32 bits, and our protocol achieves certified randomness expansion. We further note that other extractors can be used that may consume less seed but have different security guarantees.

Future experiments are expected to improve device fidelity (higher $\phi$) and execution speed (lower $t_{\rm QC}$). Adjusting protocol thresholds ($\chi$ and $t_{\rm{threshold}}$) against improved device specifications stands to improve our protocol in terms of the achievable entropy, the adversarial computational power that can be guarded against, and the soundness parameter. Fig. 3 shows these metrics as we improve $t_{\rm QC}$ and $\phi$ (see SM Sec. V for details of this calculation). Conversely, for a fixed adversary and soundness parameter, any improvement in $t_{\rm QC}$ and $\phi$ reduces the verification budget required to certify a target number of quantum samples $Q$, making our protocol more cost-effective. Any improvement in entropy, all else being equal, translates into a higher throughput in the sense of a higher rate of entropy generation per second. With $\chi=0.3$ and $t_{\rm{threshold}}=2.2$ s, our experiment has a bitrate of $71,273 / (
30,010 \times 2.2~\text{s})\approx1$ bit per second at $\varepsilon_{\rm sou}=10^{-6}$. For $\varepsilon_{\rm sou}=10^{-6}$ and $p_{\rm fail}=0.1$, improving fidelity to $\phi=0.67$ and response time to $t_{\rm QC}=0.55$ seconds would let us achieve the bitrate of the NIST Public Randomness beacon \cite{kelsey2019a} (512 bits per minute). We note that improvement in $t_{\rm QC}$ can come from higher clock rates as well as parallelization over multiple quantum processors or over many qubits of one large quantum processor.     

The security of our protocol relies on the circuits being difficult to simulate. When better exact simulation techniques are developed by researchers in the future, both the adversary and the client can use the improved techniques to spoof and verify: such symmetric gains neutralize each other. While a significant improvement in approximate simulation techniques may benefit spoofing asymmetrically, the client might be able to neutralize those gains by modifying the ensemble of challenge circuits to make approximate simulations more difficult.

In summary, this work implements a protocol for certified randomness, which also lends itself to multiparty and public verification~\cite{applicationsCertRand}. We note that the bit rate and soundness parameter achieved by our experiment, the restricted adversarial model, as well as the numerous assumptions used in our analysis limit immediate deployment of the proposed protocol in production applications. However, we numerically analyze how future developments may improve the security and cost-effectiveness of our protocol.
Our experiments pave the way for new opportunities within cryptography and communication.

\section*{Acknowledgments}
We are grateful to Jamie Dimon, Daniel Pinto and Lori Beer for their executive support of JPMorganChase's Global Technology Applied Research Center and our work in Quantum Computing.
We thank the technical staff at JPMorganChase's Global Technology Applied Research Center for their invaluable contributions to this work. We are thankful to Johnnie Gray for helpful discussions on tensor network contraction path optimization using CoTenGra.  We acknowledge the entire Quantinuum team for their many contributions toward the successful operation of the H2 quantum computer with 56 qubits, and we acknowledge Honeywell for fabricating the trap used in this experiment. J.L., M.L., Y.A., and D.L.~acknowledge support from the U.S.~Department of Energy, Office of Science, under contract DE-AC02-06CH11357 at Argonne National Laboratory and the U.S.~Department of Energy, Office of Science, National Quantum Information Science Research Centers. S.A.~and S.H.~acknowledge the support from the U.S.~Department of Energy, Office of Science, National Quantum Information Science Research Centers, Quantum Systems Accelerator. T.H.~was supported by the U.S. Department of Energy, Office of Science, Advanced Scientific Computing Research program office under the quantum computing user program. This research used supporting resources at the Argonne and the Oak Ridge Leadership Computing Facilities. The Argonne Leadership Computing Facility at Argonne National Laboratory is supported by the Office of Science of the U.S.~DOE under Contract No.~DE-AC02-06CH11357. The Oak Ridge Leadership Computing Facility at the Oak Ridge National Laboratory is supported by the Office of Science of the U.S.~DOE under Contract No.~DE-AC05-00OR22725. This research used resources of the National Energy Research Scientific Computing Center (NERSC), a Department of Energy Office of Science User Facility using NERSC award DDR-ERCAP0030284.

\section*{Author Contributions}
M.P. and R.S. devised the project.
M.L., R.S., P.N., M.D., M.F.-F. designed the protocol implementation. M.L., R.S., P.N. implemented the code for circuit generation and client-server interaction. P.N. executed the experiments on the quantum computer and collected the data. M.L., R.S., P.N., A.A., J.L., D.L. implemented and benchmarked the tensor-network-based verification code. M.L. executed the verification on supercomputers and collected the data. Y.A. and T.S.H. provided support for supercomputer runs. M.L. and P.N. analyzed the data. M.L., R.S., P.N., S.C., S.-H. H., S.A. developed the complexity-theoretic analysis. M.L., R.S., P.N., W.-Y.K., E.C.-M., K.C., O.A., C.L. developed the main security analysis.
C.L., M.P., R.S., N.K., S.E., F.J.C. improved the adversarial model and enhanced its connection to applications. K.J.B., J.M.D., N.E., C.F., D.H., M.M., S.A.M., J.W., B.N., P.S. maintained, optimized, and operated the trapped-ion hardware and software stack. M.P. led the overall project as the lead principal investigator. All authors contributed to technical discussions and the writing and editing of the manuscript.

\section*{Competing Interests}

M.P., M.L., P.N. and R.S. are co-inventors on a patent application related to this work (no 18/625,605, filed on 3 Apr 2024 by JPMorganChase). The authors declare no other competing interests.

\section*{Additional Information}
Supplementary Information is available for this paper.

\section*{Data Availability}

The full data presented in this work is available at  \url{https://doi.org/10.5281/zenodo.12952178}.

\section*{Code Availability}

The code required to verify and reproduce the results presented in this work is available at \url{https://doi.org/10.5281/zenodo.12952178}.

\section*{Disclaimer}
This paper was prepared for informational purposes with contributions from the Global Technology Applied Research center of JPMorgan Chase \& Co. This paper is not a product of the Research Department of JPMorgan Chase \& Co. or its affiliates. Neither JPMorgan Chase \& Co. nor any of its affiliates makes any explicit or implied representation or warranty and none of them accept any liability in connection with this paper, including, without limitation, with respect to the completeness, accuracy, or reliability of the information contained herein and the potential legal, compliance, tax, or accounting effects thereof. This document is not intended as investment research or investment advice, or as a recommendation, offer, or solicitation for the purchase or sale of any security, financial instrument, financial product or service, or to be used in any way for evaluating the merits of participating in any transaction.

The submitted manuscript includes contributions from 
UChicago Argonne, LLC, Operator of Argonne National Laboratory (``Argonne'').
Argonne, a U.S.\ Department of Energy Office of Science laboratory, is operated
under Contract No.\ DE-AC02-06CH11357.  The U.S.\ Government retains for itself,
and others acting on its behalf, a paid-up nonexclusive, irrevocable worldwide
license in said article to reproduce, prepare derivative works, distribute
copies to the public, and perform publicly and display publicly, by or on
behalf of the Government.  The Department of Energy will provide public access
to these results of federally sponsored research in accordance with the DOE
Public Access Plan \url{http://energy.gov/downloads/doe-public-access-plan}.

\clearpage

\section{Methods}
\label{sec:methods}

The goal of the certified randomness protocol is to achieve the two properties:
\begin{enumerate}
    \item \textbf{Randomness Certification}: Outputs generated by the protocol should be close to unpredictable and uniformly distributed, uncorrelated with any side information the client, server, and the environment might possess.
    \item \textbf{Randomness Expansion}: The entropy the client certifies in the protocol should be larger than the entropy it consumes in generating the circuits and selecting the set for validation.
\end{enumerate}

The $M$ bitstrings received from the server, which we denote as $X^M$, do not directly satisfy the randomness certification requirement since they are not uniformly distributed. They are passed to a randomness extractor $Ext$ along with an extractor seed $K_{\rm ext}$, which is private to the client, in order to obtain the final output bits $K$ that are uniformly distributed along with some side information. The possible side information we consider is any classical information possessed by the client, the server, and the environment prior to the start of the protocol, and we denote this ``snapshot" of initial classical information as $I_{\rm sn}$. This snapshot includes any initial randomness possessed by the client or the server.

An ideal randomness certification protocol outputs a string of bits (in the register $K$) that is uniformly random and independent of $I_{\rm sn}$. That is to say, the ideal output of a successful randomness certification protocol is precisely $\tau_K\otimes\rho_{I_{\rm sn}}$, where $\tau_K$ is a maximally mixed state and $\rho_{I_{\rm sn}}$ is the quantum state representing any side information. In the event that the protocol aborts, the output is expected to be some ``abort state." We quantify the security or \emph{soundness} of our protocol by the closeness (as given by a trace distance) between the ideal output and the actual output produced by the protocol. Since a lower bound to the smooth min-entropy of the $M$ raw samples returned by the server suffices to guarantee soundness via the use of randomness extractors, we present our main result in terms of bounds on the smooth min-entropy of the returned samples.

\subsection{Protocol Details}\label{sec:protocol}

Our primary objective in the protocol design is to minimize the time between the client submitting a quantum circuit and receiving the corresponding bitstring. As a result, our protocol is designed to mitigate the following experimental considerations:
\begin{enumerate}
    \item There is a significant latency due to network communication and the time to load a circuit into the quantum device controls. Furthermore, there is also overhead associated with executing a circuit. To ameliorate this, instead of submitting circuits one at a time, we group the circuits into batches of $15$ or $20$ jobs, with each job consisting of two circuits joined by a layer of mid-circuit measurements and reset. Each batch of size $b$, therefore, consists of $2b$ circuits. 
    \item There is downtime associated with the device, such as during periodic calibrations. Before submitting a batch, a client probes the machine for readiness using a predetermined precheck circuit $C_{\rm precheck}$. This circuit announces the client's intent to submit a batch of circuits and triggers any server-side maintenance if necessary. 
    \item To ensure that the device does not stall and to keep the average time per sample low, we demand that the entire batch be returned within a cutoff time $ 2.5 \times 2b$ seconds. If the entire batch is not received within this cutoff time, we cancel all outstanding jobs in the batch, and we discard all bitstrings received from this batch. 
\end{enumerate}

To formally describe our experimental protocol with all details accurately represented (including details on challenge circuits generation and randomness extraction), we present the protocol below.

\subsubsection{Protocol Arguments}
\begin{align*}
  n \in \mathbb{N}~&:~\text{Number of qubits}\\
  d \in \mathbb{N}~&:~\text{Circuit depth}\\
  M \in \mathbb{N}~&:~\text{Total number of samples}\\
  b \in \mathbb{N}~&:~\text{Batch size}\\
  m \in \mathbb{N}~&:~\text{Test set size}\\
  K_{\rm seed}\in\{0, 1\}^r ~&:~\text{Random bitstring that is private}\\
  &~~~\text{to client}\\
  T_{b, \rm{cutoff}} ~&:~\text{Round-trip communication}\\
  &~~~\text{time threshold between the client and }\\
  &~~~\text{the server for a batch}\\
  t_{\rm{threshold}}~&:~\text{Threshold on the overall average }\\
  &~~~\text{time-per-sample}\\
  \chi~&:~\text{Threshold for the XEB test}\\
  Ext:~&:~(\kappa,\varepsilon_{\rm ext})\text{-Quantum-proof strong}\\
  &~~~\text{extractor (See SM~Definition 4)}\\
  K_{\rm ext} \in \{0,1\}^s~&:~\text{A random seed for the extractor} \\ 
  C_{\rm precheck} ~&:~\text{A predetermined ``precheck" instruction}\\
  &~~~\text{used to announce the client's readiness}\\
  &~~~\text{to submit a batch}
  \end{align*}

\subsubsection{Protocol Steps}
  \begin{enumerate}
  \item Set the samples collected $\mathcal{M}_{\text{keep}} = \emptyset$.
  \item Set $i=0, T_{\rm tot} = 0$.
   \item Initialize a pseudorandom generator with an $r$-bit seed $K_{\rm seed}$.
    \item %
      While $|\mathcal{M}_{\text{keep}}| < M$, run the following steps:
  \begin{enumerate}
  \item \textbf{Challenge Circuit Generation Subroutine: }The client generates each of the circuits $\{C_{i \cdot 2b+k}\}_{k=1}^{2b}$ as follows.
  \begin{enumerate}
    \item Initialize an empty circuit on $n$ qubits.
    \item For $j=1,\ldots,d$, run the following steps:
          \begin{enumerate}
              \item Sample $n$ SU(2) gates using the seeded pesudorandom genenerator and apply them to all $n$ qubits. 
              \item Apply the two-qubit gates corresponding to layer $T_j$ of the chosen edge-colored circuit topology.
          \end{enumerate}
  \item Sample $n$ SU(2) gates using the seeded pseudorandom generator and apply them to all $n$ qubits. 
  \end{enumerate}
  \item \textbf{Precheck: } The client submits the precheck circuit $C_{\rm precheck}$ and waits for a response. 
  \item \textbf{Client-Server Interaction Subroutine: }
  \begin{enumerate}
  \item  Start a timer. 
  \item The client submits the batch of circuits $\{C_{i \cdot 2b+k}\}_{k=1}^{2b}$ to the server. 
  \item The server responds with a batch of $2b$ bitstrings $\{x_{i \cdot 2b+k}\}_{k=1}^{2b}$.
  \item Stop the timer. Record interaction time $T_{b}$. 
  \item \textbf{Time Out Scenario: }If $T_{\rm b} > T_{b, \rm{cutoff}} $, then \textbf{discard} the batch.
  \item If the batch is not discarded, then client computes $\mathcal{M}_{\text{keep}} = \mathcal{M}_{\text{keep}} \bigcup_{k=1}^{2b} \{x_{i \cdot 2b+k}\}$ and accumulates the time $T_{\rm tot} = T_{\rm tot} + T_{\rm b}$.
  \item Client increments the counter, $i = i + 1$.
  \end{enumerate}
  \end{enumerate} 
  \item \textbf{Abort Condition 1: }If $T_{\rm tot}/|\mathcal{M}_{\text{keep}}| > t_{\rm{threshold}}$, then abort the protocol.
  \item \textbf{XEB Score Verification Subroutine: } 
  \begin{enumerate}
    \item \textbf{Test Set Construction: }
    The client samples a subset $\mathcal{V}$ of size $m$ randomly from $\mathcal{M}_{\text{keep}}$ using the seeded pseudorandom generator.
  \item Compute the score ${\rm XEB}_{\rm test} = \left((2^n/m) \cdot \sum_{j \in \mathcal{V}} |\langle x_j | C_j | 0\rangle|^2\right) - 1$.
  \item \textbf{Abort Condition 2: } If ${\rm XEB}_{\rm test}  < \chi$ then abort the protocol.
  \end{enumerate}
   \item If not abort, the client feeds the $M$ samples $x_1,\ldots , x_M$ together with the random seed $K_{\rm ext}$ to the extractor $Ext$. 
  \end{enumerate}
  Output: Conditioned on the protocol not aborting, the protocol returns $Ext(K_{\rm ext},(x_1,\ldots,x_M))$ as the final bitstring.

\subsection{Protocol Security}

Our primary theoretical contribution is the security of the implemented protocol against a restricted adversary. Our adversarial model considers realistic and near-term adversaries using best-known strategies (See SM Sec. III C for details). In brief, our adversary has a bounded classical computer and a quantum computer and uses both to generate the samples. Specifically, we make the following key assumptions about the adversary (further elaborated in SM Sec. III C):
\begin{enumerate}
    \item The server does not perform any postselection attacks; that is, the $M$ detected rounds in the protocol are a fair representation of the adversary behavior;
    \item Of the $M$ valid samples, the server \textit{a priori} selects $Q$ rounds for which it honestly returns samples by executing the challenge circuit on the quantum computer. For the remaining $M-Q$ samples, it returns deterministic samples obtained by simulating the circuits on a powerful classical computer (of power $\mathcal{A}$, measured in terms of number of floating point operations per second);
    \item For each of the $Q$ quantum rounds, it only interacts with the quantum computer once 
    (it does not attempt to oversample a circuit); 
\end{enumerate}
In practice, these assumptions are likely stronger than necessary; we leave adaptation of the formal cryptographic protocol for a relaxed set of assumptions to future work. 

To prove the security of the protocol, we prove a lower bound to the smooth min-entropy $H_{\min}^{\varepsilon_{\rm s}}(X^M|\tilde{I}_{\rm sn})$ of the bits before the extractor given this adversary, where $\tilde{I}_{\rm sn}$ is the initial snapshot of side information minus the randomness extractor seed. In order to do so, we first provide a bound on the probability that the server executing a fixed number $Q$ of quantum rounds passes the ${\rm XEB}$ test with threshold $\chi$ (see SM Sec. III D). We denote the event where the protocol does not abort as $\Omega$, the probability of not aborting as $\Pr[\Omega]$, and the upper bound on the probability as $\varepsilon_{\rm adv}(Q,\chi)$.

Now, given a target not-abort probability $\varepsilon_{\rm accept} = 4\varepsilon_{\rm s}$ (for an $\varepsilon_{\rm sou}$-sound protocol, $\varepsilon_{\rm accept} = 4\varepsilon_{\rm s} = \varepsilon_{\rm sou}$), the upper bound to $\Pr[\Omega]$ allows us to compute $Q_{\rm min}=\min\{Q\,:\,\varepsilon_{\rm adv}(Q,\chi)\geq4\varepsilon_{\rm s}\}$, which represents the minimum number of quantum rounds that the server needs to perform for the protocol to not abort with probability $4\varepsilon_{\rm s}$. Given $Q_{\rm min}$, we bound the smooth min-entropy of the samples $X^M$ given classical side information $\tilde{I}_{\rm sn}$ via the following theorem.

\vspace{1em}

\noindent\textbf{Theorem 1} \textit{Let $\Omega$ denote the event where the randomness certification protocol in Sec.~\ref{sec:protocol} does not abort and let $\sigma$ be the state over registers $X^M$ and $\tilde{I}_{\rm sn}$. Given $\varepsilon_{\rm s}\in(0,1/4)$, the protocol either aborts with probability greater than $1-4\varepsilon_{\rm s}$ or}
\begin{align}
    H_{\min}^{\varepsilon_{\rm s}}(X^M|\tilde{I}_{\rm sn})\geq Q_{\rm min}(n-1) + \log \varepsilon_{\rm s},
\end{align}
\textit{where $Q_{\rm min}=\rm{arg}\: \min_Q\{\varepsilon_{\rm adv}(Q,\chi)\geq 4\varepsilon_{\rm s}\}$ and $\varepsilon_{\rm adv}(Q,\chi)$ is the upper bound to $\Pr(\Omega)$.}

\end{document}


\title{Supplemental material for:\\ Certified Randomness Using a Trapped-Ion Quantum Processor}

\maketitle

\tableofcontents

\newpage

\section{Motivation: Complexity-Theoretic Security}\label{sec:complexity-theoretic-security}

A certified randomness expansion protocol must generate randomness that is secure against potential adversaries. The adversary receiving the quantum circuits could be using strategies that provide less entropy than the client desires (or no entropy at all). Ref.~\cite{aaronson2023certified} proves that if a device generates outputs that pass the \XEB{} test with high enough probability, it must generate $\Omega(n)$ bits of entropy. However, Ref.~\cite{aaronson2023certified} considers \XEB{} tests for $k$ samples generated from a fixed circuit, whereas our setting considers \XEB{} tests for $k$ samples generated from $k$ distinct circuits because we request one sample per circuit. In the remainder of this section we extend the security analysis of Ref.~\cite{aaronson2023certified} to our setting. In some cases we reference theorems stated in the extended arXiv version of Ref.~\cite{aaronson2023certified}, which is Ref.~\cite{aaronson2023certifiedArxiv}. The notation in this section aims to be consistent with \cite{aaronson2023certified,aaronson2023certifiedArxiv} and differs from that used in subsequent sections. Readers who are uninterested in complexity-theoretic motivations can skip to other sections without issue.

This argument builds on the LLHA (Long-List Hardness Assumption) conjecture, which was stated and proven for a random oracle model in Ref.~\cite{aaronson2023certified}. To state the conjecture, we first define the following problem: 

\begin{definition}[Long List Quantum Supremacy Verification LLQSV$(\D)$, restated from Problem 2 of \cite{aaronson2023certified}] We are given oracle access to $M = O(2^{3n})$ quantum circuits $C_1, \ldots, C_M$, each on $n$ qubits, which are promised to be drawn independently from the distribution $\U$. We are also given oracle access to $M$ strings $s_1, \ldots, s_M \in \{0,1\}^{n}$. The LLQSV$(\D)$ task is to distinguish the following two cases:
\begin{enumerate}
    \item \textbf{No-Case}: Each $s_i$ is sampled uniformly from $\{0, 1\}^{n}$.
    \item \textbf{Yes-Case}: Each $s_i$ is sampled from $p_{C_i}$, the output distribution of $C_i$.
\end{enumerate}
We denote the yes-case condition as $\vec s\sim p_{\vec C}$ and the no-case condition as $\vec s\sim\U^M$.
\label{def:llqsv}
\end{definition}

The LLHA conjecture is defined with respect to the above problem and a parameter $B$:
\begin{definition}[Long List Hardness Assumption $\text{LLHA}_{B}(\D)$, restated from Eq.~3 of \cite{aaronson2023certified} and Assumption 5.1 of \cite{aaronson2023certifiedArxiv}] For an $n$-qubit quantum circuit distribution $\D$ and some parameter $B<n$, denoted as ${\rm LLHA}_B(\D)$, it holds that 
\begin{equation}
        {\rm LLQSV}(\D) \notin \textup{\textsf{QCAMTIME}}(2^Bn^{O(1)})/q(2^Bn^{O(1)}),
    \end{equation}
    \label{def:llha}
    where $q$ denotes quantum advice.
\end{definition}
Here, $\textsf{QCAM}$ (Quantum Classical Arthur--Merlin) is a class of problems that admit an Arthur-Merlin protocol with classical communication and a quantum verifier.  $\QCAMTIME(T)/q(A)$ is the generalization of $\textsf{QCAM}$, where the verifier can use running time $T$ and receives $A$ bits of quantum advice that depend only on $n$. To justify this conjecture, Ref.~\cite{aaronson2023certified} shows that LLHA holds in the random oracle model. This serves as evidence that ${\rm LLHA}$ might hold for random circuit sampling. %

The LLHA conjecture can be used to show that any algorithm %
that achieves a high cross-entropy benchmark score with a sufficiently high probability must necessarily generate entropy. In particular, we can define what we mean by ``achieving a high cross-entropy benchmarking score." 
\begin{definition}[Linear Cross-Entropy Benchmarking $\text{LXEB}_{b,k}(\D)$, restated from Problem 1 of \cite{aaronson2023certified}] Let $\D$ be a probability distribution over $n$-qubit quantum circuits. The ${\rm LXEB}_{b,k}(\D)$ problem is to perform the following task \cite{Arute2019}: given $C\sim\D$, output a set of $k$ distinct bitstrings $\{z_1,\dots,z_k\}$ such that
    \begin{equation}
        \sum_{i=1}^k p_C(z_i) \geq \frac{bk}{N},
    \end{equation}
    where $p_C(z_i)$ is the probability that the output of quantum circuit $C$ is $z_i$, and $N=2^n$.
\label{def:lxeb}
\end{definition}
In other words, the $\text{LXEB}_{b,k}$ problem amounts to drawing samples that achieve a cross-entropy benchmarking score, over $k$ samples of a \textit{single} circuit $C$, higher than some threshold $b$. We remark that here we add the requirement that the bitstrings must be distinct to avoid the situation where an adversary repeats a single large probability bitstring to get a high score. The presence of entropy in such samples is guaranteed by the following theorem.
\begin{theorem}[$\text{LXEB}_{b,k}$ ensures von Neumann entropy, restated from Theorem 5.10 of \cite{aaronson2023certifiedArxiv}]\label{thm:lxeb_entropy} For integer $n$, assume that ${\rm LLHA}_{B}(\D)$ holds for distribution $\D$ over circuits acting on $n$ qubits. Then for any device which on input of
a circuit $C\sim \D$ outputs a classical state $Z$ over $\{0,1\}^{nk}$ ($k$ bitstrings, each of length $n$) solving ${\rm LXEB}_{b,k}$ with probability $q$, it holds that
    \begin{equation}
        H(Z\vert C)\geq\frac{B}{2}\left(\frac{bq-1}{b-1}-n^{-\omega(1)}\right).
    \end{equation}
\end{theorem}
While the theorem above guarantees von Neumann entropy $H(Z\vert C)$ in samples that pass the $\text{LXEB}_{b,k}$ problem, note that the statistic we use in the main text is slightly different from the statistic in \Cref{def:lxeb}. In particular, whereas $\text{LXEB}_{b,k}$ is defined with respect to many samples corresponding to the same circuit, the statistic we use is defined over a set of circuits. We define the modified problem in the same style as $\text{LXEB}_{b,k}$:
\begin{definition}[Mixed Linear Cross-Entropy Benchmarking $\text{MLXEB}_{b,k}(\D)$] Let $\D$ be a probability distribution over quantum circuits on $n$ qubits. Then the $\text{MLXEB}_{b,k}(\D)$ problem is as follows: given $\vec{C}$ with $|\vec{C}|=k$ and with each $C_i$ drawn from $\D$, 
output samples $z_1, \ldots, z_k \in \{0, 1\}^{n}$ such that
\begin{eqnarray}
    \sum_{i=1}^k p_{C_i}(z_i) \geq \frac{bk}{N}.
    \label{eq:mlxeb}
\end{eqnarray}
\label{def:mlxeb}
\end{definition}
Our main result for this section is a proof ensuring entropy for an algorithm that solves the $\text{MLXEB}_{b,k}(\D)$ problem. 
\begin{theorem}[${\rm MLXEB}_{b,k}$ ensures von Neumann entropy.]\label{theorem:mlxeb_entropy}
    For integer $n$, assume that ${\rm LLHA}_{B}(\D)$ holds for distribution $\D$ over circuits acting on $n$ qubits. Then for any device 
    that on input of a set of $k$ independently sampled circuits $\vec C$ with $\vec C\sim \D^{k}$ outputs a classical state $Z$ over $\{0,1\}^{nk}$ ($k$ bitstrings, each of length $n$) solving ${\rm MLXEB}_{b,k}$ with probability $q$, it holds that
    \begin{equation}
        H(Z\vert \vec C)\geq\frac{B}{2}\left(\frac{bq-1}{b-1}-n^{-\omega(1)}\right).
        \label{eq:entropy-mlxeb}
    \end{equation}
\end{theorem}
\begin{proof}
    The proof is given in \Cref{sec:mlxeb_proof}.
\end{proof}
\medskip 

Complexity-theoretic evidence that a device solving the $\text{MLXEB}_{b,k}$ problem must generate entropy gives justification that our protocol based on the statistics defined in Eq.~\ref{eq:xeb} can be used for randomness certification and expansion. For notational convenience, for the rest of the paper we use the term ``cross-entropy benchmarking score" to refer to $\text{MLXEB}_{b,k}$. Furthermore, our primary statistic, the $\XEB$ score (defined in Eq.~\ref{eq:xeb}) is the same as Eq.~\ref{eq:mlxeb} in the definition of $\text{MLXEB}_{b,k}$ up to a normalization factor and a constant offset. 

\subsection{Proof of \Cref{theorem:mlxeb_entropy}}\label{sec:mlxeb_proof}

To give a complexity-theoretic guarantee of entropy generation, we follow Ref.~\cite{aaronson2023certifiedArxiv} in proving an intermediate result showing that a polynomial time quantum algorithm that solves LXEB while having low von Neumann entropy can be used to solve LLQSV, thereby violating LLHA. We begin by stating the LXEB version of the result (Theorem~\ref{theorem:lxeb_solves_llqsv}) and clarifying the difference from the result in \cite{aaronson2023certifiedArxiv}. Then we state the MLXEB version that applies to our experiment (Theorem~\ref{theorem:mlxeb_solves_llqsv}) and provide a series of lemmas that are needed to prove both. Finally, we provide proofs for both theorems.

\begin{theorem}[Low-entropy algorithm solving LXEB also solves LLQSV, modified from Theorem 5.8 of \cite{aaronson2023certifiedArxiv}]\label{theorem:lxeb_solves_llqsv}
    Consider a device $\A$ which runs in quantum polynomial time and satisfies the following condition:
    \begin{eqnarray}
        H(Z|C)_\A < \frac{B}{2}\left(\frac{bq-1-\varepsilon}{b-1}\right),\\
        q=\underset{C\sim \D,\vec z\sim\A(C)}{\Pr}\left[\sum_i p_C(z_i) \geq \frac{bk}{N}\right],
    \end{eqnarray}
    where $\varepsilon=n^{-O(1)}$ and $q$ is the probability of $\A$ passing ${\rm LXEB}_{b,k}$. If such an $\A$ exists, then there is a quantum-classical Arthur-Merlin protocol which on input of an $O(n)$-bit advice string, solves ${\rm LLQSV}_B(\D)$ in time $2^Bn^{O(1)}$. In other words, ${\rm LLQSV}_B(\D)\in\textup{\textsf{QCAMTIME}}(2^Bn^{O(1)})/O(n)$. 
\end{theorem}
This theorem is modified from \cite{aaronson2023certifiedArxiv} since the original theorem requires $\A$'s min-entropy to be bounded, whereas our theorem requires $\A$'s von Neumann entropy to be bounded. It turns out that the original theorem is not technically possible. To see this, note that the original proof provides a lower bound on the single round min-entropy, which scales linearly in $n$. This is impossible: one can honestly perform quantum sampling 99\% of the time and return the zero bitstring 1\% of the time, which is an algorithm with min-entropy $H_{\rm min}=-\log_2{0.01}$ regardless of $n$ and passes ${\rm LXEB}$ with high probability. We fix this by bounding the von Neumann entropy for the single-round analysis. The technical issue with the analysis in \cite{aaronson2023certified,aaronson2023certifiedArxiv} that leads to this impossible conclusion is that the set used in the approximate counting protocol is ill-defined since it depends on random samples from the quantum algorithm, which Arthur and Merlin cannot agree on. Our analysis below is free from this issue.

If ${\rm LLHA}$ is true, \Cref{theorem:lxeb_solves_llqsv} implies that the output of the algorithm that solves ${\rm LXEB}$ must contain entropy. However, our experimental protocol does not check if the algorithm solves ${\rm LXEB}_{b,k}$ since we request one sample per circuit. Therefore, we wish to prove an equivalent theorem for the $\MLXEB$ case.

\begin{theorem}[Low-entropy algorithm solving MLXEB also solves LLQSV]\label{theorem:mlxeb_solves_llqsv}
    Consider a device $\A$ which runs in quantum polynomial time and satisfies the following condition:
    \begin{eqnarray}
        H(Z|\vec C) < \frac{B}{2}\left(\frac{bq-1-\varepsilon}{b-1}\right),\label{eq:low_vn_entropy}\\
        q=\underset{\vec C\sim \D^k,\vec z\sim\A(\vec C)}{\Pr}\left[\sum_i p_{C_i}(z_i) \geq \frac{bk}{N}\right],
    \end{eqnarray}
    where $\varepsilon=n^{-O(1)}$. If such an $\A$ exists, ${\rm LLQSV}_B(\D)\in\textup{\textsf{QCAMTIME}}(2^Bn^{O(1)})/O(n)$.
\end{theorem}

We provide the analysis for the ${\rm MLXEB}$ case. Adaptation of the proofs to the ${\rm LXEB}$ case is straightforward, and we remark on the necessary changes at the end of this section. We first prove that an algorithm satisfying the requirements of \Cref{theorem:mlxeb_solves_llqsv} must output heavy hitters (outputs with probability greater than $\tau/2^{B/2}$ for some $\tau\in(0,1]$) often.

\begin{lemma}\label{lemma:heavy_hitter_prob}
An algorithm satisfying the requirements of \Cref{theorem:mlxeb_solves_llqsv} must output heavy hitters with probability at least $p$ satisfying $p>\frac{b(1-q)+\varepsilon}{b-1}$, where
\begin{equation}
    p(\tau)=\underset{\vec C\sim \D^k,\vec z\sim\A(\vec C)}{\Pr}\left[\Pr[\A(\vec C)=\vec z]\geq \frac{\tau}{2^{B/2}}\right],
\end{equation}
and $p\equiv p(1)$.
\end{lemma}
\begin{proof}
With Markov's inequality, Eq.~\ref{eq:low_vn_entropy} implies
\begin{equation}
    \underset{\vec C\sim\D^k,\vec z\sim\A(\vec C)}{\Pr}\left[\log\frac{1}{\Pr[\A(\vec C)=\vec z]}\geq \frac{B}{2}\right]
    \leq \frac{H(Z|\vec C)}{B/2}
    < \frac{bq-1-\varepsilon}{b-1}
    \implies p>1-\frac{bq-1-\varepsilon}{b-1}=\frac{b(1-q)+\varepsilon}{b-1}.
\end{equation}

\end{proof}

We now examine the expectation value of a random variable $Y_\tau(\vec C,\vec s)$:
\begin{equation}
    Y_\tau(\vec C,\vec s) := \underset{\vec z\sim\A(\vec C)}{\Pr}\left[\Pr[\A(\vec C)=\vec z]\geq \frac{\tau}{2^{B/2}}\wedge \exists i:z_i=s_i \right].
\end{equation}
We see that $Y_\tau$ is the probability that the output of $\A$ consists of heavy hitters with probability satisfying the conditions of Lemma~\ref{lemma:heavy_hitter_prob}, with at least one matching bitstring. Define its expectation value in the yes-case and no-case of LLQSV defined in Def.~\ref{def:llqsv} as $\mu_1$ and $\mu_0$:
\begin{equation}
    \mu_1(\tau)=\underset{\vec C\sim\D^k,\vec s \sim p_{\vec C}}{\Exp}[Y_\tau(\vec C, \vec s)], \quad \mu_0(\tau)=\underset{\vec C\sim\D^k,\vec s \sim \U^k}{\Exp}[Y_\tau(\vec C, \vec s)].
\end{equation}

The gap between $\mu_1$ and $\mu_2$ helps distinguish between the yes- and no-cases. Showing the gap between $\mu_1$ and $\mu_0$ requires the following lemma.

\begin{lemma}\label{lemma:collision}
Denoting the condition that the output $\vec z$ passes ${\rm MLXEB}_{b,k}$ as $V(\vec C, \vec z)$,
\begin{equation}
    \underset{\vec s\sim p_{\vec C}}{\Pr}\left[\exists i,z_i=s_i\big\vert V(\vec C, \vec z)\right]\geq \frac{bk}{N}-O\left(\frac{1}{N^2}\right).
\end{equation}
\end{lemma}
\begin{proof}
Given any $\vec z$, whether $z_i=s_i$ is independent for different $i$ since the $s_i$ are independently sampled. Using the inclusion-exclusion principle to the second order,
\begin{align}
\underset{\vec s\sim p_{\vec C}}{\Pr}[\exists i,z_i=s_i]\geq &\sum_i \Pr[z_i=s_i]-\sum_{i<j}\Pr[z_i=s_i]\Pr[z_j=s_j]\\
\geq & \sum_i \Pr[z_i=s_i]-\sum_i \Pr[z_i=s_i]\sum_i \Pr[z_i=s_i]\\
=&\sum_{i} \Pr\left[z_{i}=s_{i}\right] - \left(\sum_{i} \Pr\left[z_{i}=s_{i}\right] \right)^{2}
\\
= & \sum_i p_{C_i}(z_i)\left(1-\sum_i p_{C_i}(z_i)\right)\label{eq:bound_prob}.
\end{align}
The value of Eq.~\ref{eq:bound_prob} increases until $\sum_i p_{C_i}=1/2$. For $\sum_i p_{C_i}>1/2$, $\Pr[\exists i,z_i=s_i]>\frac{1}{2k}>bk/N$ since $k=O(n^2)$. Conditioned on passing ${\rm MLXEB}_{b,k}$, we have $\sum_i p_{C_i}(z_i)\geq bk/N$. Therefore,
\begin{equation}
    \underset{\vec s\sim p_{\vec C}}{\Pr}\left[\exists i,z_i=s_i\big\vert V(\vec C, \vec z)\right]\geq \frac{bk}{N}-\frac{b^2 k^2}{N^2}.
\end{equation}
\end{proof}

Now we show the gap.

\begin{lemma}\label{lemma:original_gap}
    For $\tau\in[0,1]$,
    \begin{equation}
        \frac{\mu_1(\tau)}{\mu_0(\tau)}\geq b\cdot \frac{p(\tau)+q-1}{p(\tau)}.
    \end{equation}
\end{lemma}
\begin{proof}
    For the yes-case,
    \begin{align}
        \mu_1(\tau)=&\underset{\vec C\sim \D^k,\vec s\sim p_{\vec C}}{\Exp}[Y_\tau(\vec C,\vec s)]=\underset{\vec C\sim \D^k,\vec s\sim p_{\vec C},\vec z\sim\A(\vec C)}{\Pr}\left[\Pr[\A(\vec C)=\vec z]\geq \frac{\tau}{2^{B/2}}\wedge \exists i:z_i=s_i \right]\\
        \geq & \underset{\vec C\sim \D^k,\vec s\sim p_{\vec C},\vec z\sim\A(\vec C)}{\Pr}\left[\Pr[\A(\vec C)=\vec z]\geq \frac{\tau}{2^{B/2}}\wedge \exists i:z_i=s_i \wedge V(\vec C,\vec z)\right]\\
        = & \underset{\vec s\sim p_{\vec C}}{\Pr}\left[\exists i:z_i=s_i \big\vert \Pr[\A(\vec C)=\vec z]\geq \frac{\tau}{2^{B/2}} \wedge V(\vec C,\vec z)\right] \underset{\vec C\sim \D^k,\vec z\sim\A(\vec C)}{\Pr}\left[\Pr[\A(\vec C)=\vec z]\geq \frac{\tau}{2^{B/2}} \wedge V(\vec C,\vec z)\right]\\
        \geq & \left(\frac{bk}{N}-O\left(\frac{1}{N^2}\right)\right)\cdot \underset{\vec C\sim \D^k,\vec z\sim\A(\vec C)}{\Pr}\left[\Pr[\A(\vec C)=\vec z]\geq \frac{\tau}{2^{B/2}} \wedge V(\vec C,\vec z)\right]\\
        \geq & \frac{bk}{N}\cdot (p(\tau) + q - 1)-O\left(\frac{1}{N^2}\right).\label{eq:mu_1}
    \end{align}
    The fourth line holds by \Cref{lemma:collision}, and the last line holds by union bound.
    
    For the no-case,
    \begin{equation}
        \mu_0(\tau)=\underset{\vec C\sim \D^k,\vec s\sim \U^k}{\Exp}[Y_\tau(\vec C,\vec s)]=\underset{\vec C\sim \D^k,\vec s\sim \U^k,\vec z\sim\A(\vec C)}{\Pr}\left[\Pr[\A(\vec C)=\vec z]
        \geq \frac{\tau}{2^{B/2}}\wedge \exists i:z_i=s_i \right]
        = \frac{k}{N}\cdot p(\tau).\label{eq:mu_0}
    \end{equation}
    By Eq.~\ref{eq:mu_1} and \ref{eq:mu_0}, we conclude the proof.
\end{proof}

\Cref{lemma:heavy_hitter_prob,lemma:original_gap}
of this work and Lemma 5.4 of \cite{aaronson2023certifiedArxiv} imply that for $\mu_1$, $\mu_0$ defined for an algorithm satisfying the requirement of  \Cref{theorem:mlxeb_solves_llqsv},
\begin{equation}
    \frac{\mu_1(\tau)}{\mu_0(\tau)} \geq \frac{\mu_1(1)}{\mu_0(1)}\geq b\cdot\frac{p+q-1}{p}\geq 1+\varepsilon.
\end{equation}
Further, by the same argument used in Lemma 5.5 of \cite{aaronson2023certifiedArxiv} (modulo some typos), we have the following result. This result is necessary due to the need for Arthur to verify Merlin's claim in the Goldwasser-Sipser protocol. The manner in which this arises in our work can be seen in Eq.~\ref{eq:fraction_estimation}. For Ref.~\cite{aaronson2023certifiedArxiv}, it is discussed before Lemma 5.5 and can be seen in Eq.~112. For both cases, it is due to the fact that Arthur cannot calculate some fraction with infinite precision.
\begin{lemma}
    For $T\geq\frac{16}{\varepsilon}\log\left(\frac{N}{\varepsilon}\right)$, there exists $j\in[T/2]$ such that
    \begin{equation}
        \mu_{1}\left(1/2 + j/T\right) \geq \left(1 + \varepsilon/2\right)\mu_{0}\left(1/2 + (j-1)/T\right).\label{eq:contra-4}
    \end{equation}
\end{lemma}
\begin{proof}
We modify the proof as follows. First, $j\in[T/2]$, and for the product, $j$ goes from $0$ to $T/2-1$. Second, $\frac{8}{\varepsilon}\log\frac{n}{\varepsilon}$ should be $\frac{16}{\varepsilon}\log\frac{N}{\varepsilon}$.
\end{proof}

With this gap, a quantum-classical Arthur--Merlin protocol can be used to distinguish between the two cases. This is accomplished by first defining a set whose size depends on $\Exp[Y_\tau(\vec C, \vec s)]$, such that there is a gap in the set size between the two cases. This allows us to use the Goldwasser--Sipser protocol for approximate counting \cite{goldwasser1986private} to differentiate between the two cases. To show that there exists a set with a gap in the size, we use the following lemma.

\begin{lemma}\label{lemma:gap}
    Let $\varepsilon\in[0,1],\delta,a_0,a_1\in(0,1]$ be real numbers satisfying $a_1=a_0+\delta\leq (1+\varepsilon/8) a_0$. \ 
    Let $X_0,X_1$ be random variables in $[0,1]$ such that $X_b>0$ implies $X_b\geq a_b$ for $b\in\bit$ and $\Exp[X_1]\geq(1+\varepsilon)\Exp[X_0]$.
    Then there exists a rational number $1\geq t\geq a_0$ such that 
    \begin{align}
        \Pr[X_1>t+\delta/2] \geq \left(1+\varepsilon/4\right) \left(\Pr[X_0>t] + \frac{\varepsilon}{8}\Exp[X_1]\right). 
    \end{align}
\end{lemma}
\begin{proof}
    For $b\in\bit$, let $f_b$ and $F_b$ be the PDF and the CDF of $X_b$, respectively. \ 
    Recall that by integration by parts, 
    \begin{align}
        \Exp[X_b] &= \int_{a_b}^1  x f_b(x)\d x = x F_b(x)\vert_{a_b}^1 - \int_{a_b}^1 F_b(x)\d x = 1 - a_b\Pr[X_b\leq a_b] - \int_{a_b}^1 F_b(x)\d x\\
        &= a_b\Pr[X_b>a_b] + \int_{a_b}^1 \Pr[X_b>x]\d x.
    \end{align}
    The condition that $\Exp[X_1]\geq (1+\varepsilon)\Exp[X_0]$ implies that for $\varepsilon\leq 1$, 
    \begin{align}
    \Exp[X_1]
    &\geq \left(1-\frac{\varepsilon}{4}\right)(1+\varepsilon)\Exp[X_0]+\frac{\varepsilon}{4}\cdot \Exp[X_1] \geq \left(1+\frac{\varepsilon}{2}\right)\Exp[X_0]+\frac{\varepsilon}{4}\cdot \Exp[X_1] \\
    &\geq \left(1+\frac{\varepsilon}{2}\right)\left(\Exp[X_0]+\frac{\varepsilon}{8}\cdot \Exp[X_1]\right).
    \end{align} 
    This further implies one of the following two cases:
    \begin{eqnarray}
    a_1\Pr[X_1>a_1]
    \geq a_0\left(1+\frac{\varepsilon}{2}\right)\left(\Pr[X_0>a_0] +\frac{\varepsilon}{8}\cdot\Exp[X_1]\right),\label{eq:case_1}
    \\
        \int_{a_1}^1 \Pr[X_1>x]\d x\geq (1+\varepsilon/2)\int_{a_0}^1 \left(\Pr[X_0>x] + \frac{\varepsilon}{8}\cdot\Exp[X_1]\right)\d x.\label{eq:case_2}
    \end{eqnarray}
    In the former case of Eq.~\ref{eq:case_1}, since $a_1\leq (1+\varepsilon/8)a_0$, 
    \begin{align}
        \Pr[X_1>a_1] 
        &\geq \frac{1+\varepsilon/2}{1+\varepsilon/8} \left(
        \Pr[X_0>a_0] + \frac{\varepsilon}{8} \Exp[X_1]
        \right) \\
        &\geq (1+\varepsilon/4) \left(
        \Pr[X_0>a_0] + \frac{\varepsilon}{8} \Exp[X_1]
        \right). 
    \end{align}
    In the latter case of Eq.~\ref{eq:case_2}, we partition the interval $(a_0,1]$ into $m=2/\delta$ subintervals of equal length; that is, let $\alpha_\ell=a_0+\delta\ell/2$ (hence $\alpha_0=a_0$, $\alpha_2=a_1$ and $\alpha_m=1$) be the grid points and the subintervals be $\{[\alpha_{\ell-1},\alpha_\ell]: \ell\in[m]\}$. \
    Also define 
    \begin{equation}
        A_\ell := \int_{\alpha_{\ell-1}}^{\alpha_{\ell}} \Pr[X_1>x]\d x, \quad B_\ell := \int_{\alpha_{\ell-1}}^{\alpha_{\ell}} \left(\Pr[X_0>x] + \frac{\varepsilon}{8}\cdot\Exp[X_1]\right)\d x. 
    \end{equation}
    Now by Eq.~\ref{eq:case_2}, we have 
    \begin{align}
        \sum_{\ell=3}^m A_\ell \geq (1+\varepsilon/2) \sum_{\ell=1}^m B_\ell \geq (1+\varepsilon/2) \sum_{\ell=1}^{m-2} B_\ell. 
    \end{align}
    This implies that there exists $\ell\in[m-2]$ such that $A_{\ell+2}\geq (1+\varepsilon/2) B_\ell$. \ 
    Since $\Pr[X>x]$ is monotonically nonincreasing, 
    we have $\Pr[X_1>\alpha_{\ell+1}]\geq A_{\ell+2}$ and $B_\ell\geq \Pr[X_0>\alpha_\ell]$, and therefore 
    \begin{align}
        \Pr[X_1> t+\delta/2] \geq (1+\varepsilon/2)\left(\Pr[X_0> t] + \frac{\varepsilon}{8}\Exp[X_1]\right)
    \end{align}
    for $t=\alpha_\ell=a_0+\ell\delta/2$. 
\end{proof}

We now apply this lemma to our setting. Define $Z_{\tau,t}(\vec C,\vec s):=\Id[Y_\tau(\vec C,\vec s)>t]$.
\begin{corollary}\label{corollary:gap}
Assume that there is $1\geq \varepsilon=n^{-O(1)}$ such that for $T\geq\frac{16}{\varepsilon}\log(\frac{N}{\varepsilon}))$ there is $\tau\in[1/2,1]$ such that the inequality \eq{contra-4} holds.
Then for $\delta=\frac{1}{2^{B/2}T}$, there exists $t\geq \frac{1}{2^{B/2}}(\tau-1/T)$ such that 
\begin{enumerate}
    \item $\nu_1=\underset{\vec C\sim D^k,\vec s\sim p_{\vec C}}{\Pr}\left[Z_{\tau,t+\delta/2}\right]\geq \Omega(\varepsilon^2 k/N)$, and 
    \item $\nu_0=\underset{\vec C\sim D^k,\vec s\sim U^k}{\Pr}[Z_{\tau-1/T,t}]\leq (1+\varepsilon/8)^{-1}\nu_1$. 
\end{enumerate}
\end{corollary}
\begin{proof}
    We choose $X_1$ to be $Y_\tau(\vec C,\vec s)$ for $\vec s\sim p_{\vec C}$ and $\vec C\sim \D^k$. \ 
    Similarly, we choose $X_0$ to be $Y_{\tau-1/T}(\vec C,\vec s)$ for $\vec C\sim \D^k$ and $\vec s\sim \U^k$. \ Then, by Eq.~\ref{eq:contra-4}, we have $\Exp[X_1]\geq(1+\varepsilon/2)\Exp[X_0]$.

    By definition, $Y_\tau(\vec C,\vec s)>0$ implies that there exists a $\vec z$ that is output by $\A(\vec C)$ with probability $\tau/2^{B/2}$ and there is an index $i\in[k]$ such that $s_i=z_i$. \ 
    The probability of sampling such a $\vec z$ is at least $\frac{\tau}{2^{B/2}}$. \ 
    Thus $Y_\tau(\vec C,\vec s)>0$ implies that $Y_\tau(C,s)>\frac{\tau}{2^{B/2}}$. \ 
    To apply \Cref{lemma:gap}, we set $a_0=\frac{\tau-1/T}{2^{B/2}}$, $\delta=\frac{1}{T2^{B/2}}$, and $a_1=a_0+\delta$. \ 
    The ratio $\frac{a_1}{a_0}=\frac{\tau}{\tau-1/T}\leq 1+\frac{2}{\tau T}=1+o(\varepsilon)$ for $T\geq \Omega(\frac{1}{\varepsilon}\log(N/\varepsilon))$. \ 
    Now applying \Cref{lemma:gap}, there exists $t$ such that
    \begin{align}
        \underset{\vec C,\vec s\sim p_{\vec C}}{\Pr}[Y_\tau(\vec C,\vec s) > t + \delta/2] 
        \geq \left(1+\frac{\varepsilon}{8}\right)\left(\underset{\vec C\sim \D^k,\vec s\sim U^k}{\Pr}[Y_{\tau-1/T}(\vec C,\vec s)>t] + \frac{\varepsilon}{16}\Exp_{\vec C\sim \D^k,\vec s\sim p_{\vec C}}[Y_\tau(\vec C,\vec s)]\right).\label{eq:gap}
    \end{align}
    This immediately implies the second conclusion. Further, since $\Exp_{\vec C\sim \D^k,\vec s\sim p_{\vec C}}[Y_\tau(\vec C,\vec s)]\geq bk(p(\tau)+q-1)/N \geq \varepsilon k/N$, plugging it into Eq.~\ref{eq:gap} implies the first conclusion. \
\end{proof}

The first conclusion in \Cref{corollary:gap} implies that the event of acceptance cannot occur with zero probability, and thus there is always a suitable choice of the range $R$ of the hash function in the QCAM protocol. \ 
The second implies that for a sufficiently long list of size $M\geq N^3$, the number of 1's in the yes case must be at least $(1+O(\varepsilon))$ times that in the no case with overwhelming probability. \ 
We will formally show these in the following proofs. \ 

We now explain how to approximately compute $Z_{\tau,t}(\vec C,\vec s)$ in quantum time $2^Bn^{O(1)}$ given $\vec C,\vec s$: \  
Consider the following sampling process $\mathcal{Z}(\vec C,\vec s)$ to confidence level $1-\eta$ for $\eta=2^{-n^{O(1)}}$:
\begin{enumerate}
    \item For $\delta:=\frac{1}{2^{B/2}T}$, take $K=\frac{2}{\delta^2}((nk+2)\ln 2+\ln(1/\eta))=2^Bn^{O(1)}$ samples $\vec z_1,\ldots,\vec z_K\sim\A(\vec C)$ to obtain an approximation of $\tilde p_{\A(\vec C)}$ of the distribution of $\A(\vec C)$. \ 
    \item Take $L=\frac{32}{\delta^2}\ln(2/\eta)=2^Bn^{O(1)}$ samples $\vec z_1,\ldots,\vec z_L\sim\A(\vec C)$, and count the fraction of samples 
    \begin{align}
        \tilde y = \frac{1}{L}\cdot \#\left\{j\in[L]: \tilde p_{\A(\vec C)}(\vec z_j)\geq \frac{\tau-1/{2T}}{2^{B/2}}\wedge \exists i: z_{j,i}=s_i\right\}\label{eq:fraction_estimation},
    \end{align}
    where $z_{j,i}$ is the $i$th entry of $\vec z_j$.
    \item If $\tilde y\geq t+\delta/4$, output $1$; otherwise output $0$. 
\end{enumerate}

To prove that the sampling process $\mathcal{Z}$ produces the correct output with high probability (\Cref{lemma:low-error}), we prove that each step of $\mathcal{Z}(\vec C,\vec s)$ has low errors. 
\begin{lemma}\label{lemma:approx-p}
    Step 1 outputs an approximation $\tilde p_{\A(\vec C)}$ that is close to $\A(\vec C)$'s distribution $p_{\A(\vec C)}(\vec z)=\Pr[\A(C)=\vec z]$ in $\ell_\infty$-distance at most $\delta/2$ with probability at least $1-\eta/2$. \ 
\end{lemma}
\begin{proof}
    By Hoeffding's inequality, for every $\vec z\in\bit^{nk}$, with probability at least $1-\eta/2^{nk+1}$, 
    \begin{align}\label{eq:31}
        \left| p_{\A(\vec C)}(\vec z) - \tilde p_{\A(\vec C)}(\vec z)\right| \leq \frac{\delta}{2}. \ 
    \end{align}
    By a union bound, with probability at most $\eta/2$, every $O$ satisfies \eq{31}, implying that their $\ell_\infty$-distance is at most $\delta/2$. \
\end{proof}
\begin{lemma}\label{lemma:low-error}
    The sampling process $\mathcal{Z}$ solves the promise problem of distinguishing (yes) $Z_{\tau,t+\delta/2}(\vec C,\vec s)=1$ and (no) $Z_{\tau-1/T,t}(\vec C,\vec s)=0$ to within error $\eta$. That is, 
    \begin{enumerate}
        \item If $Z_{\tau,t+\delta/2}(\vec C,\vec s)=1$, then $\mathcal{Z}(\vec C,\vec s)=1$ with probability at least $1-\eta$.  
        \item If $Z_{\tau-1/T,t}(\vec C,\vec s)=0$, then $\mathcal{Z}(\vec C,\vec s)=0$ with probability at least $1-\eta$.
    \end{enumerate}
    Furthermore, $\mathcal{Z}$ runs in time $2^Bn^{O(1)}$ for $\eta=2^{-n^{O(1)}}$. \ 
\end{lemma}
\begin{proof}
    
    For $\vec C,\vec s$ satisfying $Z_{\tau,t+\delta/2}(\vec C,\vec s)=1$, by Hoeffding's inequality, 
    with probability at least $1-e^{-L\delta^2/32}=1-\eta/2$, 
    \begin{equation}
        \tilde y
        \geq \Pr_{\vec z\sim\A(\vec C)}\left[\tilde p_{\A(\vec C)}(\vec z)\geq \frac{\tau}{2^{B/2}}-\frac{\delta}{2} \wedge \exists i:z_i=s_i\right] - \frac{\delta}{8}.
    \end{equation}
    Then, assume that the approximation $\tilde p_{\A(\vec C)}(\vec z)$ obtained from Step 1 is close to $p_{\A(\vec C)}(\vec z)$ in $\ell_\infty$-distance $\delta/2$, which happens with probability at least $1-\eta/2$ by \Cref{lemma:approx-p},
    \begin{equation}
        \tilde y \geq \Pr_{\vec z\sim\A(\vec C)}\left[\Pr[\A(\vec C)=\vec z]\geq \frac{\tau}{2^{B/2}} \wedge \exists i:z_i=s_i\right] - \frac{\delta}{8} = Y_\tau(\vec C,\vec s) - \frac{\delta}{8} > t +\frac{3\delta}{8}.  
    \end{equation}
    The last inequality follows from the assumption that $Y_\tau(\vec C,\vec s)>t+\delta/2$. \ 
    
    For $(\vec C,\vec s)$ satisfying $Z_{\tau,t}(\vec C,\vec s)=0$, again by Hoeffding's inequality, with probability at least $1-\eta/2$, 
    \begin{equation}
        \tilde y
        \leq \Pr_{\vec z\sim\A(\vec C)}\left[\tilde p_{\A(\vec C)}(\vec z)\geq \frac{\tau}{2^{B/2}}-\frac{\delta}{2}\wedge \exists i:z_i=s_i\right] + \frac{\delta}{8}.
    \end{equation}
    Again, with probability at least $1-\eta/2$, $\tilde p_{\A(\vec C)}(\vec z)$ is $\delta/2$ close to $p_{\A(\vec C)}(\vec z)$,
    \begin{equation}
        \tilde y \leq \Pr_{\vec z\sim\A(\vec C)}\left[\Pr[\A(\vec C)=\vec z]\geq \frac{\tau}{2^{B/2}}-\delta\wedge \exists i:z_i=s_i\right] + \frac{\delta}{8}
        = Y_{\tau-1/T}(\vec C,\vec s) + \frac{\delta}{8}\leq t + \frac{\delta}{8}. 
    \end{equation}
    The last inequality follows from the assumption that $Y_{\tau-1/T}(\vec C,\vec s)\leq t$. Combined with the observation that $\mathcal{Z}$ runs in time $2^Bn^{O(1)}$ for $\eta=2^{-n^{O(1)}}$ (since $K=2^Bn^{O(1)}$ samples are obtained by calling $\A(\vec{C})$ in step 1), this completes the proof. \ 
\end{proof}

Now we describe the QCAM protocol as follows. 
\begin{enumerate}
    \item Both Arthur and Merlin are given access to classical advice $\tau,T,t,R$ ($\delta$ can be computed from them).

    \item Arthur sends a hash function $h$ that maps $[M]\to[R]$ and a random image $y\in[R]$. \ 

    \item Merlin sends an index $i$. \ 

    \item Arthur accepts if $h(i)=y$ and $\mathcal{Z}(\vec C_i,\vec s_i)=1$, and rejects otherwise. 
\end{enumerate}

\begin{lemma}\label{lemma:concentration}
    For $M'\geq N^3/k$ size-$k$ sets of circuits and bitstrings, there exist integers $\kappa\in[M']$ such that with probability at least $1-e^{-\Omega(N\varepsilon^6)}$, Arthur accepts between $\kappa$ and $(1+O(\varepsilon))\kappa$ tuples in a yes instance and accepts at most $\kappa'=(1-\varepsilon/32)\kappa$ tuples in a no instance.
\end{lemma}
\begin{proof}

    In the yes-case, by \Cref{corollary:gap}, $\nu_1=\Omega(\varepsilon^2 k/N)$. The probability that Arthur accepts $\mathcal{Z}(\vec C,\vec s)$ is at least $\nu_1-\eta$ for $\eta=2^{-n^{O(1)}}$ by \Cref{lemma:low-error}. For $M'=N^3/k$, by Hoeffding's inequality, with probability at least $1-2e^{-M'(\nu_1-\eta)^2\xi^2}\geq 1-2e^{-\Omega(N\varepsilon^4 k\xi^2)}$, the number of accepted tuples is between $[(1-\xi)M'(\nu_1-\eta),(1+\xi)M'(\nu_1-\eta)]$ for some parameter $\xi\leq\varepsilon$ to be determined later. 
    We may choose $\kappa=(1-\xi)M'(\nu_1-\eta)$. \ 

    In the no-case, the probability that Arthur accepts $\mathcal{Z}(\vec C,\vec s)$ is at most $\nu_0+\eta$ by \Cref{lemma:low-error}. Let $\tilde \nu_0$ denote the fraction of tuples accepted by $Z_{\tau-1/T,\delta}$ in a no instance. \ 
    By Hoeffding's inequality, for $\varepsilon\leq 1$, 
    \begin{align}
        \Pr[\tilde \nu_0 \geq (1-\xi)^2 (\nu_1-\eta)]
        &\leq \Pr[\tilde \nu_0 - (\nu_0 + \eta)\geq (1-2\xi)(\nu_1-\eta)- (\nu_0 + \eta)] \\
        &\leq \Pr[\tilde \nu_0 -  (\nu_0 + \eta) \geq (1-2\xi)(\nu_1-\eta)-\left((1-\varepsilon/8)\nu_1+\eta\right)] \\
        &\leq e^{-M'(\varepsilon/8-2\xi)^2 \nu_1^2+O(M'\eta)}
        = e^{-\Omega(\varepsilon^6 k N)}
    \end{align}
    for $\xi=\varepsilon/32\leq 1/32$. We may choose $\kappa'=(1-\xi^2)M'(\nu_1-\eta)$.

    Therefore, the ratio is
    \begin{align}
        \frac{\kappa'}{\kappa}\leq \frac{(1-\xi)^2 (\nu_1-\eta)}{(1-\xi) (\nu_1-\eta)}=1-\xi=1-\varepsilon/32.
    \end{align}
\end{proof}

We are now ready to prove our main result.
\begin{proof}[Proof of Theorem~\ref{theorem:mlxeb_solves_llqsv}]
    Given $M\geq N^3$ circuits and bitstrings, we can group them into $M'=(M/k)$ size-$k$ sets of circuits and bitstrings. By \Cref{lemma:concentration} of this work and Lemma 5.2 of \cite{aaronson2023certifiedArxiv}, for $R=64\alpha\kappa/\varepsilon$, since $\alpha=1+O(\varepsilon)$, the gap is $\Omega(\varepsilon^2)$. Thus running an $(1/\varepsilon)^{O(1)}$-fold parallel repetition of the above protocol yields a constant gap. Since $1/\varepsilon=n^{O(1)},T=n^{O(1)},R\in[M]$, the advice $\tau,T,t,R$ all require relative precision at most $1/N^{O(1)}$, and hence the length is $O(n)$.

\end{proof}

\begin{proof}[Proof of Theorem~\ref{theorem:mlxeb_entropy}]
    Assuming LLHA holds, \Cref{theorem:mlxeb_solves_llqsv} implies that the output of the algorithm that solves ${\rm MLXEB}$ must contain entropy, completing the proof of \Cref{theorem:mlxeb_entropy}.
\end{proof}

\begin{proof}[Proof of Theorem~\ref{theorem:lxeb_solves_llqsv}]
The proof for the LXEB case follows trivially by first changing  $\vec C,\D^k,\vec s,\U^k$ into $C,\D,s,\U$ and the condition $\exists i:z_i=s_i$ into $s\in \vec z$. Furthermore, instead of \Cref{lemma:collision}, we use Eq.~103 of \cite{aaronson2023certifiedArxiv}. The rest of the analysis goes through unchanged.
\end{proof}

\rev{
\subsection{Multi-round Analysis}
The single-round result of \Cref{thm:lxeb_entropy} (or Theorem 5.10 of \cite{aaronson2023certifiedArxiv}) and \Cref{theorem:mlxeb_entropy} does not result in a complete and sound single-round protocol. The first reason is that even for a device without entropy, the theorems only guarantee $bq<1$. However, since $b$ cannot be larger than 2 for an honest quantum computer to pass with overwhelming probability, we can only say a device without entropy must fail with probability greater than or equal to $1/2$. Instead, we need it to fail with overwhelming probability.

As such, a multi-round protocol is necessary, which is presented in Figure 1 of \cite{aaronson2023certifiedArxiv}. Theorem 5.11 of \cite{aaronson2023certifiedArxiv} shows that the multi-round protocol has a lower bound on the smooth min-entropy linear in the number of rounds and the number of qubits, and this result is obtained using the entropy accumulation theorem (EAT). We note that the EAT used here is a modified version in Section 4 of \cite{aaronson2023certifiedArxiv} instead of the original one in \cite{dupuis2020entropy}. The application of EAT in \cite{aaronson2023certifiedArxiv} uses the single round von Neumann entropy bound provided by Theorem 5.10 of \cite{aaronson2023certifiedArxiv}. This is not affected by the fact that Corollary 5.9 of \cite{aaronson2023certifiedArxiv} is incorrect since Theorem 5.10 of \cite{aaronson2023certifiedArxiv} is still correct (use our \Cref{theorem:lxeb_solves_llqsv} instead of Theorem 5.8 of \cite{aaronson2023certifiedArxiv}).

We can define a similar multi-round protocol, where the only difference is instead of sending one circuit per round, we send $k$ circuits. In \oldtwo{exactly} the same manner as \cite{aaronson2023certifiedArxiv}, we \oldtwo{can}\revtwo{should be able to} apply EAT using the single round von Neumann entropy bound provided by \Cref{theorem:mlxeb_entropy}, which allows us to obtain \oldtwo{the same}\revtwo{a} lower bound on the smooth min-entropy for the modified multi-round protocol. We do not reproduce the multi-round analysis here.\oldtwo{since the application of EAT is trivial.}}

\subsection{Limitations of Asymptotic Guarantees}
\label{sec:limitations-of-asymptotic-analysis}
Three main challenges limit the applicability of the asymptotic result in Theorem~\ref{theorem:mlxeb_entropy} to our finite-sized experiment: our experiment is not in the asymptotic regime, the values of constants in Eq.~\ref{eq:entropy-mlxeb} are not known for our experiment, and there are differences between the protocol we implement experimentally and the protocol required by the asymptotic analysis.

First, the statement of Theorem~\ref{theorem:mlxeb_entropy} applies in the asymptotic limit of large $n$ and extending it to finite-sized experiments would require further refinement of the analysis in Sec.~\ref{sec:complexity-theoretic-security}.

Second, the analysis hinges on the correctness of ${\rm LLHA}_{B}(\D)$. Although an upper bound of $B\leq n/2$ is known due to Grover's algorithm, there is no lower bound on $B$ for any distribution $\D$. 

Third, the analysis guarantees single-round entropy for an output of $k$ bitstrings. To lower bound the entropy across multiple rounds, Ref.~\cite{aaronson2023certified} defines a protocol and proves entropy accumulation over multiple rounds, but the protocol is very different from our experiment. In particular, the authors perform multiple rounds of \XEB{} tests to estimate the probability of passing the test, which allows them to lower bound the entropy. If each round is over hundreds of circuits and we have hundreds of test rounds, the cost of computing the probabilities would be prohibitive.

For these reasons, we cannot determine the soundness of our protocol solely on the basis of complexity-theoretic analysis that may or may not apply for experiments of our scale. Consequently, we instead focus on finite-sized adversaries performing realistic attacks with numerically bounded computational power.

\section{Overview of the Protocol}\label{sec:protocol-summary}

To address the limitations of the asymptotic guarantees presented above, we implement a slightly modified protocol in our experiment, the security of which we analyze against a finite-sized adversary in Sec.~\ref{sec:security-analysis-of-protocol}. The protocol has been presented in Methods and Fig.~1 of the main text, and we briefly summarise it below.

\subsection{Protocol Details}

The protocol consists of four main steps: (1) challenge circuit generation, (2) client-server interaction, (3) \XEB{} score verification, and (4) randomness extraction.

\begin{enumerate}
    \item \textbf{Challenge Circuit Generation}: The client generates a large number of $n$-qubit challenge circuits $C_i$ using a pseudorandom number generator with a $r$-bit random seed. The circuits are chosen to balance the following two considerations: (1) no classical adversary should be able to simulate the circuits within the time that a quantum computer takes to run them, and (2) a classical supercomputer should be able to validate them. Detailed discussion on circuit selection is deferred to Sec.~\ref{sec:protocol_circuit_selection}. These circuits are kept secret from the client.

    \item \textbf{Client-Server Interaction}: The client sends a batch of circuits to the quantum server and requests one sample from the output distribution of each circuit. The circuits are submitted in batches of $b$ jobs, where each job consists of two circuits stitched via a layer of mid-circuit measurement and reset. Such ``batching'' and ``stitching'' allow us to amortize network latencies and execution overheads across $2b$ samples (Note: the symbol  $b$ used here to denote batch size is not related to the symbol $b$ used to parameterize LLHA in \ref{sec:complexity-theoretic-security}). Each batch is preceded by an agreed-upon ``precheck'' circuit, which announces the client's intention to submit a batch, checks for the server's readiness to accept jobs, and triggers calibrations if necessary. For each batch, the quantum server executes the circuits on the quantum computer and returns the samples $\{x_1, \ldots, x_{2b}\}$, with $x_i \in\{0,1\}^n$. On the Quantinuum H2-1 hardware, detectable faults such as loss of an ion from the trap may prevent a batch from completing execution. As a result, if the entire batch is not returned within a time of $\batchcutoff$, all outstanding jobs are cancelled, and any samples collected from this batch are discarded. The client-server interaction continues until we receive $M$ valid samples. If the average time per sample of the retained rounds exceeds a threshold $\Tqcthreshold$, the protocol aborts.

    \item \textbf{$\XEB{}$ Score Verification}:
    Having collected $M$ samples, $\{x_i\}_{i \in [M]}$, corresponding to $M$ circuits, the client randomly selects a subset of circuit-bitstring pairs defined by a set of $m$ indices $\V \subset [M]$  and computes the linear cross-entropy benchmarking score, referred to as $\XEBTest$ throughout, which we define as
    \begin{equation}\label{eq:xeb}
        \XEBTest = \frac{N}{m}\sum_{i \in \V} p_{C_i}(x_i) - 1,
    \end{equation}
    where $N=2^n$ and $p_{C_i}(x_i)$ is the probability of measuring the bitstring $x_i$ after executing circuit $C_i$ on an ideal quantum computer. These probabilities are computed classically on a large supercomputer by simulating $\{C_i\}_{i \in \V}$. If $\XEBTest$ is smaller than the protocol parameter $\chi$, the client aborts the protocol.

    \item \textbf{Randomness Extraction}: If the protocol does not abort, the client uses a seeded randomness extractor with seed $\ExtSeed$ and input $X^M := \{x_1,\ldots, x_M\}$ to extract the final random bits in register $K$.
\end{enumerate}

\subsection{Necessary Conditions for Protocol Success}\label{sec:conditions}

Let $\A$ denote the power of the adversary's supercomputer (e.g., measured by the number of floating-point operations per second, or FLOPS), and let $\budget$ denote the cost to simulate a given circuit exactly (e.g., measured by the number of floating-point operations, or FLOP count).
Let $\Tqcthreshold$ denote the average allowed time that the quantum computer can take to return one sample. Since the circuit simulation cost scales linearly in the target fidelity of simulation when using exact tensor network contraction, the adversary can simulate each circuit to a fidelity of $\A\cdot \Tqcthreshold/\budget$ in time $\Tqcthreshold$. Then, the protocol is secure only if an honest quantum server's fidelity on challenge circuits, $\phi$, is much larger than the fidelity a classical adversary may obtain. That is,
\begin{equation}
   \phi \gg \A \cdot \Tqcthreshold/\budget.
\end{equation}

Loosely speaking, the protocol guarantees true entropy resulting from the measurement of a quantum state only if three conditions are met. First, the quantum computer must be able to execute the challenge circuits with a high fidelity $\phi$. Second, the quantum computer must be able to return bitstrings within a short time $\Tqcthreshold$. Third, the verifier should be able to validate circuits with high simulation cost $\budget$ using supercomputers. In our implementation, all three conditions are satisfied.

\subsection{Difference from Existing Protocols}
In the protocols presented and analyzed in Refs.~\cite{aaronson2023certified,bassirian2021certified}, many samples from the same circuit $C$ are demanded in each round, and the cross-entropy benchmarking score is computed over those samples for each circuit. Specifically, in Refs.~\cite{aaronson2023certified,bassirian2021certified} the pertinent score is $\XEB = N\langle p_\text{C}(x_i)\rangle_{i \in \V} - 1$. In contrast, in our protocol we demand only a single bitstring from each circuit, and we compute the $\XEB$ score as an average over multiple circuits $C_i$: $\XEB = N\langle p_{C_i}(x_i)\rangle_{i \in \V} - 1$. This difference is motivated by the fact that, for an adversarial server that may be trying to pass off classically simulated samples as true measurements of a quantum circuit, classically sampling many shots of a single circuit is much easier than sampling a single shots for many circuits. In other words, classical simulation cost is strongly sublinear in the number of shots. For example, this is shown in Figure 2 of Ref.~\cite{Liu2024}. At the same time, for the trapped-ion quantum processor that we use in this experiment, the device runtime
scales nearly linearly with the number of shots.

\section{Security Analysis of the Protocol}
\label{sec:security-analysis-of-protocol}

Our experiment is motivated by complexity-theoretic security guarantees discussed in Sec.~\ref{sec:complexity-theoretic-security}. However, these guarantees do not directly apply to our finite-sized experiment (Sec.~\ref{sec:limitations-of-asymptotic-analysis}).
In this section we analyze the protocol as implemented and provide a proof that it is secure against a \rev{restricted} \old{realistic} class of adversaries operating under reasonable assumptions.

\rev{We note that the adversary we consider here is chosen such that it implements state-of-the-art techniques on both classical and quantum computers that can be realized in the near term.
Therefore, it can only execute a very limited set of actions, and the results in this section do not apply if the actual adversary strategy deviates from our prescription. We discuss the details of the adversary model itself and the limitations of the model in detail in \Cref{sec:finite-size-assumption}. Nevertheless, the analysis provides insight into the relationship between security, entropy, quantum computer fidelity and speed, and classical computational power, and provides order-of-magnitude estimates relevant for near-term implementations.

When we determine the amount of entropy safe to extract, we must make an assumption on the capabilities of the adversary. A way to numerically estimate a confidence on security is to an assume unreasonably powerful adversary (say 10 or 100 times more powerful than expected). 
We do explicitly acknowledge that if the actual adversary is fundamentally different from the analyzed adversary (e.g. a fault tolerant quantum computer executing a novel low-entropy algorithm that scores high on XEB, or a completely novel classical technique that significantly outperforms current state-of-the-art), it is not sufficient to simply adjust the power of the adversary in our model and a fundamentally different approach would be needed. }

\subsection{Security Definition}
\label{sec:security-definition}

\begin{figure}[h]
    \includegraphics[width=0.6\textwidth]{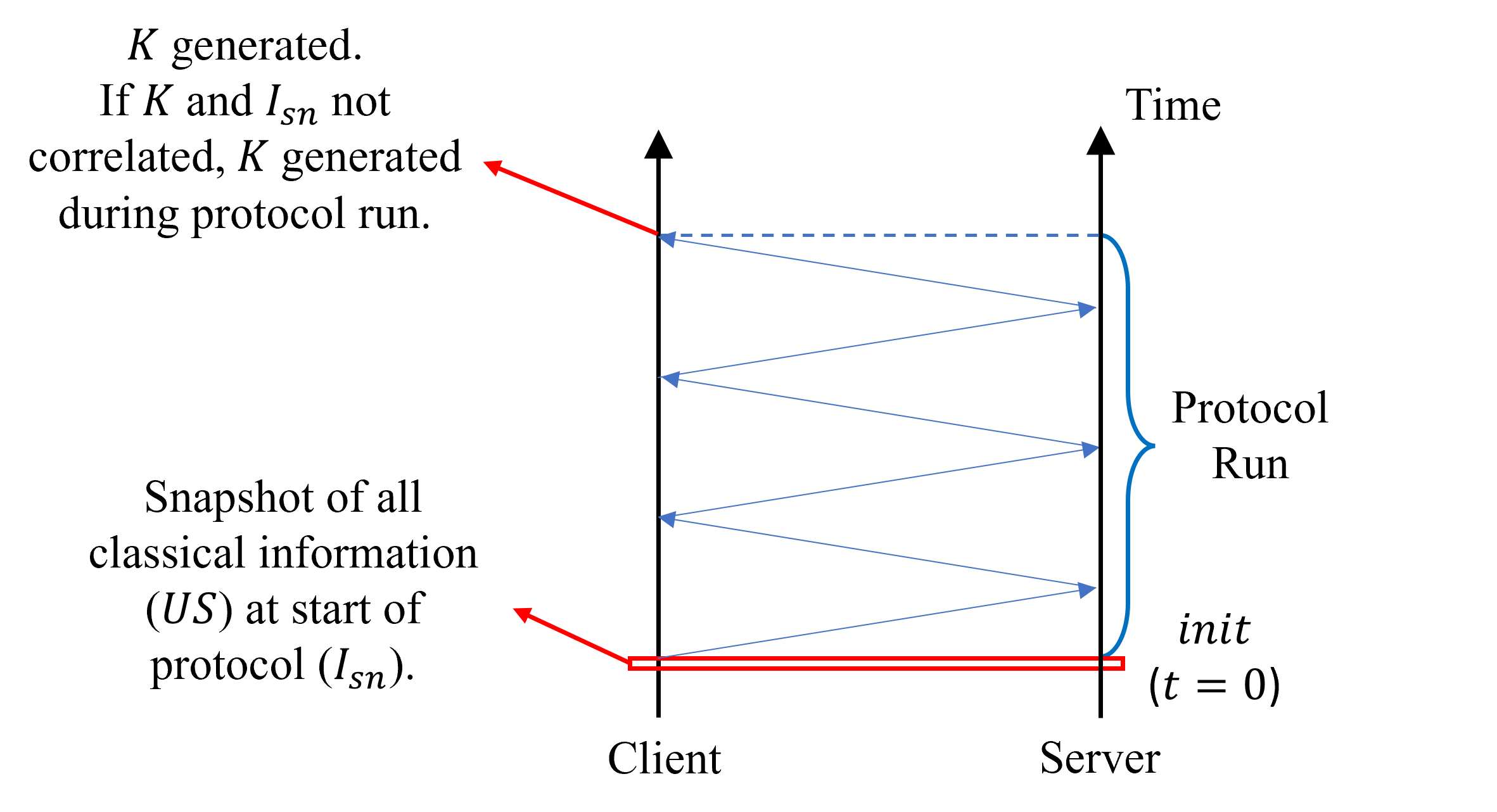}
    \caption{\label{fig:intuitive_CR}Illustration of the security definition of certified randomness. In the ideal functionality, the generated bitstring $K$ should be random and independent on any classical information available at the start of the protocol (described as a snapshot $\snapshot$ of all classical information).}
\end{figure}

In this section we formalize the definition of ``security" for certified randomness. 
In the protocol, a classical client $U$ interacts with a remote server $\tilde{S}$.  
Since we do not assume an authenticated communication channel between the client and the server, we have to account for other parties (denoted by $E$) who may have access to the communication channel. 
To simplify the analysis, we can encapsulate the remote server $\tilde{S}$ and other parties $E$ into a single entity, which we refer to as the server $S=\tilde{S}E$ in the rest of the manuscript. 
If authenticated communication is present between the client and server, we refer to the remote server as the server instead, that is, $S=\tilde{S}$.

With this setup let us consider what it means to have ``certified" randomness. The goal of ``certified" randomness is to obtain a bitstring $K$ that is ``unpredictable" and that can be certified as such. For simplicity, we {assume} that there is no initial quantum state in the server’s quantum memory at the beginning of the protocol~\footnote{If an initial quantum state is present in the server’s quantum memory, we can snapshot the classical description of the quantum state which WLOG can be stored in its classical memory since the server (adversary) prepares its own internal quantum memory.}.  We leave a proper accounting of quantum memory to future work. Furthermore, we {assume} that any classical random variables used in the protocol have been generated and stored in their respective devices at this point~\footnote{WLOG, this can always be made true. In this case, any processing of the random variables will be deterministic, and any bitstring generated from these random variables alone would be deterministic conditioned on the random variables.} (e.g., extractor seed $\ExtSeed$). Now, consider the state of the entire system $US$ before the start of the protocol at an initial moment, and take a snapshot of the system (namely, a copy of all classical registers in the system) at the initial moment ($t=0$), labeled as $\snapshot=U^{0}S^{0}$ \revtwo{(where $S$ includes both server $\tilde{S}$ and environment $E$)}.
We can then define the ``unpredictability" of bitstring $K$ as the property that $K$ is uniformly random and uncorrelated to $\snapshot$, i.e. that no classical information at the start of the protocol (captured in $\snapshot$) can influence or predict the bitstring $K$, as shown in Fig.~\ref{fig:intuitive_CR}. Intuitively, if $K$ and $\snapshot$ are uncorrelated, the source of randomness for $K$ cannot be classical, since $K$ could otherwise be deterministically computed from $\snapshot$. This leads to the conclusion that the source of randomness in $K$ must be quantum. 
As such, we can define the \underline{ideal functionality} of the protocol as generating a bitstring $K$ that is (1) uniformly distributed and (2) statistically independent of $\snapshot$ when the protocol does not abort. When the protocol aborts, no bitstring $K$ is generated, and we label $K=\,\perp$. If $\Omega$ denotes the event that the ideal protocol does not abort (its complement $\Omega^c$ is the event that the protocol aborts), then the joint probability distribution of $K\snapshot$ after an ideal protocol, represented by a density matrix, is of the form
\begin{align}
    \rho^{\rm ideal}_{K\snapshot}=\tau_K\otimes\rho_{\snapshot\wedge\Omega}+\ketbra{\perp}_K\otimes\rho_{\snapshot\wedge\Omega^c},
    \label{eq:ideal-state-security}
\end{align}
where $\tau_K$ is the maximally mixed state representing a uniform distribution of strings of length $|K|$.
In the above, we used the notation $\rho_{XA\land\Omega}=\sum_{x\in\Omega}p(x)\dyad{x}\otimes\rho_A^x$
for the subnormalized conditional states of the classical-quantum state $\rho_{XA}=\sum_{x}p(x)\dyad{x}\otimes\rho_A^x$ conditioned on the event $\Omega$.
For simplicity, we label $\land\Omega$ only at the end, e.g. we label $(\tau_K\otimes\rho_{I_{sn}})_{\land\Omega}$ as $\tau_K\otimes\rho_{I_{sn}\land\Omega}$.

In practice, the output $K$ of a protocol is  close to the ideal state in Eq.~\eqref{eq:ideal-state-security} only up to some security parameter. We refer to this distance as the soundness or the security parameter. In Eq.~\eqref{eq:ideal-state-security}, an ``abort'' event always leads to $K = \perp$. Therefore, it suffices to consider the distance from ideal in the event the protocol does not abort. We can formally state the security definition as follows.
\begin{definition}[Soundness of certified randomness protocol]\label{def:soundness}
    Let $\epssou\in(0,1]$. 
    A randomness certification protocol is $\epssou$-sound if for an honest classical client and any server, 
    \begin{align}
        \norm{\rho_{K\snapshot\wedge\Omega}-\tau_K\otimes\rho_{\snapshot\wedge \Omega}}_{\Tr}\leq \epssou,
    \end{align}
    where $\norm{\cdot}_{\Tr}$ is the trace distance measure.
\end{definition}

\subsection{XEB Preliminaries}\label{sec:preliminaries-xeb}

In our protocol, the client estimates the amount of entropy received by measuring the $\XEB$ score~\cite{Boixo2018} over a set $\V$ of $m=|\V|$ challenge circuits:
\begin{eqnarray}
    \XEBTest = \frac{N}{m} \sum_{i \in \V\subset [M]} p_{C_i}(x_i) - 1,
    \label{eq:xeb-def}
\end{eqnarray}
where $N=2^n$ and $p_{C_i}(x_i)$ is the probability amplitude $|\langle x_i|C_i|0\rangle|^2$. At a high level, if the samples $x_i$ are uniform, the $\XEB$ score is expected to be zero. On the other hand, if the server is using a perfect quantum computer to honestly sample from the corresponding quantum state, the $\XEB$ score is expected to be one. Similarly, a finite-fidelity quantum computer produces a $\XEB$ score between 0 and 1. In the following subsections we characterize the distribution of the $\XEB$ score for differently sampled bitstrings.

\subsubsection{Distribution of the XEB score of uniformly sampled bitstrings}

A random $n$-qubit quantum state $\ket{\psi}$ induces a probability distribution on bitstrings $x \in \{0, 1\}^{n}$. The probabilities $p(x)$
for a random state are assumed to follow
the Porter--Thomas distribution with a frequency density $f(p)$:
\begin{eqnarray}
    f(p) = N \cdot e^{-N\cdot p},
    \label{eq:porter-thomas-pdf}
\end{eqnarray}
where $N= 2^n$ is the dimension of the Hilbert space. 

First, consider a bitstring $x \in \{0, 1\}^{n}$ that is selected uniformly randomly (which we refer to as \textbf{uniform sampling}). Denote by $p(x)$ its measurement probability: $p(x) = \text{Tr}(\dyad{x}\rho)$, with $\rho \sim \mathcal{H}(N)$ drawn from ensemble of Haar-random states over $n$ qubits. Even though $x$ is sampled uniformly at random from $\{0, 1\}^n$, its corresponding probability, $p(x)$, follows an exponential distribution with the probability density function (PDF) given by
\begin{eqnarray}
  \text{PDF}_{x\sim \U}[p(x)]=f(p(x))=N \cdot e^{-N\cdot p(x)}.
  \label{eq:pdf-uniform-sampling}
\end{eqnarray}
The sum of $m$ independent random variables drawn from exponential distribution above is given by the Erlang distribution of ``shape" $m$, that is, $ \sum_{i \in \V\subset [M]} p_{C_i}(x_i) \sim \text{Erlang}(m, N)$. As a result, for uniform sampling, the cumulative distribution function (CDF) of the $\XEB$ score, denoted as $\XEB_{m,\U}$, is given by
\begin{eqnarray}
    \Pr\left( \sum_{i \in \V\subset [M]} p_{C_i}(x_i) \leq (\chi+1)\cdot m/N \right) = \tilde{\Gamma} (m, m \cdot (\chi+1)) \implies  \Pr\left( \XEB_{m, \U} \leq \chi \right) = \tilde{\Gamma}(m, m \cdot (\chi+1)),
\end{eqnarray}
where we use $\tilde{\Gamma}$ to denote the regularized lower-incomplete Gamma function \cite{dlmf}.

\subsubsection{Distribution of the XEB score of bitstrings perfectly sampled from the quantum state}

Second, consider a bitstring $x$ sampled from the distribution corresponding to measurement outcomes of a random quantum state 
$\rho = \dyad{\psi}$ (which we refer to as \textbf{quantum sampling}). The PDF of $p(x)$ is obtained by multiplying the frequency density $f(p(x))$ by the probability of actually measuring the bitstring $x$, $p(x)$ and by introducing an extra $N$ parameter that renormalizes the distribution:
\begin{eqnarray}
     \text{PDF}_{x \sim \Q}[p(x)]=N \cdot f(p(x)) \cdot \text{Tr}\left(\dyad{x}\rho\right)  = N^2 \cdot e^{-N \cdot p(x)} \cdot p(x).
     \label{eq:pdf-quantum-sampling}
\end{eqnarray}
The sum of $m$ independent random variables with the above PDF  given by the Erlang distribution of ``shape" $2m$,  gives us a CDF on the \XEB{} score for quantum sampling, denoted as $\XEB_{m, \Q}$:
\begin{eqnarray}
     \sum_{i \in \V\subset [M]} p(x_i) \sim \text{Erlang}(2 \cdot m, N)\implies \Pr\left( \XEB_{m, \Q} \leq \chi \right) = \tilde{\Gamma}(2\cdot m,  m \cdot (\chi+1)).
\end{eqnarray}

\subsubsection{Distribution of the XEB score of bitstrings sampled from a mixture}
If the verification set $\V$ consists of $l$ quantum samples and $m-l$ uniform samples, the sum of probabilities of the quantum samples follows the distribution $\text{Erlang}(2\cdot l, N)$, and the sum of probabilities of the uniform samples has a distribution $\text{Erlang}(m-l, N)$.  With a known property of sums of the Erlang distribution \cite{thomopoulos2017statistical}, the sum of all probabilities therefore has a distribution $\text{Erlang}(m+l, N)$. The CDF of the \XEB{} score, denoted as $\XEB_{m, l}$, is given by
\begin{eqnarray}
     \Pr\left( \XEB_{m, l} \leq \chi \right) = \tilde{\Gamma}(m + l,  m \cdot (\chi+1)).
\end{eqnarray}

\subsubsection{Distribution of the XEB score of bitstrings obtained by finite-fidelity quantum sampling}\label{sec:finite-fidelity-sampling}

The effect of a broad class of noise channels on the $\XEB$ score is the same as that of depolarizing noise \cite[Appendix A]{morvan2023phase}. Therefore, we model any finite-fidelity state as a mixture between the ideal state and the maximally mixed state: $\rho_{\phi} = \phi \cdot \dyad{\psi} + (1-\phi) \cdot \mathbb{I}/N$, where $\mathbb{I}$ is the identity matrix. Strictly speaking, $1-\phi$ is the depolarizing parameter whose corrections to fidelity, ignored in this work, are exponentially small in the number of qubits. The PDF  $p(x)$ from such a finite-fidelity state is given by
\begin{align}
     \text{PDF}_{x \sim \Q_{\phi}}[p(x)]&=N \cdot f(p(x)) \cdot \text{Tr}\left(\ket{x}\bra{x}\rho_{\phi}\right) \nonumber \\
     &= N^2 \cdot e^{-N \cdot p(x)} \cdot (\phi \cdot p(x) + (1-\phi)/N) \nonumber \\ 
     &= \phi \cdot N^2 \cdot e^{-N \cdot p(x)}  \cdot p(x) + (1-\phi) \cdot N \cdot e^{-N \cdot p(x)} \nonumber \\
     &= \phi \cdot   \text{PDF}_{x \sim \Q}[p(x)] + (1-\phi) \cdot  \text{PDF}_{x \sim \U}[p(x)]\label{eq:pdf-finite-fidelity}.
\end{align}
This suggests that the PDF of a fidelity-$\phi$ sample, which may be generated by sampling a noisy quantum computer or by performing approximate classical simulation, can be interpreted as a stochastic mixture of quantum sampling (Porter--Thomas) and uniform sampling: a bitstring is sampled from the Porter--Thomas distribution with probability $\phi$ and from the uniform distribution with probability $(1-\phi)$. The probability that $\ell$ out of $m$ bitstrings are sampled from the Porter--Thomas distribution is given by the binomial distribution. Therefore, the CDF of finite-fidelity $\XEB$, denoted below as $\XEB_{m, \phi}$, is given by
\begin{align}
\Pr\left(\XEB_{m,\phi} \leq \chi \right)
&= \sum_{l=0}^m \Pr\left(\XEB_{m,\phi} \leq \chi |\:l  \text{ PT bitstrings}\right)\cdot \Pr \left(l \text{ PT bitstrings} \right) \\
&=\sum_{l=0}^{m} \tilde{\Gamma}(m+l, m \cdot (\chi + 1)) \cdot \left[\binom{m}{\ell}\phi^{l}(1-\phi)^{m-l }\right].\label{eq:cdf_honest_server}
\end{align}

\subsection{Adversarial Model and Assumptions}\label{sec:finite-size-assumption}

\begin{figure}[h]
    \includegraphics[width=0.7\textwidth]{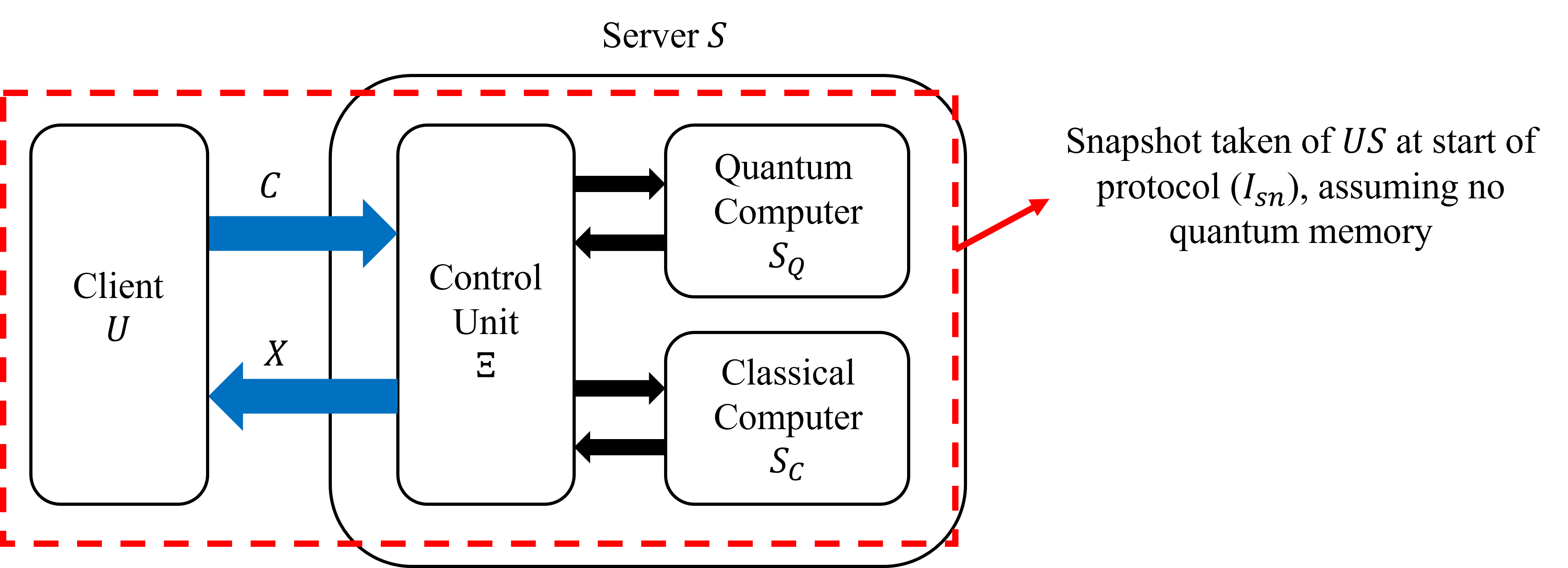}
    \caption{\label{fig:model_finite}%
    Model of classical client and malicious server in the certified randomness protocol. The server is split into a classical control unit $\Xi$, a classical computer $\SC$, and a quantum computer $\SQ$. Without assuming authenticated communication, the server $S=\tilde{S}E$ in general refers to the actual server $\tilde{S}$, along with all parties $E$ that have access to the communication channel between the client and server.}
\end{figure}

The protocol involves a classical client $U$, which interacts with a server $S$ with quantum capabilities.
The client is assumed to be honest, while the server is treated as the adversary.
Let us start by introducing a model of the server $S$, which consists of a control unit $\Xi$, classical computer $\SC$, and quantum computer $\SQ$, illustrated in Fig.~\ref{fig:model_finite}.
Without any restriction on the adversarial server, it is clear that any classical server can always pass the \XEB{} test by simulating the circuits.
However, it is widely believed that the task of random circuit sampling, commonly used for ``quantum supremacy'' demonstrations, is hard to perform classically with a reasonable computational power ~\cite{Arute2019,morvan2023phase,Zhu2022,qntm_rcs,aaronson2017complexity,Aaronson2020,Bouland2018,Bouland2022}.
Consequently, by limiting the computing power of the adversary, we can demonstrate that the protocol is sound.

\subsubsection{Circuit sampling and hardness assumptions}\label{sec:security_assumptions_rcs}
The first set of assumptions we make relate to the properties of random circuit sampling and the difficulty in performing random circuit sampling classically, which are widely used in the aforementioned ``quantum supremacy'' demonstrations.
\begin{itemize}
    \item The output probabilities of any circuit $C$ belonging to the circuit family follow the Porter--Thomas distribution given in Eq.~\ref{eq:porter-thomas-pdf}.
    \item Achieving high \XEB{} classically is as hard as computing the output probabilities of the quantum circuit for the chosen circuit family;
\end{itemize}
For the first assumption, extensive numerical evidence exists showing that the output probabilities of random quantum circuits closely follows the Porter--Thomas distribution~\cite{Arute2019,Boixo2018}. 
For the particular class of circuits considered in this work (see \Cref{sec:protocol_circuit_selection}), we show that each circuit generates a distribution close to Porter--Thomas in total variation distance (TVD) (see Fig.~\ref{fig:convergence-pt}A), with TVD exponentially decaying in $n$. Additionally, the Shannon entropy of the probability distribution of the random quantum circuits closely matches that of the Porter--Thomas distribution, with the difference vanishing exponentially as well (see Fig.~\ref{fig:convergence-pt}B).

The hardness of spoofing $\XEB$ is closely related to the hardness of estimating the output distributions of quantum circuits. It has been established that an efficient classical algorithm to estimate output amplitudes of a Haar-random circuit consisting of $m = \poly(n)$ gates to an additive precision $2^{-O(m)}$ would lead to the collapse of the polynomial hierarchy~\cite{krovi2022average, Bouland2022}. On the heels of RCS experiments, it was conjectured with the XQUATH (linear cross-entropy threshold assumption) that even estimating the output probabilities of quantum circuits better than the trivial guess of $2^{-n}$ is difficult, implying the hardness of \XEB{}. 

For noisy circuits, Ref.~\cite{gao2024limitations} pointed out the limitations of XEB as an evidence of quantum advantage. While prior work implicitly assumed XEB to be a measure of fidelity,  it was shown that there are regimes of noise where XEB is no longer a good measure of fidelity and that one can achieve a high XEB despite having a low fidelity (or high noise). Theoretical as well as experimental investigations have since \cite{ware2023sharp, morvan2023phase} revealed a sharp phase transition between the weak-noise regime where $\XEB$ tracks fidelity and the strong-noise regime with a discrepancy between the two, with the low-noise regime characterized as $\epsilon \cdot n  < \ln 3 \approx 1.1$ for our architecture \cite{ware2023sharp}. We obtain the value of $\epsilon\cdot n$ for our experiment as follows. We take the simple gate-counting model from \cite[Eq. 11]{qntm_rcs}, which estimates the fidelity as $\phi_{\rm GC}(n,d) = \bigl(1 - \varepsilon(n) \bigr)^{nd / 2} \bigl(1-p_{\text{SPAM}}\bigr)^n$ where $\varepsilon(n)$ is the effective two-qubit gate process infidelity, inclusive of memory errors (which depend on $n$) and single-qubit gate errors that may be incurred during a layer of circuit execution. Plugging in the estimated $30\%$ overall circuit fidelity in this work obtained from full-circuit mirror-benchmarking experiments, as well as the depth offset of $1.12$ and $p_{\text{SPAM}}=0.00147$ specified in Ref.~\cite{qntm_rcs}, at $n=56$ and $d=10$ we find an effective process infidelity per qubit per layer equal to $\epsilon = \varepsilon(56) / 2 \approx 1-0.99775$.\footnote{Substantial improvements made to the H2-1 processor since collection of this experiment's data have since improved the effective process fidelity per qubit per layer to $1 - \varepsilon(56) / 2= 0.9984$ as reported in Ref.~\cite{qntm_rcs}, not to be confused with the average two-qubit gate fidelity of $0.99843$ also reported in that work. However, $p_{\text{SPAM}}$ has remained approximately the same.} Since $n=56$, this gives $\epsilon \cdot n \approx 0.13$, which is well below the critical point of phase transition.

More recently, Ref.~\cite{aharonov2023polynomial} proposed a classical algorithm that refutes XQUATH for circuits composed of 2-qubit Haar-random gates and sublinear depth. The same algorithm can also produce bitstrings indistinguishable from a noisy quantum circuit using only polynomially many samples, thus weakening the reliability of \XEB{} for constant noise of any strength. However, the algorithm has a runtime that scales $O(M^{1/\varepsilon})$, where $M$ is the number of samples and $\varepsilon$ is the gate infidelity, rendering the algorithm impractical.

Importantly, asymptotic statements do not prove or disprove the difficulty of spoofing finite-sized experiments. The best-known efficient classical algorithms for spoofing XEB  either are impractical for realistic experiments or produce an XEB score much lower than that from experiments.  Adversaries therefore must resort to inefficient means such as tensor network contraction and approximate circuit simulation. Indeed, the successful spoofs of \XEB{} have come from improved tensor network contraction---with classical computing clusters obtaining an \XEB{} orders of magnitude larger than experiments on a 53-qubit computer. However, we believe that such an attack is unlikely for our experiment on a 56-qubit machine and our measured $\XEB{} \approx 0.32$ based on extensive evidence of hardness from \cite{qntm_rcs}, %
which considers the same circuit family.

\begin{figure}[!h]
\includegraphics[width=\textwidth]{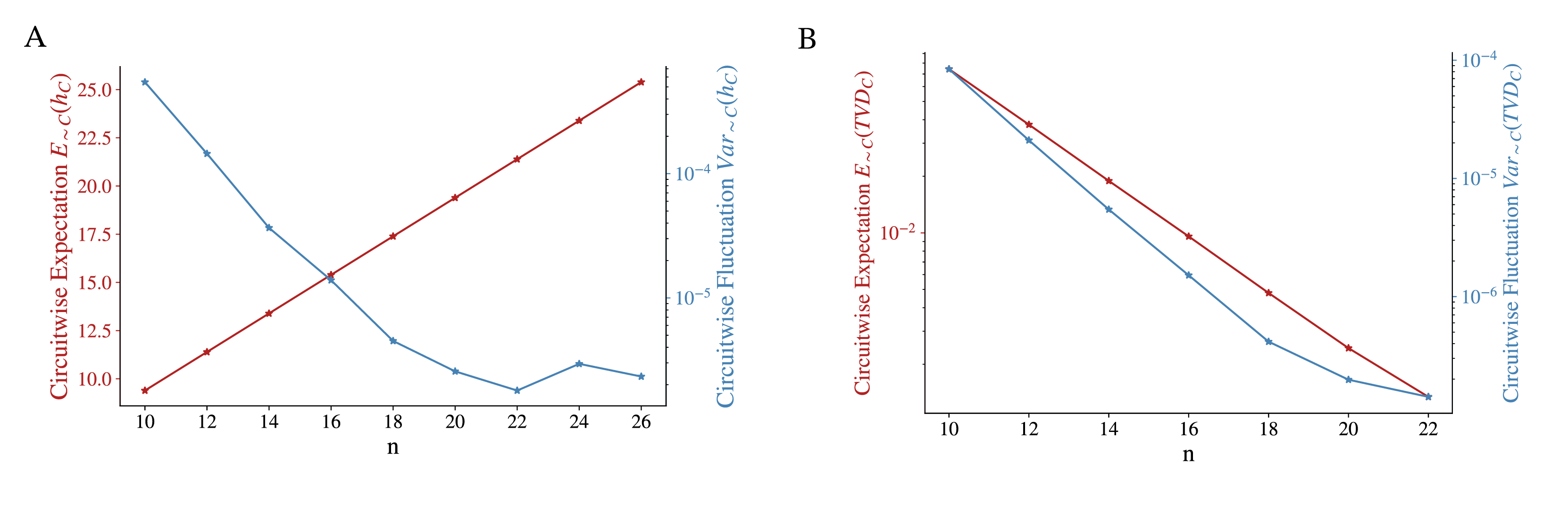}\caption{ \textbf{Numerical evidence of convergence to Porter--Thomas} (A) Circuit-wise expectation (blue) and variance (red) of Shannon entropy, defined as $-\sum_{x \in \{0, 1\}^{n}} p(x) \log p(x)$ with $p(x) = |\langle x|C|0\rangle|^2$, for distributions induced by $d=10$ circuits $C$ with different $n$. (B) Circuit-wise expectation (blue) and variance (red) of Total Variation Distance (TVD) from Porter-Thomas distributions  induced by $d=10$ circuits with different $n$. To compute the TVD, we convert the vector or probabilities $p(x)$ into a histogram, counting frequencies over discretized probability intervals. Then, we sum the difference in frequency counts so obtained with those expected from Poter-Thomas distribution. For both plots, each data point summarizes 1,000 realizations of random circuits over a fixed two-qubit topology obtained via edge coloring of an $n$-node graph. Standard error is too small to be visible on the plot.}
\label{fig:convergence-pt}
\end{figure}

Consequently we believe that the frugal rejection sampling \cite{markov2018quantum} procedure along with our tensor network contraction simulation method with finite fidelity represents the best-known method that an adversary can utilize, informing the choice of our adversarial model.

\subsubsection{Assumptions on computing devices}\label{sec:security_assumptions}

Following the model in Fig.~\ref{fig:model_finite}, we limit the power of the components in the server $S$, which we recall is composed of a quantum computer $\SQ$, a classical computer $\SC$, and a control unit $\Xi$. We remind the reader that this set of assumptions is in general imposed on the joint system of the actual server $\tilde{S}$ and any parties that can access the communication channel $E$, since we do not make the assumption of an authenticated communication channel. The hardness assumption informs our choice of assumption on the power of the classical computer $\SC$:
\begin{itemize}
    \item The classical computer is capable of $\A$ FLOPS of peak performance.
    \item The adversary possesses the same methods of tensor network contraction as the client, and that the adversary's classical methods are as equally performant (e.g., precision, numerical efficiency of tensor network contraction) as the client's.
    \item The classical computer can  perform only frugal rejection sampling, which we summarize in Fig.~\ref{fig:frugal_rejection_sampling}. More concretely, it can accept an input circuit $C_i$ and target fidelity $\phi_{\A}^{(i)}$ from $\Xi$ and run the sampling procedure at the target fidelity.
\end{itemize}
We note that any violation of the second assumption can be captured by increasing the power $\A$ in the first assumption: An adversary possessing more powerful classical methods (higher numerical efficiency given the same contraction scheme, the ability to find contraction schemes with lower costs, the ability to use lower precision, etc.) than the client may be described by absorbing its performance gain into its computational power $\A$. 

The output probabilities of the frugal rejection sampling procedure in Fig.~\ref{fig:frugal_rejection_sampling} \rev{are assumed to} follow the same distribution as \rev{quantum} finite-fidelity sampling in Sec.~\ref{sec:finite-fidelity-sampling}. \rev{We provide evidence for this assumption in \Cref{sec:frugal_pdf_assumption}.} However, recent work \cite{zhao2024leapfrogging} argues that one can boost the \XEB{} score of classically obtained samples by postselecting from a large number of low-fidelity samples, leading to a seven-fold improvement for the 2019 Sycamore circuits \cite{Arute2019} compared with previous techniques. This effect can be approximately incorporated by absorbing a factor into the adversary computational power $\A$. We did not investigate the improvement factor that could be associated with our experiment. Further, the probability distribution of the \XEB{} score of such an adversary is not understood. Therefore, we leave the analysis of such an adversary for future work.

\begin{figure}[!h]
    \hrule\vspace{.5em}
    Input: A single quantum circuit $C_i$, target classical simulation fidelity $\phi_{\A}^{(i)}$.\\ \vspace{.5em}\begin{flushleft}Algorithm:\end{flushleft}
  \begin{enumerate}
      \item Choose a sufficiently large $M'$.
      \item Sample a subset of distinct $M'$ bitstrings $\{x_j\}$ uniformly at random.
      \item Calculate all $M'$ probabilities at once by contracting a fraction $\phi_{\A}^{(i)}$ of all slices of the tensor network.
      \item For each $j$: accept $x_j$ with probability $\min\left(1, p(x_j)N/M'\right)$.
  \end{enumerate}
  \begin{flushleft}Output: The first accepted bitstring $x_j$.
  \end{flushleft}
  \hrule\vspace{1em}
  \caption{The frugal rejection sampling algorithm.}
\label{fig:frugal_rejection_sampling}
  \end{figure}

It is difficult to analyze the security of the protocol if the adversary is allowed to perform arbitrary operations on the quantum computer.
Thus, we assume a restricted quantum computer $\SQ$.
\begin{itemize}
    \item The quantum computer is only allowed to execute the circuit, obtaining state $C_i\ket{0^n}$, and measure it to obtain an $n$-bitstring $X_i$ with perfect fidelity. More concretely, it can accept the circuit $C_i$ from $\Xi$, run the circuit $C_i$, and return the output $X_i$.
    \item Every round using the quantum computer (quantum round) is i.i.d.
\end{itemize}
The second restriction is implied by the first restriction since it specifies exactly how the quantum computer must be used. We also make an assumption of the circuits that are submitted to the quantum computer:
\begin{itemize}
    \item The choice of circuit for every quantum round is i.i.d. 
\end{itemize}
This assumption is valid when the client selects the circuits in an i.i.d. fashion and  $\Xi$ is not choosing which circuits to return quantum samples or which batches to fail in a way that depends on the circuits \rev{(no post-selection). We believe assuming no postselection is reasonable for
near term adversaries. The argument for this is as follows. For tensor network contraction-based approaches considered for our adversary, the hardness is identical across circuits since the hardness only depends on the circuit topology which stays the same, and only the single-qubit gates change. Additionally, there is no known method for the adversary to say anything about the quality of the sample without running the classical simulation and examine the bitstring probability. If it runs the classical simulation, it might as well use the classical compute budget to provide a classical sample instead of determining post-selection. However, in the future, it will be beneficial to go beyond heuristic arguments and conduct a more rigorous analysis on the possibility of using postselection and understand how much impact that may have.}

We note that the quantum computer utilized in the experiment (the honest case) is imperfect and is only able to sample circuits with fidelity $\phi\approx 0.3$. %
In general, it is possible to assume $\SQ$ to be a quantum computer with ability to sample with fidelity \textit{tunable} in $\phi\in[0.3, 1]$. For the purpose of \XEB{} calculations, this corresponds to a probability $\phi$ of drawing a sample from the PT-distribution and $1-\phi$ of drawing a sample from the uniform distribution. 
Since it is straightforward to replace the uniform distribution sample with a known bitstring from its memory device, the entropy contribution from this component in the worst case, conditioned on the snapshot $\snapshot$, is 0. In contrast, the samples obtained with simulation in the classical computer $\SC$ have a higher $\phi_{\A}^{(i)}$ \XEB{} score contribution also with zero conditional entropy. Since, at a fixed \XEB{}, the conditional entropy is expected to be proportional to $\phi \cdot Q$ (where $Q$ is the number of samples executed on the quantum computer), the server's optimal strategy (for maximizing the average of \XEB{}) is to use the classical computer as much as possible together with the highest-fidelity quantum computer possible (minimum $Q$). We assume $\phi=1$ in our analysis for simplicity as a worst-case assumption.

\rev{

\subsubsection{Assumption on frugal rejection sampling}\label{sec:frugal_pdf_assumption}

For the purposes of our calculations, it suffices to show that the probabilities resulting from bitstrings obtained via partial-contraction of tensor networks in which we contract $\phi$ fraction of slices and the probabilities corresponding to a bitstring obtained by sampling a state depolarized by strength $(1-\phi)$ are equivalent. To numerically emulate the classical sampling process, we consider a tensor network representation of a $n$-qubit quantum circuit $C$, contract $k = \phi \cdot K$ slices, and use the partial amplitudes to sample bitstrings through the frugal rejection sampling process.

Sampling from a depolarized finite fidelity quantum state has the PDF given by Eq. \eqref{eq:pdf-finite-fidelity}. However, it is difficult to rigorously derive the PDF associated with samples obtained from partial tensor network contraction. Nevertheless, we numerically observe (see Fig. \ref{fig:frugal_pdf_comparison}) that such classical samples have probabilities which closely follow Eq.~\eqref{eq:pdf-finite-fidelity}. This suggests that, for the purposes of our XEB analysis, finite-fidelity classical simulation via partial contraction of tensor networks is equivalent to sampling from a noisy depolarized state.

}

\begin{figure}[!ht]
    \centering
    \includegraphics[width=0.4\textwidth]{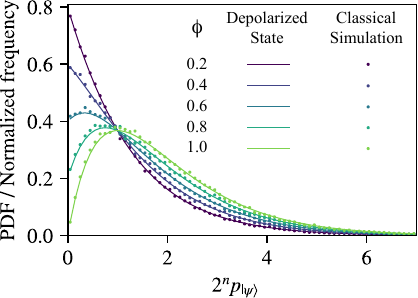}
    \caption{\rev{Distribution of the probability $p_{\vert\psi\rangle}(x)$. A random circuit with depth $8$ acts on $n=20$ qubits. The tensor network representation is sliced into  $K=1024$ total slices of which we only contract $k = \phi \cdot K$ slices. The partial amplitudes are used to sample bitstrings following the frugal rejection sampling process (Fig. \ref{fig:frugal_rejection_sampling}). For each $\phi$, we sample ten thousand samples via partial contraction (dots).} 
    }
    \label{fig:frugal_pdf_comparison}
\end{figure}

\subsubsection{Adversary strategy restriction}

We restrict the actions of the control unit $\Xi$ and thus the strategy of the adversary. We note that the same restricted adversary is commonly considered in the analysis of certified randomness protocols based on random circuit sampling~\cite{brandao2020notes,morvan2023phase}.
In particular, the assumptions are as follows.
\begin{itemize}
    \item We assume $\Xi$ is not allowed to perform any postselection attacks; in other words,  the $M$ detected rounds in the protocol are a fair representation of the adversary's behavior. Therefore, the analysis can be reduced to focus on the $M$ detected rounds;
    \item We assume that $\Xi$ performs a restricted attack: Among the $M$ detected rounds, $\Xi$ \textit{a priori} selects $Q$ rounds, sends all the circuits $C_i$ for those rounds to the quantum computer, and returns the sample $X_i$ that $\SQ$ provides---henceforth called \emph{quantum rounds}. For the remaining $M-Q$ rounds, it sends $C_i$ and $\phi_{\A}^{(i)}$ to the classical computer and returns the sample $X_i$ that $\SC$ provides----henceforth called \emph{classical rounds}. We note here that the target fidelity $\phi_{\A}^{(i)}$ may differ between circuits and is dependent on the runtime $T_i$ of each round, as decided by the control unit $\Xi$;
    \item For each of the $Q$ rounds, it interacts with the quantum computer only once; that is, it does not attempt to request multiple $X_i$ for the same $C_i$.
\end{itemize}
We note here that in general one can consider an arbitrary probabilistic choice of $Q$.
However, since we can describe such an adversary by a convex sum of adversaries with a fixed $Q$, the satisfaction of soundness for all adversaries with fixed $Q$ would imply that the more general adversary remains sound as well. As also noted in Methods, we point out that these assumptions are likely stronger than necessary and may be remedied in the future. 

These assumptions allow us to focus on the $M$ detected rounds, of which $Q$ quantum rounds are a priori selected by the server. A summary of this adversarial model is presented in Fig.~\ref{fig:model_finite} and in the main text Fig.~1D. Essentially, this adversary generates the $Q$ samples using a perfect quantum computer, with perfect-fidelity sampling of the Porter--Thomas distribution. For the rest of the $M-Q$ samples, the adversary performs finite-fidelity simulation using frugal rejection sampling. We assume the outputs of frugal rejection sampling are indistinguishable from that obtained by sampling from a finite-fidelity quantum computer as far as \XEB{} statistics are concerned \rev{which we justify in \Cref{sec:frugal_pdf_assumption}}, which is equivalent to sampling from the Porter--Thomas distribution with probability $\phi_{\A}^{(i)}$ and the uniform distribution with probability $1-\phi_{\A}^{(i)}$.\\

The assumptions listed above are necessary for simplification of the security analysis.
We leave a general analysis of protocol security to future work.
There are some limitations that such a general analysis should address, including (1) the i.i.d. assumption of circuit choice and quantum round, which is inherently at odds with the use of pseudorandom circuit selection, (2) postselection attacks available when we allow batches to be discarded, and (3) oversampling attack.
Some of these limitations, such as oversampling, have been previously analyzed in Refs.~\cite{brandao2020notes,morvan2023phase}, and one can adapt similar solutions in the general security analysis.

\subsection{Bounds on the Entropy Certified by the Protocol}
\label{sec:bound-on-q}

It is known \cite{renner2008security,tomamichel2017largely} that a lower bound on the conditional smooth min-entropy of the samples $X^M$ generated through the protocol is sufficient to prove soundness as per \Cref{def:soundness}.
In this section we present and prove our main results regarding the amount of smooth min-entropy generated by our randomness certification protocol, deferring the particulars pertaining to soundness and randomness extraction to \Cref{sec:soundness}.

The min-entropy of a classical-quantum state $\rho_{XA}=\sum_{x}p(x)\dyad{x}\otimes\rho_A^x$ is defined as $H_{\min}(X|A)_\rho=-\log p_{\rm guess}(X|A)_\rho$,
where $p_{\rm guess}(X|A)_\rho:=\sup_{\{M_x\}_x}\sum_x p(x)\Tr[\rho_A^xM_x],$ with supremum over POVMs on register $A$, is known as the guessing probability.
The min-entropy can be generalized by including a smoothing parameter $\epss\in(0,\sqrt{\Tr[\rho_{XA}]})$ such that 
\begin{align}
    H_{\min}^{\epss}(X|A)_{\rho}:=\sup_\sigma H_{\min}(X|A)_{\sigma}\,,
\end{align}
where the supremum is over states $\sigma_{XA}$ in the $\epss$-ball in purified distance centered around $\rho$, as defined in \cite{tomamichel2015quantum}.

The strategy to prove a lower bound on the smooth min-entropy of the samples generated in our protocol proceeds by first bounding the probability that the $\XEB$ test passes when the server executes $Q$ quantum rounds, for fixed $Q$.
This bound indirectly allows us to determine $\Qmin$, the minimum number of quantum rounds the server has to carry out, from which we can then bound the smooth min-entropy.

Let the client's total computational budget to verify $m$ circuit-sample pairs be $\totalbudget$ floating-point operations. The computational cost of simulating a circuit with perfect fidelity is $\budget=\totalbudget/m$ floating-point operations. 
In our protocol, the adversary returns $M$ bitstrings in the non-discarded rounds with a maximum duration of $\Tthreshold = M \cdot \Tqcthreshold$. Let $\ceff$ denote the numerical efficiency, which is by assumption equal between the client and the adversary.
If the adversary had a large enough classical computer such that $\ceff\cdot\A \cdot \Tthreshold/(M\cdot \budget) \geq 1$, the adversary can simulate all circuits to perfect fidelity and will return classically simulated samples, with confidence that the returned samples will achieve a high $\XEB$ score when verified by the client. On the other hand, if $\ceff\cdot\A \cdot \Tthreshold/(M\cdot \budget) $ is much less than one, the client cannot simulate all circuits to high fidelity and is thereby compelled to return genuine quantum samples in order to achieve a high $\XEB$ score.

The adversary has maximum time $\Tthreshold$ to return $M$ samples, which is the same time it has to simulate $M-Q$ circuits classically. Assuming a simple but realistic linear model of simulation fidelity~\cite{markov2018quantum}, we find that, in the worst case, the maximum sum of the simulation fidelities of the $M-Q$ circuits is given by the total executed FLOP count (efficiency times power times maximum time) divided by $\budget$:
\begin{eqnarray}
 \Phi_{\A}  =   \sum_{i \text{ classically simulated}} \phi_{\A}^{(i)} = \min \left( M-Q, \frac{\ceff\cdot\A \cdot \Tthreshold}{ \budget}\right)=\min \left( M-Q, \frac{\ceff\cdot\A \cdot \Tthreshold}{ \totalbudget/m}\right).\label{eq:phi_A}
\end{eqnarray}
This equation shows that at fixed $M$ and $m$, the protocol performance should remain the same as long as $\A\cdot \Tthreshold/\totalbudget$ is unchanged. For example, an increase in the adversarial power can be canceled out by a reduction in $\Tthreshold$ or an increase in $\totalbudget$. 

As we argue in \Cref{sec:preliminaries-xeb}, finite-fidelity classical bitstrings can be interpreted as being drawn from the Porter--Thomas distribution with probability $\phi_{\A}^{(i)}$ and from the uniform distribution with probability $1-\phi_{\A}^{(i)}$. The total number of Porter--Thomas bitstrings is therefore a random number. Let $Z_i$ be the indicator random variable such that among samples classically simulated, the sample $i$ is drawn from the Porter--Thomas distribution if and only if $Z_i = 1$ and from $Z_i = 0$ otherwise. The random variable $L_\text{C}$ representing the total number of Porter--Thomas bitstrings in $M-Q$ classical samples is therefore
\begin{equation}
    L_\text{C}=\sum_{i\text{ classically simulated}} Z_i,
\end{equation}
and its expected value is, in the worst case, given by
\begin{equation}
    \Exp[L_\text{C}]=\sum_{i \text{ classically simulated}} \phi_{\A}^{(i)}=\Phi_{\A}.
\end{equation}
Recalling that we assume the server is i.i.d., we use Chernoff's inequality for sum of Bernoulli random variables \cite{hagerup1990guided} to obtain an upper bound on the probability that $L_C$ exceeds $L_{\text{C},\rm max}:=(1+\delta)\cdot \Exp[L_\text{C}]$ for some $\delta>0$:
\begin{eqnarray}
    \Pr\left[L_\text{C} \geq L_{\text{C},\rm max} \right] \leq \exp \left( - \frac{\delta^2 \cdot \Exp[L_\text{C}]}{3}\right) = \varepsilon_1 \implies L_{\text{C},\rm max} =   \Phi_{\A}\left(1 + \sqrt{\frac{3}{\Phi_{\A}}\ln \frac{1}{\varepsilon_1}}\right).\label{eq:L_C_max}
\end{eqnarray}
Since all quantum samples have fidelity one and are Porter--Thomas bitstrings, we can obtain an upper bound $L_{\rm max}$ on the random variable representing the total number of Porter--Thomas bitstrings $L$ in \emph{all} $M$ samples (both quantum and classical):
\begin{eqnarray}
    L = Q+L_\text{C}\implies \Pr\left[L\geq L_{\rm max}= Q + L_{\text{C},\rm max}\right]\leq\varepsilon_1.
    \label{eq:pt-concetration-inequality}
\end{eqnarray}
Since the client performs verification only on a much smaller set of circuits $\V$ with $|\V| = m$, the distribution of the $\XEB$ score over the verification set depends on the number of Porter--Thomas samples in the verification set, not the number of Porter--Thomas samples in total. Each of the $m$ verification samples is drawn from a total of $M$ samples without replacement, and up to $L_{\rm max}$ of the $M$ samples follow the Porter--Thomas distribution. 
If $L_{\max}$ out of $M$ samples are Porter--Thomas, then the probability that $l$ out of $m$ samples are Porter--Thomas is given precisely by the hypergeometric distribution $\text{Hypergeometric}(M, L_{\rm max}, m)$ \cite{rice2007mathematical}.

We can now obtain the probability of the adversary achieving measured $\XEB$ score below threshold $\chi$ as
 \begin{align}
\Pr\left[\XEBTest \leq \chi | L =
L_{\rm max}\right]&= \sum_{l=0}^m \Pr\left[\XEBTest \leq \chi |l  \text{ PT bitstrings in }\V\right]\cdot \Pr \left[l \text{ PT bitstrings in }\V | L = L_{ \rm max} \right] \\
&=\sum_{l=0}^{m} \tilde{\Gamma}(m+l, m \cdot (\chi + 1)) \cdot \left[\binom{M}{m}^{-1} {\binom{m}{l}\binom{M-L_{\rm max}}{m-l}}\right] = 1-\varepsilon_2,    \label{eq:cdf_adversary}
\end{align}
 where $\tilde{\Gamma}$ is the regularized lower-incomplete Gamma function that, as we discuss in \Cref{sec:preliminaries-xeb}, gives the distribution of $\XEB$ with a mixture of uniform and Porter--Thomas bitstrings. 
 Here, we remark that Eq.~\eqref{eq:cdf_adversary} holds when the choice of testing subset $\V\subset[M]$ of size $m$ is random. When the choice of $\V$ is pseudorandom, then $\Pr\left[\XEBTest \leq \chi | L =
L_{\rm max}\right]$ is negligibly close to Eq.~\eqref{eq:cdf_adversary}.
 Putting it all together, the probability that an adversary sampling the quantum computer $Q$ times attains a threshold $\chi$ is 
 \begin{align}
 \Pr[\Omega]
    &=\Pr\left[\XEBTest\geq \chi\right]\\
    &=\Pr\left[L>L_{\max}\right]\Pr\left[\XEBTest\geq \chi\,|\,L>L_{\max}\right] \\
    &\quad+ \Pr\left[L\leq L_{\max}\right]\Pr\left[\XEBTest \geq \chi \,|\, L\leq L_{\max}\right]\\
    &\leq \Pr\left[L>L_{\max}\right] + \Pr\left[\XEBTest \geq \chi \,|\, L\leq L_{\max}\right]\\
    &\leq \varepsilon_1 + \varepsilon_2,
\label{omega_xeb}
\end{align}
 with $\varepsilon_1$ and $\varepsilon_2$ as in  Eq.~\eqref{eq:pt-concetration-inequality} and Eq.~\eqref{eq:cdf_adversary}, respectively. 
 Here, the last inequality follows from the fact that for all $l\leq L_{\max}$ $\Pr\left[\XEBTest \geq \chi \,|\, L= l\right]\leq \Pr\left[\XEBTest \geq \chi \,|\, L= L_{\max}\right]$, which consequently implies 
 \begin{align}
     \Pr\left[L\leq L_{\max}\right]\Pr\left[\XEBTest \geq \chi \,|\, L\leq L_{\max}\right] 
     &= \sum_{l\leq L_{\max}} \Pr\left[L=l\right]\Pr\left[\XEBTest \geq \chi \,|\, L=l\right]\\
     &\leq \sum_{l\leq L_{\max}} \Pr\left[L=l\right]\Pr\left[\XEBTest \geq \chi \,|\, L=L_{\max}\right]\\
     &=\Pr\left[\XEBTest \geq \chi \,|\, L=L_{\max}\right]. 
 \end{align}

We summarize this result in the following lemma.
\begin{lemma}\label{lem:q_bound} Let $\Omega$ be the event that an adversary, with a classical computational power of $\A$ and samples $Q$ out of $M$ samples on a perfect-fidelity quantum computer, passes the XEB test with threshold $\chi$. We have $\Pr[\Omega] \leq  \epsadv(Q, \chi) = \varepsilon_1 + \varepsilon_2$, for
\begin{equation}\label{eq:bound_on_q_summarized}
     \varepsilon_1 =  \exp \left( - \frac{\delta^2 \cdot \Phi_{\A}}{3}\right) \qquad  \varepsilon_2 = 1-\sum_{l=0}^{m} \tilde{\Gamma}(m+l, m \cdot (\chi + 1)) \cdot \left[\binom{M}{m}^{-1} {\binom{m}{l}\binom{M-(Q + \Phi_{\A}(1+\delta))}{m-l}}\right],
\end{equation}
$\Phi_{\A}$ given by Eq.~(\ref{eq:phi_A}), and $\delta \geq 0$.
\label{lem:lower-bound-on-q}
\end{lemma}

Given this lemma, an XEB threshold $\chi$, and a target $\epss\in(0,1/4)$, we can compute $\Qmin=\argmin_Q\{\epsadv(Q,\chi)\geq 4\epss\}$, which allows us to bound the smooth min-entropy of the samples $X_M$ conditioned on side-information $\partialsn=\seed S^{0}$, the initial snapshot $\snapshot$ excluding the randomness extractor seed $\ExtSeed$. The prescription of how to solve for $\Qmin$ using this lemma is provided in \Cref{sec:extraction}. Formally we have the following theorem:
\begin{theorem}\label{thm:entropy_bound}
    Let $\Omega$ denote the event where the randomness certification protocol in Figure 4 of Methods does not abort, and let $\sigma$ be the state over registers $X^M$ and $\partialsn$ prior to the randomness extraction phase of the protocol. Given $\epss\in(0,1/4)$, the protocol either aborts with probability greater than $1-4\epss$ or
    \begin{align}
        H_{\min}^{\epss}(X^M|\tilde{I}_{sn})_{\sigma_{\wedge\Omega}} \geq \Qmin(n-1) - \log \frac{1}{\epss},
    \end{align}
    where $\sigma=\sigma_{X^M|\tilde{I}_{sn}}$ is the state prior to the randomness extraction stage and $\Qmin=\argmin_Q\{\epsadv(Q,\chi)\geq 4\epss\}$ for $\epsadv(Q,\chi)$ as in \Cref{lem:q_bound}.
\end{theorem}

\begin{proof}
    Denoting the initial snapshot as $\tilde{I}_{\rm sn}$, consider the smooth-min-entropy $H_{\min}^{\epss}(X^M|\partialsn)_{\sigma_{\wedge\Omega}}$ conditioned on the snapshot.
    Note that the choice of smoothing parameter $\epss$ is valid since $\Tr[\sigma_{X^M\partialsn\wedge\Omega}]\geq 4\epss$ and $0<\epss<2\sqrt{\epss}$ which implies $H_{\min}^{\epss}(X^M|\partialsn)_{\sigma_{\wedge\Omega}}\leq H_{\min}^{2\sqrt{\epss}}(X^M|\partialsn)_{\sigma_{\wedge\Omega}}$.
    The adversary's strategy is to return quantum samples for $Q$ out of $M$ total samples.  For any choice of $Q$, we have 
    \begin{align}
        H_{\min}^{\epss}(X^M | \partialsn)_{\sigma_{\wedge\Omega}}\geq H_{\min}^{\epss}(X_Q | \partialsn)_{\sigma_{\wedge\Omega}} \geq H_{\min}^{\epss}(X_Q | \partialsn)_{\sigma},
    \end{align}
    where the first inequality follows since the registers $X^M$ are classical and the second from \cite[Proposition 10]{tomamichel2017largely}.
    Now, since all of the registers in the last quantity are classical and produced in an i.i.d. manner (once the global event $\Omega$ is removed), we can use \cite[Equation 5]{Tomamichel09_aep} to obtain
    \begin{align}
        H_{\min}^{\epss}(X_Q | \partialsn)_{\sigma} \geq Q \cdot H_2(X_i | C_i)_{\sigma}-\log\frac{1}{\epss} = Q \cdot (n-1) - \log\frac{1}{\epss},
    \end{align}
    where $i$ is an arbitrary choice of index in $[Q]$ (since the $Q$ rounds are i.i.d.) and we have noted  that the only available side information in quantum rounds is the choice of circuit $C_i$ (computable from $\partialsn$).
    Further, $H_2(X):=-\log\sum_{x}p(x)^2$ is the $2$-Renyi entropy (or collision entropy), which for the Porter--Thomas distribution is
    \begin{equation}
        H_2(X_i|C_i)_{\sigma}=-\log_2\left[\int_0^{\infty} N \cdot f(p) \cdot p^2 \cdot \d p\right] = n-1.
    \end{equation}
    {Given $\Pr[\Omega]_{\sigma}\geq 4\epss$, the minimum $Q$ that any adversary strategy can have is precisely}
    \begin{align}
        \Qmin = \min\{Q\,:\, \epsadv(Q, \chi)\geq 4\epss\},
    \end{align}
    with $\epsadv(Q, \chi)$ defined in Lemma~\ref{lem:lower-bound-on-q}. This immediately gives 
    \begin{align}
        H_{\min}^{\epss}(X^M | \partialsn)_{\sigma} \geq \Qmin (n-1) - \log\frac{1}{\epss},
    \end{align}
    as required.  
\end{proof}

\subsection{Proof of Protocol Soundness}\label{sec:soundness}
Recalling Definition~\ref{def:soundness}, to prove $\epssou$-soundness for some $\epssou\in(0,1]$, we want to show that
\begin{align}
    \norm{\rho_{K\snapshot\wedge\Omega}-\tau_K\otimes \rho_{\snapshot\wedge\Omega}}_{\Tr}\leq \epssou\,,
\end{align}
where $\Omega$ is the event that the protocol does not abort, $K$ is the register containing the $\ell$-bitstring output by the seeded extractor, and $\snapshot$ is the snapshot, composed of the client's (uniform and secret) extractor seed $K_{\rm ext}$, random bitstring $\seed$, and the snapshot of the server's memory $S^{0}$ \revtwo{(which, by definition, includes both server and the environment)}. Further, $\rho_{K\snapshot\wedge\Omega}$ is the state at termination of the randomness certification protocol and $\tau_K\otimes \rho_{\snapshot\wedge\Omega}$ is the ideal functionality of the protocol, which replaces $\rho_K$ with the maximally mixed state $\tau_K=\frac{1}{2^\ell}\sum_{i\in[2^\ell]}\dyad{i}$.

We begin by defining the class of randomness extractors used in our randomness certification analysis: quantum-proof strong extractors \cite{renner2008security, konig2011sampling,de2012trevisan, tomamichel2017largely, foreman2024cryptomite}.

\begin{definition}[Quantum-proof strong extractor \cite{renner2008security, konig2011sampling,de2012trevisan}]
\label{sup_def:ext}
A function $\Ext:\{0,1\}^{n} \times \{0,1\}^s \rightarrow \{0,1\}^{\ell}$ is a quantum-proof strong $(\kappa,\epsext)$-extractor if, for any classical quantum state $\rho_{SE}$ where $S$ is the classical register with dimension $2^{n}$ for which $H_{\min}^{\epss}(S|E)_{\rho} \geq \kappa$, it holds that
\begin{equation}
\label{supp_eq:ext}
\norm{\Ext(\rho_{SE} \otimes \tau_D) - \tau_{K} \otimes \tau_{D} \otimes \rho_E }_{\Tr} \leq \frac{1}{2}\epsext + 2\epss,
\end{equation}
where $\tau_D$, and $\tau_K$ are maximally mixed state with dimension $2^s$ and $2^{\ell}$ respectively and $\epss\in\big[0,\sqrt{\Tr[\rho_{SE}]}\big]$. The map $\Ext$ acts on the classical systems $S$ and $D$. The input on system $D$ is called the seed of the extractor.
\end{definition}

In our protocol, the $S$-register is precisely the $M$ classical registers $X_1, \ldots , X_M$, and the $D$ register is $K_{\rm ext}$ and contains the $s$-bit random seed of the extractor that is private to the client. The output length of the extractor depends on the amount of entropy of the input as well as $\varepsilon_{\text{ext}}$ and the seed length. For example, if we use a $2$-universal strong extractor, then we get the following bound on $\ell$.

\begin{lemma}[$2$-Universal Strong Extractor \cite{renner2008security}]
\label{sup_thm:toep}
There exist a quantum-proof strong $(\kappa,\epsext)$-extractor $\Ext:\{0,1\}^{nM}\times\{0,1\}^s\rightarrow\{0,1\}^\ell$ with seed length $s = nM$ and $\ell\leq\kappa - 2\log(\frac{1}{\epsext})$.
\end{lemma}

Alternatively we can use strong extractors requiring a shorter seed at the cost of a smaller output.  

\begin{lemma}[Trevisan Extractor \cite{foreman2024cryptomite}]
\label{sup_thm:trev}
There exist a quantum-proof strong $(\kappa,\epsext)$-extractor $\Ext:\{0,1\}^{nM}\times\{0,1\}^s\rightarrow\{0,1\}^\ell$ with seed length $s=O(\log (nM))$ and $\ell\leq\kappa - 4\log(\frac{1}{\epsext}) - 4\log \ell -6$.
\end{lemma}

In the following, we show that the smooth min-entropy bound provided in \Cref{thm:entropy_bound} can be used to guarantee $\epssou$-soundness when using a strong randomness extractor.

\begin{corollary}\label{cor:soundness}
    Let $\epssou\in(0,1]$, and suppose that a $(\kappa,\epssou)$-quantum-proof strong extractor is used in the randomness extraction step in the randomness certification protocol, where
    \begin{equation}
        \kappa=\Qmin(n-1) - \log \frac{1}{\epssou} - 2
    \end{equation}
    and where $\Qmin=\min\{Q\,:\,\epsadv(Q,\chi)\geq\epssou\}$.
    Then, the protocol is $\epssou$-sound.

    In particular, the protocol is $\epssou$-sound if a two-universal hash function (e.g., Toeplitz hashing) is utilized and the length of the output satisfies
    \begin{align}
        \ell \leq \Qmin(n-1) - 3\log\frac{1}{\epssou} - 2.
    \end{align}
    Alternatively, the protocol is $\epssou$-sound if a Trevisan extractor is utilized and the length of the output $\ell$ satisfies
    \begin{align}
        \ell\leq \Qmin(n-1) - 5\log \frac{1}{\epssou} -4\log\ell - 8.
    \end{align}
\end{corollary}
\begin{proof}
    We begin by considering the trace distance $\norm{\rho_{KK_{\rm ext}\partialsn\wedge\Omega}-\tau_K\otimes\tau_{K_{\rm ext}}\otimes\rho_{\partialsn\wedge \Omega}}_{\Tr}$, where $\partialsn$ is the initial snapshot $\snapshot$ excluding the randomness extractor seed $\ExtSeed$. In the case where $\Pr[\Omega]_\rho<\epssou=4\epss$,
    \begin{align}
        \norm{\rho_{KK_{\rm ext}\partialsn\wedge\Omega}-\tau_K\otimes\tau_{K_{\rm ext}}\otimes\rho_{\partialsn\wedge \Omega}}_{\Tr} \leq \frac{1}{2}\norm{\rho_{KK_{\rm ext}\partialsn\wedge\Omega}}_{1} + \frac{1}{2}\norm{\tau_K\otimes\tau_{K_{\rm ext}}\otimes\rho_{\partialsn\wedge \Omega}}_1= \Pr[\Omega]_\rho<\epssou,
    \end{align}
    and soundness is automatically guaranteed.
    In the case where $\Pr[\Omega]_\rho \geq\epssou=4\epss,$ \Cref{thm:entropy_bound} implies
    \begin{align}
        H_{\min}^{\epss}(X^M|\partialsn)_{\sigma_{\wedge\Omega}} \geq \Qmin(n-1) - \log \frac{1}{\epss} = \Qmin(n-1) - \log \frac{1}{\epssou} - 2.
    \end{align}
    Setting $\kappa = \Qmin(n-1) - \log \frac{1}{\epssou} - 2$, \Cref{sup_def:ext} of $(\kappa, \epssou)$-quantum-proof strong extractors implies
    \begin{align}
        \norm{\rho_{KK_{\rm ext}\partialsn\wedge\Omega}-\tau_K\otimes\tau_{K_{\rm ext}}\otimes\rho_{\partialsn\wedge \Omega}}_{\Tr} \leq \frac{1}{2}\epssou + 2 \epss = \epssou,
    \end{align}
    provided the length $\ell$ of the output string $K$ satisfies either
    \begin{align}\label{eq:ext_toep}
        \ell \leq \kappa - 2\log\frac{1}{\epssou},
    \end{align}
    as per \Cref{sup_thm:toep} if a $2$-universal extractor is used, or 
    \begin{align}\label{eq:ext_trev}
        \ell\leq\kappa - 4\log\frac{1}{\epssou} - 4\log \ell -6,
    \end{align}
    as per \Cref{sup_thm:trev} if a Trevisan extractor is used. This concludes the proof.
\end{proof}

\subsection{Randomness Expansion}
As introduced in Methods, the protocol generally requires some initial randomness, and we ideally like the output of the protocol (the extracted bits together with the seed used by the extractor) to be larger than the required randomness input.
Here, we examine in detail the amount of randomness expansion of a single protocol run, although the actual expansion rate also depends on assumptions on the source of initial randomness and the output string $K$.
In the worst case, the expected increase in randomness can be given by
\begin{equation}
    \ell_{\rm expand}=\Pr[\Omega](\ell-\abs{\ExtSeed})-r
\end{equation}
with the terms being the following:
\begin{enumerate}
    \item The expected length of random bitstring generated, $\Pr[\Omega]\ell$, noting that the protocol returns $K=\perp$ when the protocol aborts;
    \item The expected length of the randomness extractor seed utilized, $\Pr[\Omega]\abs{\ExtSeed}$, noting that the randomness extraction is not performed when the protocol aborts;
    \item The length of seed required to generate the circuits and to select the test rounds, $r$. We note that this randomness is utilized even when the protocol aborts.
\end{enumerate}

The assumptions on the sources of the random string can impact the terms to consider in $\ell_{\rm expand}$.
For instance, the source of random strings can be from sources that are not certified, but trusted to be random and inaccessible to the adversary (e.g., public randomness beacon announcing random bits after the server's response).
The consideration here would be that the client sacrifices randomness that the client is unable to certify in order to gain randomness that is certified.
Whether this assumption is suitable depends on the context of the protocol usage.\\

Separately, we note that the extractor seed cannot be reused in general, unlike in some QKD protocols where the seed is preshared (not announced).
The reason is that the input to the extractor, $X^M$, is public and known to the server. Since $K$ is a deterministic function of $X^M$ and $\ExtSeed$, the usage of $K$ could result in reveal of the seed $\ExtSeed$.
If we assume that $K$ is never utilized in any public manner, in other words, is used only internally by the client for applications that would not leak $K$ externally, $\ExtSeed$ remains private and can be reused. This is the assumption we use in the main text of the paper to claim randomness expansion.

\section{Details of Protocol Implementation}

We now provide additional details on the engineering challenges and implementation choices associated with our experiment. 

\subsection{Quantum Circuit Simulation by Tensor Network Contraction}\label{sec:preliminaries_simulation_by_TN}

A tensor is a collection of indexed numbers, which can graphically be represented by a node with an edge or ``bond" leaving the node for every index. A tensor network is a graph representation of an equation corresponding to the product of many tensors, where edges/bonds shared between nodes in the graph correspond to dummy indices that are summed over in the product. Open indices that remain unsummed can  belong only to one tensor and typically represent physical degrees of freedom in the context of quantum simulation. The process of calculating the output of a tensor network by iteratively summing over each dummy index is called ``contraction" of the tensor network.

Tensor networks have a wide range of applications including quantum simulation, quantum computation, quantum control, and machine learning~\cite{NIPS2016_5314b967,Ors2019}. In the context of quantum simulation, there are two leading algorithms for simulating random circuits, and both are based on tensor networks. The first approach constructs a tensor network corresponding to the product of unitary matrices in the quantum circuit and contracts the network. The second approach constructs an approximate quantum state using a tensor network and approximately time-evolves it according to the quantum circuit. As shown in Ref.~\cite{qntm_rcs}, achieving fidelities and runtimes comparable to the quantum computer used in this demonstration on circuits of the structure we consider appears to be well beyond state-or-the-art implementations of known algorithms for exact tensor network contraction.  Moreover, one can frustrate approximate tensor-network methods considerably more deeply by slight modifications to the single-qubit gate set of our circuits, without impacting the verification time. A similar conclusion was reached in Ref.~\cite{morvan2023phase} for random circuit sampling on a two-dimensional grid of qubits. Therefore, in this work we focus on analyzing the first approach and refer the reader to Ref.~\cite{qntm_rcs} for a detailed comparison of simulation techniques.

A quantum circuit has a corresponding unitary matrix that can be expressed as matrix multiplications of individual quantum gates. One can  write the unitary matrix of the quantum circuit as a contraction equation of individual gate unitaries and graphically represent it as a tensor network. If we fix the values of the indices of all input qubits to be zero and output qubits to the bits in bitstring $x$, %
respectively, then contracting the tensor network gives the transition amplitude $\langle x \vert U \vert 0\rangle$. The order in which indices are contracted can change the computational cost of performing the contraction dramatically, since different contraction orders lead to different intermediate tensor sizes.

\subsubsection{Index slicing}
Even with an optimized contraction path, the dimension of the largest intermediate tensor may be so large that it cannot be fit into a single graphics processing unit (GPU). As a result, Ref.~\cite{chen2018classical} introduced index slicing, where tensor network indices are removed from the computation by fixing them to certain values. A sliced tensor network can then be contracted with lower memory overhead. This process is repeated for all possible values of the removed indices. For example, if two indices are sliced, we construct four tensor networks that correspond to those indices fixed to $\{0,0\},\{0,1\},\{1,0\},\{1,1\}$. Each tensor network is called a slice and is contracted before the scalar values are aggregated to give the final result of contraction of the original tensor network. If only a fraction $\Phi_{\A}$ of slices are contracted, the simulation fidelity is approximately $\Phi_{\A}$ \cite{markov2018quantum}.

\subsubsection{Trade-off between compute and memory operations}

The joint optimization of the contraction order and the choice of which indices to slice have a significant impact on the total floating-point operations required, the numerical efficiency, and the memory requirement of performing a tensor network contraction~\cite{Liu2021,Chen2023,morvan2023phase}. Moreover, optimizing the tensor network contraction order generally presents a trade-off between these metrics. Optimization generally aims to minimize the number of FLOPs required, but contraction schemes with low FLOP count may be highly inefficient. For example, these schemes might involve contractions between a large tensor and a very small tensor. Such a contraction has low ratio of compute (FLOP count) to data movement (limited by memory bandwidth), which is referred to as arithmetic intensity. As a result, the optimization is usually performed with respect to some combined objective function that balances FLOP count with data movement.

\subsubsection{Simulation algorithm in our experiment}\label{sec:simulation_algorithm_details}

We now present the techniques used to estimate the circuit simulation cost in Sec.~\ref{sec:protocol_circuit_selection} and to perform verification in Sec.~\ref{sec:protocol_verification}. We use the tensor-network optimizer package ``CoTenGra''~\cite{gray2021hyper} to determine the contraction scheme. We refer the reader to Ref.~\cite{qntm_rcs} for a detailed comparison of various methods, which shows that CoTenGra is one of the best tensor-network optimization packages among those that are suitable for single-amplitude contractions. We note that since we are obtaining one amplitude or one sample per circuit, intermediate tensor caching and reuse does not improve performance. 

Within CoTenGra, we perform contraction order optimization using ``KaHyPar'' as the primary method. We then perform slicing by interleaving slicing and subtree reconfiguration as well as simulated annealing. This process ensures that the largest intermediate tensor has $28$ indices, to limit the amount of overall memory demanded to store intermediate states of the contractions.

Each contraction scheme has a certain number of slices with a certain FLOP count. Given a GPU with a known computational power (FLOPs per second, or FLOPS), the time it takes for a GPU to contract a slice is at least the FLOP count per slice divided by the theoretical peak performance of the GPU. We denote this ratio the theoretical contraction time. In practice, the actual time it takes to perform contraction on a GPU differs from the estimate we obtain from the optimized contraction order; the ratio of theoretical contraction time and the actual contraction time is referred to as \textit{efficiency}. We observe that the efficiency of contraction also depends on how much we account for memory operations in the optimization of contraction order. In CoTenGra, the memory-weighted cost is referred to as ``combo-$\alpha$,'' where $\alpha$ is the weighing factor for memory operations. In the limit where we do not take memory operations into account ($\alpha \to 0$), the contraction order is optimized to minimize only the FLOP count. We observe that higher $\alpha$ increases the efficiency at the expense of a potentially higher FLOP count.

To find the best contraction scheme that results in the shortest actual time on GPUs, we optimize the contraction order at different $\alpha$ and measure the actual contraction time and efficiency on different GPUs. As an example, the $\alpha$-dependence on V100 GPUs is shown in Fig.~\ref{fig:contraction-vary-alpha}. The best contraction orders for different GPUs are different. Overall, we produce a total of $\approx 750$k contraction schemes at different $\alpha$ values. For each GPU, we pick the contraction scheme with the lowest time-to-solution for actual verification. Different GPUs have different optimal contraction schemes with different FLOP counts and efficiencies, which we report in \Cref{tab:gpu}. In practice, however, the FLOP counts can be higher than the theoretical estimate due to some additional tensor manipulations. For NVIDIA A100, the total FLOPs for a single circuit is $36.6\times10^{18}$ when benchmarked using NVIDIA Nsight Compute, which corresponds to an efficiency of $63\%$, slightly higher than the $60\%$ efficiency estimated using CoTenGra FLOP count.

\begin{figure}
    \centering
    \includegraphics[width=0.55\textwidth]{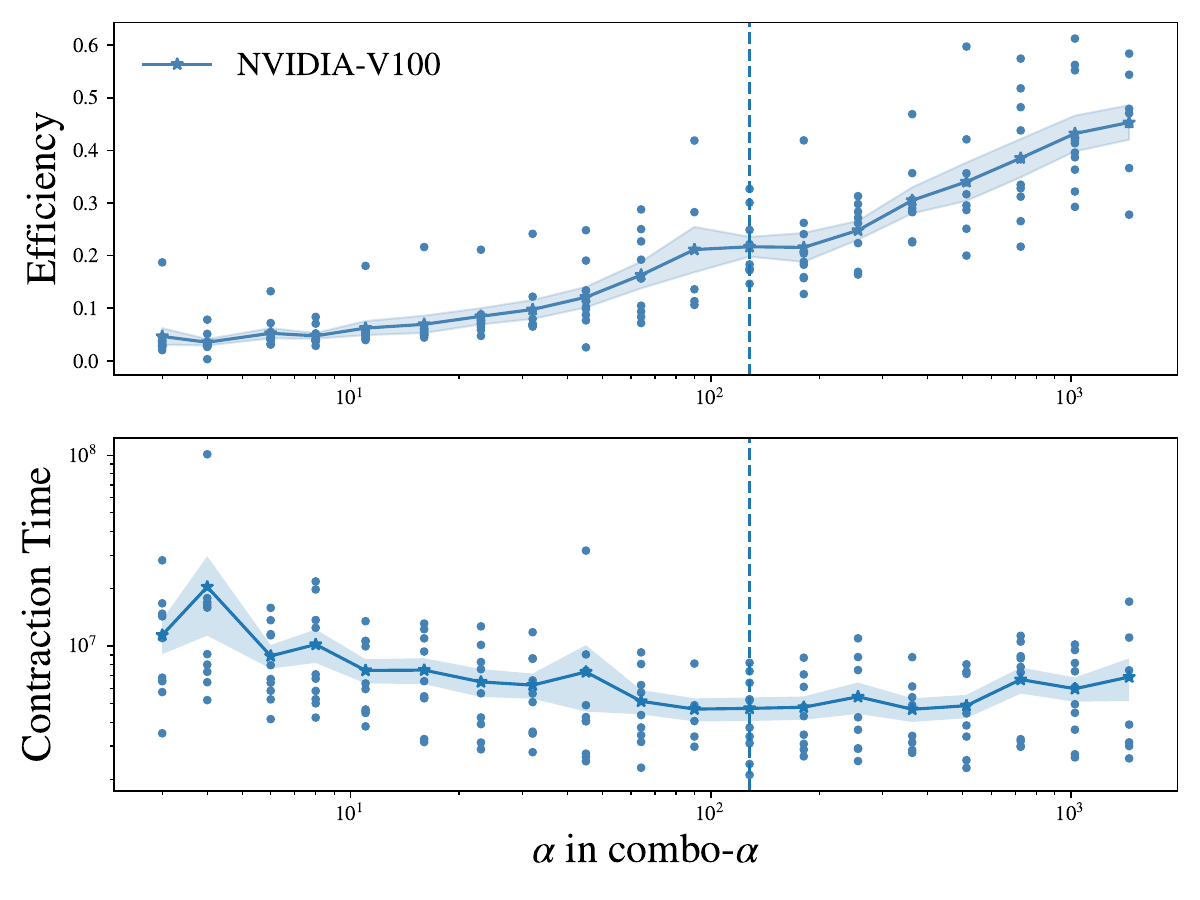}
    \caption{Contraction performance, in terms of efficiency (top) and time (bottom), of ten depth-10 circuits on NVIDIA-V100, as a function of $\alpha$, the weight parameter in the objective of contraction order. The blue shaded region around blue line (mean) represents standard error. The vertical dashed line represents $\alpha=128$, chosen for contraction on Summit. }
    \label{fig:contraction-vary-alpha}
\end{figure}

\begin{table}
    \centering
        \begin{tabular}{c|c|c|c|c}
        \toprule
         GPU type & \centering TFLOPS per GPU (FP32) & Time to simulate one circuit ($10^6$ seconds) & Flop count per circuit ($10^{18}$) & Efficiency\\
         \hline
          AMD MI250X & 53 & 3.5 & 90 & 49\% \\
         \hline
         NVIDIA V100 & 14 & 3.9 & 35 & 64\% \\
         \hline
          NVIDIA A100 & 19.5 & 3.0 & 35 & 60\% \\
     \bottomrule
\end{tabular}
    \caption{GPU performance of AMD MI250X (Frontier \cite{frontier}), NVIDIA V100 (Summit \cite{summit}), and NVIDIA A100 (Perlmutter \cite{perlmutter} and Polaris \cite{polaris}). Each AMD MI250X GPU has 2 Graphics Compute Dies (GCDs). Each GCD has a theoretical peak performance of 26.5 teraFLOPS (TFLOPS) for double and single precision.
    }
    \label{tab:gpu}
\end{table}

\subsubsection{Heterogeneous high-performance computing platforms of our experiment} \label{sec:verification_compute_details}

We use four U.S. Department of Energy supercomputers to perform quantum circuit simulation for verification: Frontier and Summit at the Oak Ridge Leadership Computing Facility, Perlmutter at the National Energy Research Scientific Computing Center, and Polaris at the Argonne Leadership Computing Facility. The technical specifications needed to calculate the single-precision theoretical peak performance (number of  FLOPS) of the supercomputers are listed in \Cref{tab:supercomputers}, and we estimate the time to simulate one circuit as well.

The actual time-to-solution on any given supercomputer depends on numerous factors in addition to the theoretical peak performance. Some computational tasks may be more compute-heavy and are bottlenecked by the theoretical peak performance, whereas others may be more memory operations-heavy and are bottlenecked by the memory bandwidth. The device architecture also affects the ability for the tensor network contraction to utilize caching, reduced or mixed precision, and other techniques. The software stack may also affect the ability to use just-in-time compilation, efficiency of high-dimensional tensor manipulation, specialized math libraries, and proprietary quantum simulation software such as cuQuantum by NVIDIA. For example, while we use CoTenGra for contraction scheme optimization, we use cuQuantum to perform actual contractions on NVIDIA GPU supercomputers and CoTenGra with the JAX backend on Frontier. As a result, the only reliable way to determine the cost of a computational task is to profile the actual time-to-solution on the device instead of relying on metrics such as FLOP count or total data movement cost.

Consider the computational cost of the actual quantum circuits executed on the quantum device to generate certified randomness. In Sec.~\ref{sec:protocol_circuit_selection} we explain how we chose the circuits. For a unified and simple understanding of the computational budget and computational power, we convert the overall supercomputing budget on various supercomputers into Frontier node-hours. The lowest time-to-solution contraction scheme on Frontier, after balancing FLOP count and memory operations, would take 3.5 million seconds to be simulated on a single MI250X GPU. Therefore, simulating a single circuit takes one Frontier node $3.5\times 10^6 / (4 \times 3600)\approx 243$ hours. One Frontier node-hour thus corresponds to $1/243$ of the cost of a single circuit.

To convert our budget on other supercomputers into Frontier node-hours, we find the lowest time-to-solution contraction schemes on these supercomputers. We multiply $243$ by the number of simulated circuits (1522) to get the effective Frontier node-hours we have for validation. We estimate a total of approximately 370k effective Frontier node-hours across all of the supercomputers. This corresponds to $370000\time 4\times 3600\times53\times10^{12}=2.8\times10^{23}$ theoretical floating-point operations on Frontier. The contraction scheme on Frontier corresponds to a FLOP count of $\budget\approx 10^{20}$ per circuit and can be executed with an efficiency of $\ceff\approx 50\%$. We remark that a loose upper bound on maximum achievable Frontier efficiency on a realistic application is provided by the LINPACK performance benchmark used by the TOP500 list \cite{top500}, on which Frontier achieved only $71.1\%$ efficiency. We further note that efficiency at full-machine scale that we measure and report in Table~\ref{tab:verification_scaling} is slightly lower than single-node efficiency. However, since the efficiency decay is small and comparable across all supercomputers utilized in this work, we use single-node benchmarking data for the purpose of node-hour conversion and total budget estimation.

We emphasize that we do not perform exactly $2.8\times 10^{23}$ floating point operations. One obvious reason is that the efficiency of Frontier is not 100\%. A different reason is that any given contraction scheme on a different supercomputer will have a different FLOP count and efficiency. Even the same contraction scheme has different efficiencies on different supercomputers. This is why when computing effective Frontier node-hours, we use only the time-to-solution, and we measure our budget using effective Frontier FLOP count and use the efficiency on Frontier. Similarly, we report the cost of simulating quantum circuits using the FLOP count and efficiency on Frontier.

\begin{table}[]
    \centering
        \begin{tabular}{c|c|c|c|c}
        \toprule
         & \multicolumn{1}{m{0.8in}|}{Compute nodes} & \multicolumn{1}{m{0.8in}|}{\centering GPU type} & \multicolumn{1}{m{0.8in}|}{GPUs per node} & \multicolumn{1}{m{1 in}}{\centering PFLOPS (FP32)}  \\
         \hline
         Frontier & 9408 & AMD MI250X & 4 & 1994 \\ %
         \hline
         Summit & 4608 & NVIDIA V100 & 6 & 387 \\ %
         \hline
         Perlmutter & 1536 & NVIDIA A100 & 4 & 120 \\ %
         \hline
         Polaris & 560 & NVIDIA A100 & 4 & 44 \\ %
     \bottomrule
\end{tabular}
    \caption{Supercomputer technical specifications for Frontier \cite{frontier}, Summit \cite{summit}, Perlmutter \cite{perlmutter}, and Polaris \cite{polaris}. All FLOPS listed are for single precision (FP32). Total performance in petaFLOPS (PFLOPS) is the theoretical peak performance obtained by multiplying the number of GPUs and the single GPU performance, available in \Cref{tab:gpu}.
    }
    \label{tab:supercomputers}
\end{table}

\subsection{Selection of Experimental Parameters}\label{sec:optimization}

The success of our protocol, outlined in the main text and in the overview in Sec.~\ref{sec:protocol-summary}, depends on the careful choice of challenge circuits: classical simulation of these circuits must be difficult enough to preclude fast classical simulation by the adversary and easy enough such that client-side verification is possible with a supercomputer. Since in practice we have a limited verification budget, the difficulty of the circuits is inversely proportional to the number of circuits that we can verify, $m=|\V|$. If the difficulty is too low, then the adversary can simulate the circuits with high fidelity close to the experimental fidelity, making it impossible to distinguish between classical and quantum samples even for large $m$. On the other hand, if the difficulty is too high, then $m$ is too small, and the $\XEB$ score has significant fluctuations from experiment to experiment, and we cannot confidently distinguish between classical and quantum samples either in a given experiment. Therefore, the choice of $m$ is crucial to the success of the protocol.

When planning the experiment, we first specify a target entropy (or $Q$) and the soundness parameter $\epssou$. If the number of verification circuits $m$, the cost of simulating one circuit $\budget$, and the classical computational power of the adversary $\A$ are given, we can determine the $\XEB$ score $\chi$ needed in order for the client to guarantee the target $Q$ and $\epssou$, which allows us to compute the honest case failure probability $\pfail$. However, even in the setting of fixed total verification budget $\totalbudget$, we have the freedom of choosing $m$ as long as we adjust $\budget$ appropriately as well, subject to $m = \lfloor \totalbudget/\budget\rfloor$.

The trade-off between adversary simulation hardness and statistical uncertainties in the \XEB{} score leads to an existence of an optimal $m$ such that $\pfail$ is minimal. Therefore, we optimize $m$ to minimize $\pfail$. Mathematically, this amounts to the following:
\begin{align*}
\text{minimize} & \quad  \pfail(m) = \Pr\left(\XEB_{m, \phi} < \chi \right), \\
\text{ such that } & \quad  \epsadv(Q,\chi,m) = \epssou.
\end{align*}

The failure probability $\pfail(m)$ is a function of $m$ in a sense that $\Phi_{\A}$ in Eq.~\ref{eq:phi_A} depends on $\budget=\totalbudget/m$ and the honest server and adversarial CDFs in Eq.~\ref{eq:cdf_honest_server} and \ref{eq:cdf_adversary} explicitly depend on $m$ as well. The $\Phi_{\A}$ dependence favors small $m$ to decrease the adversarial classical fidelity. The explicit CDF dependence favors large $m$ to reduce the statistical fluctuations. %
With the fixed total verification budget we possessed, with consideration of our circuit hardness and other protocol parameters, we verified $m = \HowManySamplesVerified$ circuits.

\subsection{Selection of Challenge Circuits}\label{sec:protocol_circuit_selection}

When choosing a circuit family for the challenge circuits, we must balance two considerations. On one hand, the circuit should be deep enough to make simulation difficult. On the other hand, the circuit should be shallow enough to enable high-fidelity execution on the quantum computer. Ref.~\cite{qntm_rcs} identifies a circuit family that achieves a high simulation cost while achieving very high fidelity on the H2-1 quantum processor. In this work we use the same circuit family. We refer the interested reader to Ref.~\cite{qntm_rcs} for a detailed discussion of the hardness of simulation of these circuits.

We specify the challenge circuits as follows.
An $n$-qubit random circuit of depth $d$,  $C_{n, d}$, consists of $d$ layers of entangling gates. Each entangling layer is composed of a random set of $n/2$ disjoint $\ZZgate(\pi/2)$ gates, and each entangling layer is sandwiched by layers of random $SU(2)$ gates on all $n$ qubits. $\ZZgate$ is the native two-qubit gate on Quantinuum hardware. The arrangement of entangling operations in $C_{n,d}$ is obtained by solving the edge-coloring problem over a $d$-regular graph with $n$ nodes, with each node representing a qubit. Assigning each layer to a color, the set of edges colored by a color $i \in [d]$ gives the disjoint qubit-pairings for entangling layer $i$. We refer to the arrangement of entangling pairs across all layers as a \textit{topology}.

With a fixed circuit family, we must choose the depth $d$ and a particular $d$-regular graph specifying the topology. For a fixed topology, the cost of performing exact tensor network contraction is independent of the choice of single-qubit gates. We generate numerous topologies with different $d$ and calculate their ``combo-128''  costs. We then ignore topologies that are obviously too expensive or too cheap to simulate, and we estimate contraction time for the rest on the types of GPUs we expect to verify the circuits on. The details of the methods used to estimate the cost of simulation are given in Sec.~\ref{sec:simulation_algorithm_details}. This approach allows us to identify a $d=10$ topology such that we can simulate close to $m$ circuits with our total verification budget. With $n=56$ and $d=10$, these circuits have a total of $10 \times (56/2) = 280$ entangling gates and $11 \times 56 = 616$ single-qubit gates. As discussed earlier, once the circuit is chosen based on estimated costs from preliminary optimizations, $\approx 750$k contraction schemes are then generated with different $\alpha$ to identify the scheme with the lowest time-to-solution on each supercomputer. This process lets us identify a \textit{single} topology with the desired computational hardness and contraction schemes corresponding to the GPU models we use.

We generate challenge circuits by randomizing only the single-qubit $SU(2)$ gates over this chosen topology. In practice, this is achieved with finite randomness by discretizing the continuous group $SU(2)$. Crucially, the seed used to randomize single-qubit gates is kept secret.

\subsection{Client-Server Interaction}
\label{sec:protocol_interaction}

Having fixed a topology, the client maintains a reservoir of circuits, where the single-qubit gates have been randomized using seeds kept secret from the server. These circuits are generated by using the commonly used library ``Pytket'' and are subsequently compiled to a H2-specific file type that spells out ion transport and gating operations necessary to implement the circuit in the H2-1 machine.

In total, we submit \HowManyCircuitsSubmitted{} circuits in \HowManyBatchesSubmitted{} batches and receive \HowManySamplesTotal{} valid samples. The timeline of data collection is shown in Fig.~\ref{fig:time-statistics}. The $\Tqc$ for each circuit is taken to be the time interval recorded for that batch, divided by the size of the batch. 

\begin{figure}
    \centering
    \includegraphics[width=\textwidth]{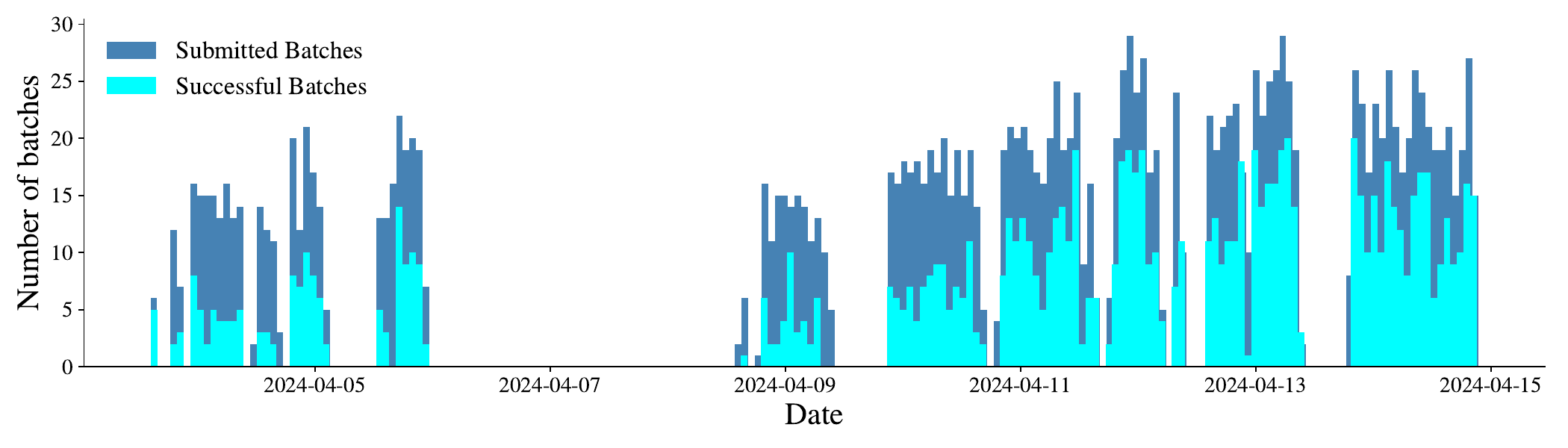}
    \caption{Timeline of data collection}
    \label{fig:time-statistics}
\end{figure}

\subsection{Verification}\label{sec:protocol_verification}

As discussed above, the final contraction schemes used by different supercomputers are obtained by choosing the lowest time-to-solution scheme from $\approx 750$k schemes obtained by using KaHyPar, subtree reconfiguration, and simulated annealing provided by CoTenGra. For Frontier, CoTenGra with the JAX backend is used to perform contraction. For other supercomputers, cuQuantum is used to perform contraction, taking the contraction scheme obtained from CoTenGra as input. The details of techniques used to perform the exact contraction are given in Sec.~\ref{sec:simulation_algorithm_details}.

\begin{table}[h!]
    \centering
        \begin{tabular}{c|c|c|c|c|c}
        \toprule
         & Single-GPU & 10 nodes & 100 nodes & 1000 nodes & full-machine \\
         \hline
         Frontier & 49\% & 45\% & 44\% & 45\% & 45\% \\
         \hline
         Summit & 64\% & 62\% & 61\% & 61\% & 59\% \\
         \hline
         Perlmutter & 60\% & 59\% & 59\% & 60\% & N/A \\
         \hline
         Polaris & 60\% & 60\% & 60\% & N/A & N/A \\
     \bottomrule
\end{tabular}
    \caption{Efficiency of tensor network contraction on different scales. We did not perform full-machine benchmarks on Perlmutter and Polaris. Additionally, Polaris has fewer than 1,000 nodes.}
    \label{tab:verification_scaling}
\end{table}

We remark that contraction of each slice is independent from each other, and the task is embarrassingly parallel. In \Cref{tab:verification_scaling} we show trivial scaling of the verification algorithm on all supercomputers. We also report the cost of contracting a single circuit in terms of the FLOP count, which is different between GPUs since the lowest time-to-solution contraction scheme is different between GPUs. Additionally, the full machines of Frontier and Summit were utilized, achieving numerical efficiencies of 45\% and 59\%, respectively. This translates to a sustained performance of 897 PFLOPS on Frontier and 228 PFLOPS Summit, which gives a total of 1.1 EFLOPS of sustained performance across both machines. The resulting time to simulate one circuit is $\HowManySecondsToVerifyFrontier$ seconds on Frontier and $\HowManySecondsToVerifySummit$ seconds on Summit.

\subsection{Randomness Extraction}\label{sec:extraction}

The numbers corresponding to our experimental results are $n=56, M=30010, T_{\rm tot} = 64641~\text{seconds}, \XEBTest = 0.32, \text{ and } m = 1522$. The experiment passes the randomness extraction protocol with abortion thresholds $\chi=0.30$ and $\Tqcthreshold=2.2$ s. For this protocol we first estimate the minimum number of $\Qmin$ against an adversary $\A$ for various values of $\epssou$ using Lemma~\ref{lem:q_bound} by performing binary search on $Q$ such that $\epsadv(Q,\chi)=\epssou$. Since Lemma~\ref{lem:q_bound} holds for any $\delta\geq 0$, we could optimize $\delta$ to minimize $\epsadv$. In practice, we simply solve for $\delta$ such that $\varepsilon_1=\epsadv/2$.

We then determine the smooth min-entropy certified by our soundness result (Theorem~\ref{thm:entropy_bound} for $\Qmin$ samples, with smoothing parameter $\epss = \epssou/4$). Table I in the main text reports the smooth min-entropy rate $h=H^{\epss}_{\rm min}/(M \cdot n)$. For example, at $\epssou=10^{-6}$ and $\A=4\times$Frontier, we have $h=0.04$.

Having determined the certified entropy present in the $M$ samples, %
we can pass the raw bits into a randomness extractor. The length of the extractor output is given by \Cref{cor:soundness}. We feed a total of $56 \cdot M$ bits into a seeded randomness extractor to obtain $\ell$ bits, which are now expected to be close to the uniform distribution. In this work we use the Toeplitz extractor, which is known to be a quantum-proof strong extractor \cite{konig2005power}. A property of a strong extractor is that the extractor output is independent of the seed used. If the client uses randomness privately, the net output is the concatenation of the extraction seed and the output of the extractor. We use the implementation of the Toeplitz extractor from the open-source package ``Cryptomite'' \cite{foreman2024cryptomite}.

\section{Details on Outlook for Future Experiments}
\label{sec:future}

We would like to understand how improvements in the quantum computer fidelity, average time per sample, and verification budget improve the protocol performance. We are interested in the change in the resulting smooth-min-entropy as well as achievable security parameters. For the protocol to be economically viable, only a very small fraction of samples can be verified ($m/M \ll 1$). In the limit of infinite $M$ and finite $m$, we can lower bound the fraction $Q/M$ of quantum samples $R_\text{Q}$ instead of the absolute number of quantum samples, and the normalized entropy is $h=R_\text{Q}\times (n-1)/n$. The average classical simulation fidelity is
\begin{equation}
    \langle \phi_{\A}\rangle=\min \left(1,\frac{\ceff\cdot \A \cdot \langle \Tqcthreshold \rangle}{(1-R_\text{Q})\cdot\budget}\right)=\min \left(1,\frac{\ceff\cdot \A \cdot \langle \Tqcthreshold \rangle}{(1-R_\text{Q})\cdot\totalbudget/m}\right)\label{eq:ave_phi_a}.
\end{equation}
In the limit of infinite $M$ and finite $m$, there is no uncertainty in the fraction $R$ of Porter--Thomas bitstrings (as can be seen in Eq.~\ref{eq:L_C_max} since $\Exp[L_\text{C}]\rightarrow\infty$), and we have $R=R_\text{Q}+(1-R_\text{Q})\cdot\langle \phi_{\A}\rangle$. Therefore, the number of Porter--Thomas bitstrings in the verification set follows the binomial distribution, and we have
\begin{align}
\Pr\left(\XEBTest \leq \chi\right)
&= \sum_{l=0}^m \Pr\left(\XEBTest \leq \chi |l  \text{ PT bitstrings in }\V\right)\cdot \Pr \left(l \text{ PT bitstrings in }\V \right) \\
&=\sum_{l=0}^{m} \tilde{\Gamma}(m+l, m \cdot (\chi + 1)) \cdot \left[\binom{m}{l}R^l(1-R)^{m-l}\right],
\end{align}
which is the same as the CDF for a finite-fidelity honest server with fidelity $R$.

Now, we would like to understand the effect of improving the quantum computer fidelity $\phi$, decreasing the threshold of average time per sample $\Tqcthreshold$, increasing the verification budget $\budget$, and decreasing the adversary computational power $\A$. Further, as  discussed in Sec.~\ref{sec:bound-on-q}, $\Tqcthreshold$, $\budget$, and $\A$ can be treated interchangeably since they only affect the CDF calculations by entering $\Phi_{\A}$ in Eq.~\ref{eq:phi_A}. For each set of parameters, we can similarly optimize $m$ following the procedure in Sec.~\ref{sec:optimization}. This allows us to map out the landscape of achievable performance in Fig.~3 of the main text.

\section{Table of Variables}

\begin{table*}[ht]
\begin{ruledtabular}
\begin{tabular}{cl}
Label & Meaning \\[.1em]  \hline \\[-.6em]
$n$ & Number of qubits \\[.1em]
$\budget$ & Cost of simulating challenge circuits \\[.1em]
$\totalbudget$ & Total classical computational budget for verification \\[.1em]
$M$ & Number of successful samples  \\[.1em]
$r$ & Length of the circuit generation seed \\[.1em]
$\ExtSeed$ & Seed for the randomness extractor \\
$x_i$ & Bitstring for the $i$th circuit \\[.1em]
$X^M$ & Client's classical register of the received $M$ samples \\[.1em]
$K$ & Client's output register \\[.1em]
$b$ & Number of stitched circuits per job, where each stitched circuits is composed of 2 circuits \\[.1em]
$\batchcutoff$ & Cutoff time for a single batch of $2b$ circuits \\[.1em]
$\Ttot$ & Total time for all successful batches \\[.1em]
$\Tqc$ & Average response time per successful quantum sample \\[.1em]
$\Tqcthreshold$ & Target for average time per successful quantum sample for protocol abort \\[.1em]
$\Tthreshold$ & $M \cdot \Tqcthreshold$ \\[.1em]
$\V$ & The set of indices for circuits used for verification \\[.1em]
$m$ & The size of the verification set \\[.1em]
$\XEBTest$ & The \XEB{} score of the verification set \\[.1em]
$\chi$ & \XEB{} score threshold \\[.1em]
$m$ & Number of samples used in \XEB{} \\[.1em]
$\phi$ & Expected fidelity of H2-1 on challenge circuits \\[.1em]
$\pfail$ & Abort probability for an honest server \\[.1em]
$\epssou$ & Soundness parameter \\[.1em]
$\epsna$ & Probability of the protocol not aborting when interacting with an adversary \\[.1em]
$\epss$ & Min-entropy smoothing parameter \\[.1em]
$U$ & Client \\[.1em]
$S$ & Server \\[.1em]
$\Xi$ & Server controller \\[.1em]
$\SQ$ & Server quantum computer \\[.1em]
$\SC$ & Server classical computer \\[.1em]
$\snapshot$ & Initial snapshot of all classical information of the client and the server \\[.1em]
$\A$ & Adversary classical computational power in FLOPS \\[.1em]
$\ceff$ & Numerical efficiency of tensor network contraction \\[.1em]
$Q$ & Number of quantum samples from the adversary \\[.1em]
$\Qmin$ & Minimum number of quantum samples the adversary must generate to pass the \XEB{} test with probability at least $\epsna$ \\[.1em]
$M'$ & Number of amplitudes to compute for frugal rejection sampling by the adversary \\[.1em]
$\phi_{\A}^{(i)}$ & Adversary's classical simulation fidelity of the $i$th accepted circuit \\[.1em]
$\Phi_{\A}$ & Sum of classical simulation fidelity for an adversary with computational power $\A$ \\[.1em]
$\langle\phi_{\A}\rangle$ & Average classical simulation fidelity for an adversary with computational power $\A$ in the limit of infinitely many samples \\[.1em]
$\U$ & Uniform distribution \\[.1em]
$\Haar$ & Ensemble of Haar-random states \\[.1em]
\end{tabular}
\end{ruledtabular}
\caption{Summary of experimental parameters used in this work.}

\label{tab:estimation}
\end{table*}

\bibliographystyle{apsrev4-1}

\bibliography{citations}